\documentclass[11pt,a4paper,reqnno]{amsart}
\usepackage{mathtools,bm,a4wide,amssymb,amsmath,amsthm,graphicx,fancyhdr,cases,float,accents}
\usepackage[center]{caption}
\usepackage[toc,title]{appendix}
\usepackage[usenames,dvipsnames]{color}

\newtheorem{thm}{Theorem}[section]

\newtheorem{lem}[thm]{Lemma}
\newtheorem{prop}[thm]{Proposition}

\theoremstyle{definition}
\newtheorem{conj}[thm]{Conjecture}

\newtheorem{remark}[thm]{Remark}

\definecolor{light-gray}{gray}{0.85}

\catcode`,\active
\newcommand{\uhat}{\underaccent{\check}}

\catcode`\,12

\numberwithin{equation}{section}

\begin{document}
\title{Analytic solutions of $q$-$P(A_1)$ near its critical points}
\author{Nalini Joshi}
\address{School of Mathematics and Statistics F07, The University of Sydney, 
	NSW 2006, Australia}
	\email{nalini.joshi@sydney.edu.au}
\author{Pieter Roffelsen}
\address{School of Mathematics and Statistics F07, The University of Sydney, 
	NSW 2006, Australia}
\email{P.Roffelsen@maths.usyd.edu.au}

\begin{abstract}
For transcendental functions that solve non-linear $q$-difference equations, the best descriptions available are the ones obtained by expansion near critical points at the origin and infinity. We describe such solutions of a $q$-discrete Painlev\'e equation, with 7 parameters whose initial value space is a rational surface of type $A_1^{(1)}$. The resultant expansions are shown to approach series expansions of the classical sixth Painlev\'e equation in the continuum limit.
\end{abstract}

\maketitle

\section{Introduction}
Sakai \cite{Sakai} classified initial value spaces of Painlev\'e equations with singularities in $\mathbb P^2$. On $11$ such spaces, $q$-discrete Painlev\'e equations are realised when choosing a direction of the corresponding affine Weyl group acting on it. While many questions about solutions of the classical Painlev\'e equations have been answered, corresponding questions remain open for most of these $q$-discrete Painlev\'e equations, with the gap being widest for those equations ranked higher in the classification of Sakai \cite{Sakai}.  This paper focuses on such an example, denoted by $q$-$P(A_1)$. We describe its solutions near the origin and infinity and show that in the continuum limit these reduce to the expansions found by Jimbo \cite{Jimbo} in his groundbreaking study of the third, fifth and sixth Painlev\'e equation.

More specifically, we study analytic properties near the origin and infinity of solutions of the equation
\begin{equation}
\label{qpa1}
q\text{-}P(A_1)\quad
\begin{cases}
\displaystyle\frac{(gf-t^2)(g\overline{f}-qt^2)}{(gf-1)(g\overline{f}-1)}=\displaystyle\frac{(g-b_1  t)(g-b_2  t)(g-b_3  t)(g-b_4  t)}{(g-b_5)(g-b_6)(g-b_7)(g-b_8)},\\
\displaystyle\frac{(g\overline{f}-qt^2)(\overline{g}\overline{f}-q^2t^2)}{(g\overline{f}-1)(\overline{g}\overline{f}-1)}=\displaystyle\frac{(\overline{f}- b_1^{-1} qt)(\overline{f}- b_2^{-1}qt)(\overline{f}- b_3^{-1} qt)(\overline{f}-  b_4^{-1}qt)}{(\overline{f}-b_5^{-1})(\overline{f}-b_6^{-1})(\overline{f}-b_7^{-1})(\overline{f}-b_8^{-1})},
\end{cases}
\end{equation}
where $f=f(t)$ and $g=g(t)$ are the dependent variables, $t$ is the independent variable, we denote $\overline{f}=f(qt)$ and $\overline{g}=g(qt)$, and $b_1,\ldots,b_8\in \mathbb{C}^*$ are complex parameters satisfying the single constraint
\begin{equation}
q=\frac{b_1b_2b_3b_4}{b_5b_6b_7b_8}. \label{eqconstraint}
\end{equation}
Given any analytic $q$-periodic function $\Lambda$, i.e. $\Lambda(qt)=\Lambda(t)$, subject to a certain bound, and any nonvanishing analytic function $\phi$  satisfying $\overline{\phi}=\lambda \phi$, where $ b_1b_2b_3b_4\lambda=\Lambda^2$, we construct a meromorphic solution to $q$-$P(A_1)$ whose critical behaviour near the origin is described by
\begin{equation}
\begin{cases}
&f(t)\sim \phi(t)t\\
&g(t)\sim \Lambda(t)\phi(t)t
\end{cases}
\label{asymptotics}
\end{equation}
as $t\rightarrow 0$, on a given domain.

The family of solutions we find depends on two degrees of freedom (i.e., two arbitrary $q$-periodic functions). This richness collapses in the continuum limit $q\rightarrow 1$ to two complex constants and the resulting solutions for the sixth Painlev\'e equation \eqref{pVI} are identified as the well known solutions derived by Jimbo \cite{Jimbo}.

The theory of Painlev\'e equations goes back more than a century, to the pioneering work of Painlev\'e, Gambier and their colleagues, in their study of second-order nonlinear ordinary differential equations. Painlev\'e et al. found $6$ new ordinary differential equations, now known as the Painlev\'e equations, which define new transcendental functions called Painlev\'e transcendents. Since their discovery, the Painlev\'e equations have appeared in numerous physical applications. For an overview, see for instance Fokas et al. \cite{Fokas}.

The discrete Painlev\'e equations started appearing in the 1980s in different contexts. Br\'ezin and Kazakov \cite{Brezin} first calculated the continuum limit of what is now known as $d$-$P_{\text{I}}$ in 1990 and thereby identified it as a discrete version of the first Painlev\'e equation. This initiated an exciting new area of research, from which a surprisingly rich theory of discrete Painlev\'e equations arose. Many other discrete Painlev\'e equations were found after Grammaticos et al. \cite{grammaticossing} formulated a discrete analog of the Painlev\'e property called the singularity confinement property. In particular, Ramani, Grammaticos and Hietarinta \cite{GRH91} discovered a $q$-discrete version of the third Painlev\'e equation. Later work by Jimbo and Sakai \cite{JS96} uncovered a $q$-discrete version of the sixth Painlev\'e equation and Ramani and Grammaticos \cite{firstqp6} also deduced another $q$-discrete version of the sixth Painlev\'e equation, which turned out to be $q$-$P(A_1)$. 

In 2001, Sakai \cite{Sakai} completed the list of discrete Painlev\'e equations, whilst classifying them in terms of groups of Cremona isometries on the Picard group of certain rational surfaces. The rational surface at the top of the list is of type $A_0^{(1)}$, while $A_1^{(1)}$ is the next highest.  $q$-$P(A_1)$ is associated with the latter rational surface. We remark that in the scheme given by Sakai, $q$-$P(A_1)$ is ranked higher than $q$-$P_{\text{VI}}$ and hence higher than the continuous Painlev\'e equations.

Jimbo and Sakai \cite{JS96} were the first to study the linear problem of $q$-$P_{\text{VI}}$ from Birkhoff's analytic point of view. Moreover, Grammaticos et al's studies of singularity confinement and special solutions could be described as analytic information about solutions. Ohyama \cite{Ohyama1,Ohyama} studied analytic properties of $q$-Painlev\'e equations,  classified all analytic solutions to the equations $q$-$P_{\text{VI}}$, $q$-$P_{\text{V}}$ and $q$-$P_{\text{III}}$ around the origin and solved the associated linear connection problems.
Mano \cite{Mano}, building on Ohyama's results, derived solutions to the equation $q$-$P_{\text{VI}}$ described by a broad range of asymptotics near the origin and likewise near infinity. He solved the connection problem for these, in analogy to Jimbo's work \cite{Jimbo} on Painlev\'e VI.

The plan of this paper is as follows. We start our investigation by considering holomorphic solutions of $q$-$P(A_1)$ at the origin in Section \ref{sectionholo}. In Section \ref{sectionlo}, we identify the leading order autonomous system \eqref{autonomoufg} of $q$-$P(A_1)$ as a QRT mapping of a special type.
Using the generic solution to the autonomous leading order system, we can construct a corresponding formal series solution in $t$ and $\phi$ to the entire $q$-$P(A_1)$ equation, as shown in Section \ref{sectiongeneral}. In Section \ref{sectiontruesol} we show how substituting analytic functions for the formal variables, the formal series solution becomes a true solution to $q$-$P(A_1)$ with asymptotics as described by equations \eqref{asymptotics}.

\section{Solutions of $q$-$P(A_1)$ holomorphic at the origin}
\label{sectionholo}
In this section we study holomorphic solutions of $q$-$P(A_1)$ at the origin. These solutions play a special role in the more general solution we derive later, as they correspond to constant solutions of the leading order autonomous system \eqref{losystemeq}. Classifying the holomorphic solutions is done using the power series method and using the $q$-Briot-Bouquet Theorem \ref{genbriottheoremsev} to prove convergence.
Ohyama \cite{Ohyama1,Ohyama} classified the meromorphic solutions of the discrete Painlev\'e equations $q$-$P_{\text{VI}}$, $q$-$P_{\text{V}}$ and $q$-$P_{\text{III}}$ around the origin in this fashion. As the $q$-$P(A_1)$ case is quite similar, we keep our discussion brief. We denote
\begin{equation*}
\mathbf{b}=(b_1,b_2,b_3,b_4,b_5,b_6,b_7,b_8),
\end{equation*}
and throughout this article we write the parameter space of $q$-$P(A_1)$ as
\begin{equation*}
\mathcal{B}=\{\mathbf{b}\in\mathbb{C}^8|\text{$b_i\neq 0$ for $1\leq i\leq 8$} \}.
\end{equation*}
Let us rewrite the $q$-$P(A_1)$ equation \eqref{qpa1} as
\begin{subequations}
\label{qpa1expand}
\begin{align}
&\begin{multlined}[b][.87\textwidth]
(gf-t^2)(g\overline{f}-qt^2)(g-b_5)(g-b_6)(g-b_7)(g-b_8)=\\
(gf-1)(g\overline{f}-1)(g-b_1  t)(g-b_2  t)(g-b_3  t)(g-b_4  t),
\end{multlined}\\
&\begin{multlined}[b][.87\textwidth]
(g\overline{f}-qt^2)(\overline{g}\overline{f}-q^2t^2)(\overline{f}-b_5^{-1})(\overline{f}-b_6^{-1})(\overline{f}-b_7^{-1})(\overline{f}-b_8^{-1})=\\
(g\overline{f}-1)(\overline{g}\overline{f}-1)(\overline{f}- b_1^{-1} qt)(\overline{f}- b_2^{-1}qt)(\overline{f}- b_3^{-1} qt)(\overline{f}-  b_4^{-1}qt).
\end{multlined}
\end{align}
\end{subequations}
Suppose we have a power series solution of $q$-$P(A_1)$ around $t=0$, say
\begin{align}
f(t)&=\sum_{n=0}^{\infty}{f_nt^n}, & g(t)&=\sum_{n=0}^{\infty}{g_nt^n}.
\end{align}
Evaluating equation \eqref{qpa1expand} at $t=0$ gives
\begin{subequations}
\label{constantterms}
\begin{align}
(f_0g_0-1)^2f_0^4&=f_0^2g_0^2(f_0-b_5^{-1})(f_0-b_6^{-1})(f_0-b_7^{-1})(f_0-b_8^{-1}),\\
(f_0g_0-1)^2g_0^4&=f_0^2g_0^2(g_0-b_5)(g_0-b_6)(g_0-b_7)(g_0-b_8).
\end{align}
\end{subequations}
Equation \eqref{constantterms} has several trivial solutions, given by $(f_0,g_0)=(0,0)$ and $(f_0,g_0)=(b_i^{-1},b_i)$ for $i=5,6,7,8$. Furthermore there are generally three nontrivial solutions, given by
\begin{subequations}
\label{orderzerocoef}
\begin{align}
\left(f_0^{(0,1)},g_0^{(0,1)}\right)&=\left(\frac{b_5b_6-b_7b_8}{b_5b_6(b_7+b_8)-b_7b_8(b_5+b_6)}, \frac{b_5b_6-b_7b_8}{b_5+b_6-(b_7+b_8)}\right),\label{constant1}\\
\left(f_0^{(0,2)},g_0^{(0,2)}\right)&=\left(\frac{b_6b_7-b_8b_5}{b_6b_7(b_8+b_5)-b_8b_5(b_6+b_7)}, \frac{b_6b_7-b_8b_5}{b_6+b_7-(b_8+b_5)}\right),\label{constant2}\\
\left(f_0^{(0,3)},g_0^{(0,3)}\right)&=\left(\frac{b_5b_7-b_6b_8}{b_5b_7(b_6+b_8)-b_6b_8(b_5+b_7)}, \frac{b_5b_7-b_6b_8}{b_5+b_7-(b_6+b_8)}\right).\label{constant3}
\end{align}
\end{subequations}
If $(f_0,g_0)=(0,0)$, then there are no terms $t^n$ with $n<4$ appearing in \eqref{qpa1expand}, equating the coefficients of $t^4$ in \eqref{qpa1expand} gives
\begin{subequations}
\label{degreeoneterm}
\begin{align}
(f_1g_1-1)^2b_1b_2b_3b_4&=(g_1-b_1)(g_1-b_2)(g_1-b_3)(g_1-b_4),\\
(f_1g_1-1)^2\frac{1}{b_1b_2b_3b_4}&=(f_1-b_1^{-1})(f_1-b_2^{-1})(f_1-b_3^{-1})(f_1-b_4^{-1}).
\end{align}
\end{subequations}
Equation \eqref{degreeoneterm} has several trivial solutions, given by $(f_1,g_1)=(0,0)$ and $(f_1,g_1)=(b_i^{-1},b_i)$ for $i=1,2,3,4$. Furthermore there are generally three nontrivial solutions, given by
\begin{subequations}
\label{orderonecoef}
\begin{align}
\left(f_1^{(1,1)},g_1^{(1,1)}\right)&=\left(\frac{b_1+b_2-(b_3+b_4)}{b_1b_2-b_3b_4},\frac{b_1b_2(b_3+b_4)-b_3b_4(b_1+b_2)}{b_1b_2-b_3b_4}\right),\label{degreeone1}\\
\left(f_1^{(1,2)},g_1^{(1,2)}\right)&=\left(\frac{b_2+b_3-(b_4+b_1)}{b_2b_3-b_4b_1},\frac{b_2b_3(b_4+b_1)-b_4b_1(b_2+b_3)}{b_2b_3-b_4b_1}\right),\label{degreeone2}\\
\left(f_1^{(1,3)},g_1^{(1,3)}\right)&=\left(\frac{b_1+b_3-(b_2+b_4)}{b_1b_3-b_2b_4},\frac{b_1b_3(b_2+b_4)-b_2b_4(b_1+b_3)}{b_1b_3-b_2b_4}\right).\label{degreeone3}
\end{align}
\end{subequations}
Each of the cases in equations \eqref{orderzerocoef} and \eqref{orderonecoef} generically determines an unique converging power series solution.
\begin{prop}
\label{merosolution0}
For $k\in\{1,2,3\}$, the $q$-$P(A_1)$ equation has an unique power series solution
\begin{align}
f^{(0,k)}(t)&=\sum_{n=0}^{\infty}{f_n^{(0,k)}t^n},&&g^{(0,k)}(t)=\sum_{n=0}^{\infty}{g_n^{(0,k)}t^n},\label{eq:holo0}
\end{align}
with $f_1^{(0,k)}$ and $g_1^{(0,k)}$ as defined in equation \eqref{orderzerocoef}, given that the following conditions are satisfied for the case $k=1$, $k=2$ and $k=3$ respectively,
\begin{align*}
\frac{b_5 b_6}{b_7 b_8}\notin q^\mathbb{Z}, & & b_5+b_6\neq b_7+b_8 & &\text{and} & & b_5^{-1}+b_6^{-1}\neq b_7^{-1}+b_8^{-1},\\
\frac{b_6 b_7}{b_5 b_8}\notin q^\mathbb{Z}, & & b_6+b_7\neq b_5+b_8 & &\text{and} & &  b_6^{-1}+b_7^{-1}\neq b_5^{-1}+b_8^{-1},\\
\frac{b_5 b_7}{b_6 b_8}\notin q^\mathbb{Z}, & & b_5+b_7\neq b_6+b_8 & &\text{and}& & b_5^{-1}+b_7^{-1}\neq b_6^{-1}+b_8^{-1}.
\end{align*}
Furthermore, for each of these power series solutions, if the first condition still holds when $q^\mathbb{Z}$ is replaced by its closure $\overline{q^\mathbb{Z}}$, then the series have a positive radius of convergence.
\end{prop}
\begin{proof}
We discuss the case $k=1$. Note that we can rewrite $q$-$P(A_1)$ as
\begin{align}
\overline{f}=H_1(t,f,g),&& \overline{g}=H_2(t,f,g), \label{Hqpa1}
\end{align}
for some rational functions $H_1$ and $H_2$.\\
We apply the $q$-Briot-Bouquet theorem \ref{genbriottheoremsev} with $m=1$ and $n=2$ to this system, where $y_1=f$, $y_2=g$ and 
\begin{equation*}
\mathbf{Y}=(f_0,g_0)=\left(f_0^{(1)},g_0^{(1)}\right).
\end{equation*}
It is not hard to see that $H(t,f,g)$ is holomorphic at $(t,f,g)=(0,f_0,g_0)$ and $H(0,f_0,g_0)=(f_0,g_0)$, as this is essentially the calculation done to obtain the case \eqref{constant1}. We hence calculate
\begin{equation*}
\begin{pmatrix}
\frac{\partial H_1}{\partial f}(0,\mathbf{Y}) & \frac{\partial H_1}{\partial g}(0,\mathbf{Y})\\
\frac{\partial H_2}{\partial f}(0,\mathbf{Y}) & \frac{\partial H_2}{\partial g}(0,\mathbf{Y})
\end{pmatrix}=\begin{pmatrix}
-1 & -\frac{(b_5+b_6-b_7-b_8)(b_5 b_6+b_7 b_8)}{(b_5 b_6(b_7+b_8)-b_7b_8(b_5+b_6))^2}\\
\frac{(b_5b_6+b_7b_8)(b_5 b_6(b_7+b_8)-b_7b_8(b_5+b_6))^2}{b_5b_6b_7b_8(b_5+b_6-b_7-b_8)^2} & \frac{(b_5b_6+b_7b_8)^2}{b_5b_6b_7b_8}-1
\end{pmatrix}.
\end{equation*}
The eigenvalues of this matrix are equal to $\frac{b_5 b_6}{b_7 b_8}$ and $\frac{b_7 b_8}{b_5 b_6}$. Since $\frac{b_5 b_6}{b_7 b_8}\neq q^n$ for any $n\in\mathbb{Z}^*$, we can apply the $q$-Briot-Bouquet theorem \ref{genbriottheoremsev} to obtain the desired results.
\end{proof}
\begin{prop}
\label{merosolution1}
For $k\in\{1,2,3\}$, the $q$-$P(A_1)$ equation has an unique power series solution
\begin{align}
f^{(1,k)}(t)&=\sum_{n=1}^{\infty}{f_n^{(1,k)}t^n},&&g^{(1,k)}(t)=\sum_{n=1}^{\infty}{g_n^{(1,k)}t^n},\label{eq:holo1}
\end{align}
with $f_1^{(1,k)}$ and $g_1^{(1,k)}$ as defined in equation \eqref{orderonecoef}, given that the following conditions are satisfied for the case $k=1$, $k=2$ and $k=3$ respectively,
\begin{align*}
\frac{b_1 b_2}{b_3 b_4}\notin q^\mathbb{Z}, & & b_1+b_2\neq b_3+b_4 & &\text{and} & & b_1^{-1}+b_2^{-1}\neq b_3^{-1}+b_4^{-1},\\
\frac{b_2 b_3}{b_1 b_4}\notin q^\mathbb{Z}, & & b_2+b_3\neq b_1+b_4 & &\text{and} & &  b_2^{-1}+b_3^{-1}\neq b_1^{-1}+b_4^{-1},\\
\frac{b_1 b_3}{b_2 b_4}\notin q^\mathbb{Z}, & & b_1+b_3\neq b_2+b_4 & &\text{and}& & b_1^{-1}+b_3^{-1}\neq b_2^{-1}+b_4^{-1}.
\end{align*}
Furthermore, for each of these power series solutions, if the first condition still holds when $q^\mathbb{Z}$ is replaced by its closure $\overline{q^\mathbb{Z}}$, then the series have a positive radius of convergence.
\end{prop}
\begin{proof}
Note we can apply the $q$-Briot-Bouquet Theorem \ref{genbriottheoremsev} as done in the proof of Proposition \ref{merosolution0}.
However, for a more elegant proof, we make use of one of the many symmetries of $q$-$P(A_1)$. Indeed applying the B\"acklund transformation $\mathcal{T}_1$, defined in \eqref{backlundt}, to each of the solutions defined in Proposition \ref{merosolution0}, gives the desired results directly.
\end{proof}
By Remark \ref{remarkincorporatepar}, the solutions defined in Propositions \ref{merosolution0} and \ref{merosolution1} are also analytic in the parameters $\mathbf{b}$.
Furthermore meromorphic solutions at $t=\infty$ can be obtained by application of the B\"acklund transformation $\mathcal{T}_4$ \eqref{backlundt}
to each of these.

\section{The leading order autonomous system}
\label{sectionlo}
For complex functions $f$ and $g$ we write $f(t)\asymp g(t)$ as $t\rightarrow t_0$ if and only if $f(t)=\mathcal{O}(g(t))$ and $g(t)=\mathcal{O}(f(t))$ as $t\rightarrow t_0$. Note that the holomorphic solutions $(f,g)$ defined in Propositions \ref{merosolution0} and \ref{merosolution1} satisfy respectively $f,g\asymp 1$ and $f,g\asymp t$ as $t\rightarrow 0$.
We therefore consider, on a formal level, any of the following  $25$ combinations of asymptotic relations as $t\rightarrow 0$, for a solution $(f,g)$ of $q$-$P(A_1)$,
\begin{align*}
f&\ll t, &  f&\asymp t,& t&\ll f\ll 1,& f&\asymp 1 &\text{or}& &f&\gg 1;& &and\\
g&\ll t, &  g&\asymp t,& t&\ll g\ll 1,& g&\asymp 1 &\text{or}& &g&\gg 1. &&
\end{align*}
Using B\"acklund transformations $\mathcal{T}_1$ and $\mathcal{T}_3$ \eqref{backlundt}, we can reduce the number of individual cases to be studied to $9$. 
By some laborious comparison of dominant and subdominant terms in equations \eqref{qpa1expand}, it is possible to show that, for generic parameter values, the only $3$ consistent combinations are
\begin{align}
f,g\asymp t,&& t\ll f,g\ll 1,&& f,g\asymp 1. \label{eq:consistent}
\end{align}
Furthermore, there are $6$ combinations which are only conditionally consistent, given by
\begin{subequations}
	\label{eq:conditionalcomb}
\begin{align}
f,g&\ll t, & f\ll t &\text{ and } g\asymp t, & f\asymp t &\text{ and } g\ll t,\\
f,g&\gg 1, & f\gg 1 &\text{ and } g\asymp 1, & f\asymp 1 &\text{ and } g\gg 1.
\end{align}
\end{subequations}
For example, $f,g\ll t$ is only consistent if
\begin{align}\label{eq:condfgsmallert}
b_1+b_2+b_3+b_4=0&& \text{and} && b_1^{-1}+b_2^{-1}+b_3^{-1}+b_4^{-1}=0,
\end{align}
and $f\ll t$ with $g\asymp t$ is only consistent if
\begin{align}\label{eq:condfsmallertgt}
b_1+b_2=b_3+b_4,&& b_1+b_3=b_2+b_4 && \text{or} && b_1+b_4=b_2+b_3.
\end{align}
We give explicit examples of such cases in Section \ref{section:6special}.
The interested reader can find the conditions, of the other conditionally consistent combinations, using B\"acklund transformations $\mathcal{T}_1$ and $\mathcal{T}_3$ \eqref{backlundt}. The remaining combinations are inconsistent for all parameter values $\mathbf{b}\in\mathcal{B}$. From these considerations it easily follows that Propositions \ref{merosolution0} and \ref{merosolution1} list all meromorphic solutions of $q$-$P(A_1)$ for generic parameter values.

Let us focus on the case $t\ll f,g\ll 1$ in \eqref{eq:consistent}. We put $f=tf_1$, and $g=tg_1$, then $1\ll f_1,g_1\ll t^{-1}$ as $t\rightarrow 0$, and by substitution into equations \eqref{qpa1expand}, we obtain
\begin{subequations}
\label{autonomoufg}
\begin{align}
(g_1f_1-1)(g_1\overline{f}_1-1)&\sim \left(b_1^{-1}g_1-1\right)\left(b_2^{-1}g_1-1\right)\left(b_3^{-1}g_1-1\right)\left(b_4^{-1}g_1-1\right),\\
(g_1\overline{f}_1-1)(\overline{g}_1\overline{f}_1-1)&\sim \left(b_1\overline{f}_1-1\right)\left(b_2\overline{f}_1-1\right)\left(b_3\overline{f}_1-1\right)\left(b_4\overline{f}_1-1\right),
\end{align}
\end{subequations}
as $t\rightarrow 0$.\\
So up to leading order $f_1$ and $g_1$ satisfy an autonomous equation.

\subsection{Derivation generic solution of leading order system}
\label{subsectionlo}
Inspired by equations \eqref{autonomoufg}, we study the following autonomous system of equations,
\begin{subequations}
\label{losystemeq}
\begin{align}
(GF-1)(G\overline{F}-1)&= \left(b_1^{-1}G-1\right)\left(b_2^{-1}G-1\right)\left(b_3^{-1}G-1\right)\left(b_4^{-1}G-1\right),\\
(G\overline{F}-1)(\overline{G}\overline{F}-1)&= \left(b_1\overline{F}-1\right)\left(b_2\overline{F}-1\right)\left(b_3\overline{F}-1\right)\left(b_4\overline{F}-1\right),
\end{align}
\end{subequations}
which we identify as a QRT mapping \eqref{QRT} with
\begin{align*}
A_0=\begin{pmatrix}
0 & 0 & 1\\
0 & S_2^- & -S_1^-\\
S_4^- & -S_3^- & 0
\end{pmatrix}, & &
 A_1=\begin{pmatrix}
0 & 0 & 0\\
0 & 1 & 0\\
0 & 0 & -1
\end{pmatrix},
\end{align*}
where $S_i^\pm$ denotes the $i$th degree elementary symmetric polynomial in $b_1^{\pm 1}$, $b_2^{\pm 1}$, $b_3^{\pm 1}$ and $b_4^{\pm 1}$, that is,
\begin{equation*}
(z-b_1^{\pm 1})(z-b_2^{\pm 1})(z-b_3^{\pm 1})(z-b_4^{\pm 1})=z^4-S_1^\pm z^3+S_2^\pm z^2-S_3^\pm z+S_4^\pm.
\end{equation*}
Note that conditions \eqref{assumptions} are satisfied, which means we can apply the method to find the generic solution as described in Section \ref{qrtsolve}.
First of all, the invariant of \eqref{losystemeq} is given by
\begin{equation*}
I(F,G)=\frac{F^2+S_2^- FG+S_4^- G^2-S_1^- F-S_3^-G}{FG-1},
\end{equation*}
and we set $I(F,G)=P$.\\
The linear system \eqref{mcmillanlinear} becomes
\begin{align}
F+\overline{F}+(S_2^--P)G=S_1^-, & & G+\overline{G}+b_1b_2b_3b_4(S_2^--P)\overline{F}=S_1^+. \label{mcmillanFG}
\end{align}
If $b_1b_2b_3b_4(P-S_2^-)^2\neq 4$, there exists an equilibrium solution $\left(F_{eq},G_{eq}\right)$ to this system given by
\begin{subequations}
\label{FGeqsol}
\begin{align}
F_{eq}=\frac{S_1^+(P-S_2^-)+2S_1^-}{4-b_1b_2b_3b_4(P-S_2^-)^2}, && G_{eq}=\frac{S_1^-(P-S_2^-)+2S_1^+}{4-b_1b_2b_3b_4(P-S_2^-)^2}.
\end{align}
\end{subequations}
The special case $b_1b_2b_3b_4(P-S_2^-)^2=4$, requires a separate analysis, which we discuss in Section \ref{section:loga}.
The matrix $M$ \eqref{Mdefi} equals
\begin{equation*}
M=\begin{pmatrix}
-1 & P-S_2^-\\
-b_1b_2b_3b_4(P-S_2^-)& b_1b_2b_3b_4(P-S_2^-)^2-1
\end{pmatrix},
\end{equation*}
and its characteristic equation is given by
\begin{equation}\label{eq:Mcharacter}
|M-\lambda I|=\lambda^2+\left(2-b_1b_2b_3b_4 (P-S_2^-)^2\right)\lambda+1=0.
\end{equation}
At this stage, we consider $P$ as a formal variable satisfying $\overline{P}=P$, and as such, the characteristic equation of $M$ does not have a solution $\lambda\in\mathbb{C}(P)$. However we can rewrite \eqref{eq:Mcharacter} as
\begin{equation*}
b_1b_2b_3b_4 (P-S_2^-)^2=\lambda+2+\lambda^{-1}=\left(\lambda^{\tfrac{1}{2}}+\lambda^{-\tfrac{1}{2}}\right)^2,
\end{equation*}
which inspires us to reparameterise
\begin{equation}
\label{Pdefi}
P=\epsilon_0+\frac{\Lambda}{b_1b_2b_3b_4}+\Lambda^{-1},
\end{equation}
where $\overline{\Lambda}=\Lambda$, giving
\begin{equation*}
|M-\lambda I|=\left(\lambda-\frac{\Lambda^2}{b_1b_2b_3b_4}\right)\left(\lambda-\frac{b_1b_2b_3b_4}{\Lambda^2}\right).
\end{equation*}
We put $\lambda=\frac{\Lambda^2}{b_1b_2b_3b_4}$, and $M$ can be diagonalised as follows,
\begin{align*}
M=Q\begin{pmatrix}
\lambda & 0\\
0& \lambda^{-1}
\end{pmatrix} Q^{-1}, && Q=\begin{pmatrix}
1 & 1\\
\Lambda& \frac{b_1b_2b_3b_4}{\Lambda}
\end{pmatrix}.
\end{align*}
We introduce an independent variable $\phi$ characterised by $\overline{\phi}=\lambda \phi$, which allows us to write the general solution to the linear system \eqref{mcmillanFG} as
\begin{align}
F(\phi)=F_{\text{eq}}(\Lambda,\mathbf{b})+\phi+\mu\phi^{-1},&& G(\phi)=G_{\text{eq}}(\Lambda,\mathbf{b})+\Lambda\phi+\frac{b_1b_2b_3b_4}{\Lambda}\mu\phi^{-1}, \label{FGdefi}
\end{align}
where $\mu$ is an arbitrary periodic constant, that is, $\overline{\mu}=\mu$, and by substituting identity \eqref{Pdefi} into equations \eqref{FGeqsol},
\begin{align*}
F_{\text{eq}}(\Lambda,\mathbf{b})&=-\frac{b_1b_2b_3b_4\Lambda\left(S_1^++2S_1^-\Lambda+S_3^-\Lambda^2\right)}{(b_1b_2b_3b_4-\Lambda^2)^2},\\
G_{\text{eq}}(\Lambda,\mathbf{b})&=-\frac{b_1b_2b_3b_4\Lambda\left(S_3^++2S_1^+\Lambda+S_1^-\Lambda^2\right)}{(b_1b_2b_3b_4-\Lambda^2)^2}.
\end{align*}
It is easy to see, that for $F$ and $G$ as defined in equation \eqref{FGdefi}, the identity $I(F,G)=P$, is equivalent to
\begin{equation}
\label{mudefi}
\mu=\mu(\Lambda,\mathbf{b}):=\frac{\Lambda(\Lambda+b_1b_2)(\Lambda+b_1b_3)(\Lambda+b_1b_4)(\Lambda+b_2b_3)(\Lambda+b_2b_4)(\Lambda+b_3b_4)}{\left(b_1b_2b_3b_4-\Lambda^2\right)^4}.
\end{equation}
So $F$ and $G$ as defined in equation \eqref{FGdefi}, with $\mu=\mu(\Lambda,\mathbf{b})$ as defined above, satisfy equations \eqref{mcmillanFG} and \eqref{inv1}. Hence, by Lemma \ref{lemqrtmclillan},
\begin{align}
F(\phi)=\phi+F_{\text{eq}}(\Lambda,\mathbf{b})+\mu(\Lambda,\mathbf{b})\phi^{-1},&& G(\phi)=\Lambda\phi+G_{\text{eq}}(\Lambda,\mathbf{b})+\frac{b_1b_2b_3b_4}{\Lambda}\mu(\Lambda,\mathbf{b})\phi^{-1}, \label{generalsolaut}
\end{align}
defines a solution to the QRT mapping \eqref{losystemeq}, for all $\Lambda$ and $\phi$ satisfying
\begin{align}
\overline{\Lambda}=\Lambda,&& \overline{\phi}=\lambda \phi,&& \lambda=\frac{\Lambda^2}{b_1b_2b_3b_4}.\label{indvariabledefi}
\end{align}

\subsection{Six special families of solutions}
\label{section:special6}
Note that the autonomous system \eqref{losystemeq} reduces to the system of algebraic equations \eqref{degreeoneterm} if we assume $\overline{F}_1=F_1$ and $\overline{G}_1=G_1$. In particular, equations \eqref{orderonecoef} give three constant solutions to system \eqref{losystemeq}. In this section we see that any of these constant solutions has two associated $1$-parameter families of solutions of  \eqref{losystemeq}. We denote the roots of $\mu(\Lambda,\mathbf{b})$ by $\Lambda=\Lambda_k^\pm$ ($k=1,2,3$), where
\begin{align*}
\Lambda_1^{+}&=-b_1b_2,&\Lambda_2^{+}&=-b_2b_3,& \Lambda_3^{+}&=-b_1b_3,\\
\Lambda_1^{-}&=-b_3b_4,& \Lambda_2^{-}&=-b_1b_4,& \Lambda_3^{-}&=-b_2b_4.
\end{align*}
Let $k\in\{1,2,3\}$, then we have
\begin{align}
F_{\text{eq}}\left(\Lambda_k^{\pm},\mathbf{b}\right)=f_1^{(1,k)},&& G_{\text{eq}}\left(\Lambda_k^{\pm},\mathbf{b}\right)=g_1^{(1,k)}, \label{constantsolaut}
\end{align}
where the $f_1^{(1,k)}$ and $g_1^{(1,k)}$, as defined in \eqref{orderonecoef}, denote a constant solutions of \eqref{losystemeq}.\\
Associated we find two special $1$-parameter families of solutions, by setting $\Lambda=\Lambda_k^\pm$ in \eqref{generalsolaut}, given by
\begin{align} \label{eq:6specialfamilies}
F_k^\pm(\phi):=\phi+f_1^{(1,k)},&& G_k^\pm(\phi):=\Lambda_k^\pm\phi+g_1^{(1,k)}, && \overline{\phi}=\lambda_k^{\pm 1}\phi,
\end{align}
where
\begin{align}
\lambda_1=\frac{b_1b_2}{b_3b_4}, &&\lambda_2=\frac{b_2b_3}{b_1b_4}, && \lambda_3=\frac{b_1b_3}{b_2b_4}.
\end{align}
Note that for the particular choice $\phi=0$, the families $\left(F_k^+(\phi),G_k^+(\phi)\right)$ and $\left(F_k^-(\phi),G_k^-(\phi)\right)$  coincide with the constant solution $\left(f_1^{(1,k)},g_1^{(1,k)}\right)$.

\subsection{Exceptional logarithmic-type solutions}
\label{section:loga}
We consider the remaining case 
\begin{equation*}
b_1b_2b_3b_4(P-S_2^-)^2=4,
\end{equation*} for the linear system \eqref{mcmillanFG}. Note that the equilibrium solution \eqref{FGeqsol} no longer exists and we show that this case gives rise to logarithmic-type solutions. We write $r_\pm=\pm \sqrt{b_1b_2b_3b_4}$ and assume
\begin{equation} \label{eq:Pcondition}
P=S_2^-+\frac{2}{r_\pm}.
\end{equation}
The system of equations \eqref{mcmillanFG} becomes
	\begin{align}
	\overline{F}=-F+\frac{2}{r_\pm} G+S_1^-,&&
	\overline{G}=-2r_\pm F+3G+2r_\pm S_1^- +S_1^+.\label{FGloglin}
	\end{align}
We write $V=G-r_\pm F$, then
\begin{equation}\label{eq:V}
\overline{V}=V+S_1^++r_\pm S_1^-,
\end{equation}
and we therefore introduce a formal variable $\chi$ which satisfies 
\begin{equation}
\label{chieq}
\overline{\chi}=\chi+1,
\end{equation}
and set
\begin{equation*}
V=\left(S_1^++r_\pm S_-\right)\chi.
\end{equation*}
This allows the first equation in \eqref{FGloglin} to be rewritten as
\begin{equation*}
\overline{F}=F+2\left(\frac{1}{r_\pm}S_1^++S_-\right)\chi+S_1^-,
\end{equation*}
which gives
\begin{equation}
\label{Flogdefi}
F(\chi)=F_0-\frac{1}{r_\pm}S_1^+\chi+\left(\frac{1}{r_\pm}S_1^++S_1^-\right)\chi^2,
\end{equation}
for some $F_0$ with $\overline{F}_0=F_0$.\\
As $V=G-r_\pm F$, we obtain a corresponding expression for $G$,
\begin{equation}
\label{Glogdefi}
G(\chi)=r_\pm F_0+r_\pm S_1^-\chi+\left(S_1^++r_\pm S_1^-\right)\chi^2.
\end{equation}
Upon substitution of \eqref{Flogdefi} and \eqref{Glogdefi} into the identity $I(F,G)=P$, or equivalently into the leading order autonomous system \eqref{losystemeq}, we find
\begin{equation*}
F_0=\frac{2+r_\pm S_2^-}{S_1^++r_\pm S_1^-}.
\end{equation*}
We conclude that the general solution of the leading order autonomous system \eqref{losystemeq}, subject to \eqref{eq:Pcondition}, is given by
\begin{subequations}
	\label{eq:logsol}
\begin{align}
F_l^\pm(\chi)&=\frac{2+r_\pm S_2^-}{S_1^++r_\pm S_1^-}-\frac{1}{r_\pm}S_1^+\chi+\left(\frac{1}{r_\pm}S_1^++S_1^-\right)\chi^2,\\
G_l^\pm(\chi)&=\frac{2r_\pm+S_2^+}{S_1^++r_\pm S_1^-}+r_\pm S_1^-\chi+\left(S_1^++r_\pm S_1^-\right)\chi^2,
\end{align}
\end{subequations}
where $\chi$ free satisfying \eqref{chieq}.\\
The subscripts `l' stand for logarithmic-type, as the time evolution of $\chi$, equation \eqref{chieq}, is characteristic for $\log_q(t)$  when interpreted as a $q$-difference equation in $t$.
Note that we used $S_1^++r_\pm S_1^-\neq 0$ in the above derivation, we leave the degenerate case $S_1^++ r_\pm S_1^-=0$ to the interested reader.

\begin{remark}
	We would like to note that the classification of solutions of \eqref{losystemeq} is now complete. That is, given any   initial data $(F_0,G_0)\in\mathbb{C}^2$ satisfying regularity condition $F_0\cdot G_0\neq 0,1$. Let $F_{n+1}=\overline{F}_n$ and $G_{n+1}=\overline{G}_n$ be defined recursively by \eqref{losystemeq} for $n\in\mathbb{Z}$. Then $(F_n,G_n)_{n\in\mathbb{Z}}$ is captured by \eqref{generalsolaut} or \eqref{eq:logsol}. Indeed let $P=I(F_0,G_0)$, and assume $b_1b_2b_3b_4(P-S_2^-)^2\neq 4$, then \eqref{Pdefi} has two distinct solutions $\Lambda=\Lambda_1,\Lambda_2$, which are related by $\Lambda_1\Lambda_2=b_1b_2b_3b_4$. Hence $\mu(\Lambda_1)=0$ iff $\mu(\Lambda_2)=0$. Assume $\mu(\Lambda_1)\neq 0$, then the overdetermined system
	\begin{align*}
	F_0=\phi_0+F_{\text{eq}}(\Lambda_1,\mathbf{b})+\mu(\Lambda_1,\mathbf{b})\phi_0^{-1},&& G_0=\Lambda_1\phi_0+G_{\text{eq}}(\Lambda_1,\mathbf{b})+\frac{b_1b_2b_3b_4}{\Lambda_1}\mu(\Lambda_1,\mathbf{b})\phi_0^{-1},
	\end{align*}
	has an unique solution $\phi_0\in\mathbb{C}^*$, and we obtain, for all $n\in\mathbb{Z}$,
	\begin{align*}
	F_n=\phi_n+F_{\text{eq}}(\Lambda_1,\mathbf{b})+\mu(\Lambda_1,\mathbf{b})\phi_n^{-1},&& G_n=\Lambda_1\phi_n+G_{\text{eq}}(\Lambda_1,\mathbf{b})+\frac{b_1b_2b_3b_4}{\Lambda_1}\mu(\Lambda_1,\mathbf{b})\phi_n^{-1},
	\end{align*}
	where $\phi_n=\phi_0 \lambda_1^n$ with $\lambda_1=\Lambda_1^2/(b_1b_2b_3b_4)$.
	Of course the choice $\Lambda=\Lambda_2$ would have led to the same result. This, however, is no longer the case when $\mu(\Lambda_1)=0$. We leave it to the reader to work through these degenerate cases as well as the logarithmic one,  $b_1b_2b_3b_4(P-S_2^-)^2=4$.
\end{remark}

\section{The formal series solution}
\label{sectiongeneral}
In equations \eqref{autonomoufg} we saw that the leading order behaviour of the solutions $f=tf_1$ and $g=tg_1$ is described by the autonomous system \eqref{losystemeq}. Furthermore we found the general solution to this autonomous system in the previous section. We therefore consider the following formal solution ansatz for $q$-$P(A_1)$,
\begin{align*}
f=\sum_{i=1}^\infty F_i t^i, && g=\sum_{i=1}^\infty G_i t^i.
\end{align*}
This approach reduces to the power series method if we assume that the $F_i$ and $G_i$ are plain complex numbers. However for now we work with these coefficients on a formal level, for example,
\begin{align*}
\overline{f}=\sum_{i=1}^\infty q^i\overline{F}_i t^i, && \overline{g}=\sum_{i=1}^\infty q^i\overline{G}_i t^i.
\end{align*}
We substitute these formal series into equations \eqref{qpa1expand} and compare coefficients of $t$ order by order. First of all, note that no terms $t^n$ with $n<4$ occur in equations \eqref{qpa1expand}. By comparing the coefficients of $t^4$ in equations \eqref{qpa1expand} we obtain the autonomous system \eqref{losystemeq} with $F=F_1$ and $G=G_1$. As to the higher order coefficients, for $n>1$, by comparing the coefficients of $t^{n+3}$ in equations \eqref{qpa1expand}, we obtain
\begin{subequations}
\label{recursiongencoef}
\begin{align}
&\begin{multlined}[b][.7\textwidth]
G_1(G_1F_1-1)q^n\overline{F}_n+qG_1(G_1\overline{F}_1-1)F_n+
q\left(2G_1F_1\overline{F}_1-F_1-\overline{F}_1\right)G_n=\\
\frac{q}{b_1b_2b_3b_4}Q^{(1)}(G_1)G_n+R_n^{(1)}\left[(F_i)_{1\leq i<n},(\overline{F}_i)_{1\leq i<n},(G_i)_{1\leq i<n}\right],
\end{multlined}\\
&\begin{multlined}[b][.7\textwidth]
q\overline{F}_1\left(\overline{F}_1\overline{G}_1-1\right)G_n+\overline{F}_1\left(G_1\overline{F}_1-1\right)q^n \overline{G}_n+
\left(2F_1G_1\overline{G}_1-G_1-\overline{G}_1\right)q^n\overline{F}_n=\\
b_1b_2b_3b_4Q^{(-1)}\left(\overline{F}_1\right)q^n\overline{F}_n+R_n^{(2)}\left[(\overline{F}_i)_{1\leq i<n},(G_i)_{1\leq i<n},(\overline{G}_i)_{1\leq i<n}\right],
\end{multlined}
\end{align}
\end{subequations}
for certain $R_n^{(1)}$ and $R_n^{(2)}$ which are polynomial with respect to their inputs, where the polynomials $Q^{(i)}(z)$ are defined by
\begin{equation*}
Q^{(i)}(z)=\frac{d}{dx}\left[\left(x-b_1^i\right)\left(x-b_2^i\right)\left(x-b_3^i\right)\left(x-b_4^i\right)\right],
\end{equation*}
for $i\in\{1,-1\}$.\\
Note that these equations are linear autonomous equations with respect to $F_n$ and $G_n$. It is straightforward to obtain explicit expressions for $R_n^{(1)}$ and $R_n^{(2)}$, these are however rather lengthy, which is why we omit them. As an example, $R_2^{(1)}$ and $R_2^{(2)}$ are given by
\begin{align*}
R_2^{(1)}\left(F_1,\overline{F}_1,G_1\right)&=\left(b_5^{-1}+b_6^{-1}+b_7^{-1}+b_8^{-1}\right)qG_1(F_1G_1-1)\left(\overline{F}_1G_1-1\right),\\
R_2^{(2)}\left(\overline{F}_1,G_1,\overline{G}_1\right)&=\left(b_5+b_6+b_7+b_8\right)q^2\overline{F}_1\left(\overline{F}_1G_1-1\right)\left(\overline{F}_1\overline{G}_1-1\right).
\end{align*}
Furthermore the polynomials $R_n^{(1)}$ and $R_n^{(2)}$ are of degree at most $n+3$ with respect to the weighted gradation $\deg_w$ on $\mathbb{C}\left[\cup_{i=1}^\infty\{F_i,\overline{F}_i,G_i,\overline{G}_i\}\right]$, which is uniquely defined by its values on the generators of this polynomial ring, as
\begin{align*}
\deg_w{F_i}=\deg_w{\overline{F}_i}=\deg_w{G_i}=\deg_w{\overline{G}_i}=i,
\end{align*}
for $i\in\mathbb{N}^*$.\\
The importance of this observation becomes clear when we substitute the generic solution \eqref{generalsolaut} to equations \eqref{losystemeq} for $F_1$ and $G_1$.
Indeed, if we set $F_1=F_1(\phi)=F(\phi)$ and $G_1=G_1(\phi)=G(\phi)$ as defined in equations \eqref{generalsolaut}, then $F_1(\phi)$ and $G_1(\phi)$ are Laurent polynomials in $\phi$ of degree $1$ in both $\phi$ and $\phi^{-1}$.
Hence the right-hand sides of equations \eqref{recursiongencoef} for $n=2$, are Laurent polynomials in $\phi$ of at most degree $n+3=5$ in both $\phi$ and $\phi^{-1}$, which shows that the system of equations \eqref{recursiongencoef} for $n=2$ possibly has a solution $\left(F_2(\phi),G_2(\phi)\right)$, such that $F_2(\phi)$ and $G_2(\phi)$  are Laurent polynomials in $\phi$ of at most degree $2$ in both $\phi$ and $\phi^{-1}$. Indeed a lengthy calculation confirms this.
More generally, we conjecture that there is an unique solution $\left((F_n(\phi))_{n=1}^\infty,(G_n(\phi))_{n=1}^\infty\right)$ to equations \eqref{recursiongencoef} with $F_1(\phi)=F(\phi)$ and $G_1(\phi)=G(\phi)$ as above, such that $F_n(\phi)$ and $G_n(\phi)$ are Laurent polynomials in $\phi$ of at most degree $n$ in both $\phi$ and $\phi^{-1}$. An equivalent formulation of this statement is given in Conjecture \ref{conjecture}. This however seems difficult to prove directly and we hence prove a weaker version, which states that there is an unique solution where the coefficients $F_n(\phi)$ and $G_n(\phi)$ are Laurent series in $\phi$ with highest order term of degree less or equal to $n$, for $n\in\mathbb{N}^*$. 

\begin{thm}
\label{thmgeneralsol0+}
There exists an unique formal series solution to $q$-$P(A_1)$ of the form
\begin{align}
f^{0,+}(t,\phi;\Lambda,\mathbf{b})=\sum_{n=1}^\infty{F_n^{0,+}(\phi;\Lambda,\mathbf{b})t^n}, && g^{0,+}(t,\phi;\Lambda,\mathbf{b})=\sum_{n=1}^\infty{G_n^{0,+}(\phi;\Lambda,\mathbf{b})t^n},\label{powerexpansionsol0+}
\end{align}
with, for $n\in\mathbb{N}^*$,
\begin{align}
F_n^{0,+}(\phi;\Lambda,\mathbf{b})=\sum_{i=-\infty}^n{F_{n,i}^{0,+}(\Lambda,\mathbf{b})\phi^i}, && G_n^{0,+}(\phi;\Lambda,\mathbf{b})=\sum_{i=-\infty}^n{G_{n,i}^{0,+}(\Lambda,\mathbf{b})\phi^i},\label{coefpowerexpansionsol0+}
\end{align}
where $F_{1,1}^{0,+}(\Lambda,\mathbf{b})=1$, $G_{1,1}^{0,+}(\Lambda,\mathbf{b})=\Lambda$ and $\Lambda$ and $\phi$ satisfy equations \eqref{indvariabledefi}, with
\begin{align*}
q=q(\mathbf{b})=\frac{b_1b_2b_3b_4}{b_5b_6b_7b_8},&&\lambda=\lambda(\Lambda,\mathbf{b})=\frac{\Lambda^2}{b_1b_2b_3b_4}.
\end{align*}
For $n\in\mathbb{N}^*$ and $i\in\mathbb{Z}_{\leq n}$, the coefficients $F_{n,i}^{0,+}(\Lambda,\mathbf{b})$ and $G_{n,i}^{0,+}(\Lambda,\mathbf{b})$ are rational functions in their inputs, which are regular at points $(\Lambda,\mathbf{b})\in\mathbb{C}^*\times\mathcal{B}$ such that
\begin{equation}
1\notin Q:=\{q_1^mq_2^n:(m,n)\in\mathbb{N}^2\setminus\{(0,0)\}\}, \label{Qcondition}
\end{equation}
where $q_1=q_1(\mathbf{b},\Lambda)=q\lambda$ and $q_2=q_2(\mathbf{b},\Lambda)=\lambda^{-1}$.\\
Furthermore, for fixed $\mathbf{b}\in\mathcal{B}$ with $|q|<1$, for any $\Lambda\in L_0(\mathbf{b})$, where
\begin{equation}\label{eq:L0defi}
L_0(\mathbf{b}):=\{x\in\mathbb{C}^*: |b_1b_2b_3b_4|<|x|^2<|b_5b_6b_7b_8|\},
\end{equation}
condition \eqref{Qcondition} is satisfied and this formal solution, written in terms of the variables 
$\zeta_1=t\phi$ and $\zeta_2=\phi^{-1}$,
\begin{subequations}
\label{fgzetavar}
\begin{align}
f^{0,+}(\zeta_1\zeta_2,\zeta_2^{-1};\Lambda,\mathbf{b})&=\sum_{n=1}^\infty{\sum_{m=0}^\infty{F_{n,n-m}^{0,+}(\Lambda,\mathbf{b})\zeta_1^n\zeta_2^m}},\\ g^{0,+}(\zeta_1\zeta_2,\zeta_2^{-1};\Lambda,\mathbf{b})&=\sum_{n=1}^\infty{\sum_{m=0}^\infty{G_{n,n-m}^{0,+}(\Lambda,\mathbf{b})\zeta_1^n\zeta_2^m}},
\end{align}
\end{subequations}
converges near $(\zeta_1,\zeta_2)=(0,0)$.\\
In fact, these expansions also depend holomorphically on $\Lambda$. That is, for any $L\subseteq L_0(\mathbf{b})$ open with $\overline{L}\subseteq L_0(\mathbf{b})$, there is an open environment $Z\subseteq \mathbb{C}^2$ of $\mathbf{0}$, such that the series \eqref{fgzetavar} converge uniformly on $Z\times L$, defining holomorphic functions on this set in $\left(\bm{\zeta},\Lambda\right)$.
\end{thm}
\begin{proof}
We apply the $q$-Briot-Bouquet Theorem \ref{genbriottheoremsev} with $m=2$ to $q$-$P(A_1)$, after a change of dependent and independent variables. More precisely, inspired by equations \eqref{fgzetavar}, we introduce the following variables,
\begin{align}
\zeta_1=t\phi, &&  \zeta_2=\phi^{-1},&& y_1=\frac{f}{\zeta_1}-1,&&y_2=\frac{g}{\zeta_1}-\Lambda, \label{changefgy1y2}
\end{align}
where $\zeta_1$ and $\zeta_2$, in accordance with equations \eqref{indvariabledefi}, satisfy
\begin{align*}
\overline{\zeta}_1=q_1\zeta_1, && \overline{\zeta}_2=q_2\zeta_2.
\end{align*}
As $t=\zeta_1\zeta_2$, we can rewrite $q$-$P(A_1)$ in terms of these new variables as
\begin{subequations}
\label{qpa1zetavariables}
\begin{align}
y_1(q_1\zeta_1,q_2\zeta_2)&=H_1\left(\zeta_1,\zeta_2,y_1(\zeta_1,\zeta_2),y_2(\zeta_1,\zeta_2);\Lambda,\mathbf{b}\right),\\
y_2(q_1\zeta_1,q_2\zeta_2)&=H_2\left(\zeta_1,\zeta_2,y_1(\zeta_1,\zeta_2),y_2(\zeta_1,\zeta_2);\Lambda,\mathbf{b}\right),
\end{align}
\end{subequations}
for certain rational functions $H_1(\zeta_1,\zeta_2,y_1,y_2;\Lambda,\mathbf{b})$ and $H_2(\zeta_1,\zeta_2,y_1,y_2;\Lambda,\mathbf{b})$.\\
We apply the $q$-Briot-Bouquet Theorem \ref{genbriottheoremsev} to this system of $q$-difference equations.
We denote
\begin{align*}
\bm{\zeta}=(\zeta_1,\zeta_2),&&\mathbf{y}=(y_1,y_2),&&\mathbf{q}=(q_1,q_2),
\end{align*}
and leave it to the interested reader to write down $H_1(\bm{\zeta},\mathbf{y};\Lambda,\mathbf{b})$ and $H_2(\bm{\zeta},\mathbf{y};\Lambda,\mathbf{b})$ explicitly.
A rather lengthy calculation shows
\begin{align*}
H_1(\mathbf{0},\mathbf{y};\Lambda,\mathbf{b})=\frac{(y_2+\Lambda)^2}{\Lambda^2(y_1+1)}-1,&&H_2(\mathbf{0},\mathbf{y};\Lambda,\mathbf{b})=\frac{(y_2+\Lambda)^3}{\Lambda^2(y_1+1)^2}-\Lambda,
\end{align*}
in particular $H(\mathbf{0},\mathbf{0};\mathbf{q},\Lambda)=\mathbf{0}$ and we have
\begin{equation}
\label{DH}
D(\Lambda,\mathbf{b}):=\begin{pmatrix}
\frac{\partial H_1}{\partial y_1}(\mathbf{0},\mathbf{0};\Lambda,\mathbf{b}) &&\frac{\partial H_1}{\partial y_2}(\mathbf{0},\mathbf{0};\Lambda,\mathbf{b}) \\
\frac{\partial H_2}{\partial y_1}(\mathbf{0},\mathbf{0};\Lambda,\mathbf{b})  &&\frac{\partial H_2}{\partial y_2}(\mathbf{0},\mathbf{0};\Lambda,\mathbf{b}) 
\end{pmatrix}=\begin{pmatrix}
-1 & 2\Lambda^{-1}\\
-2\Lambda & 3
\end{pmatrix}.
\end{equation}
Note that $1$ is the only eigenvalue of $D(\Lambda,\mathbf{b})$, with multiplicity $2$. Therefore, by the $q$-Briot-Bouquet Theorem \ref{genbriottheoremsev}, if conditions \eqref{Qcondition} are satisfied, then the system of $q$-difference equations \eqref{qpa1zetavariables} has an unique power series solution of the form
\begin{equation}
\label{yzeta}
y_i(\zeta_1,\zeta_2;\Lambda,\mathbf{b})=\sum_{n=0}^\infty{\sum_{m=0}^\infty{y_{n,m}^{(i)}(\Lambda,\mathbf{b})\zeta_1^n\zeta_2^m}},
\end{equation}
with $y_{0,0}^{(i)}(\Lambda,\mathbf{b})=0$ for $i\in\{1,2\}$.\\
Associated via equations \eqref{changefgy1y2}, we have the following expansions for $f=f(\zeta_1,\zeta_2;\Lambda,\mathbf{b})$ and $g=g(\zeta_1,\zeta_2;\Lambda,\mathbf{b})$,
\begin{align}
f(\zeta_1,\zeta_2;\Lambda,\mathbf{b})=\sum_{n=1}^\infty{\sum_{m=0}^\infty{f_{n,m}(\Lambda,\mathbf{b})\zeta_1^n\zeta_2^m}}, && g(\zeta_1,\zeta_2;\Lambda,\mathbf{b})=\sum_{n=1}^\infty{\sum_{m=0}^\infty{g_{n,m}(\Lambda,\mathbf{b})\zeta_1^n\zeta_2^m}}, \label{fgzeta}
\end{align}
where the coefficients are defined by
\begin{align*}
f_{n,m}(\Lambda,\mathbf{b})=y^{(1)}_{n-1,m}(\Lambda,\mathbf{b}), &&g_{n,m}(\Lambda,\mathbf{b})=y^{(2)}_{n-1,m}(\Lambda,\mathbf{b}),
\end{align*}
for $n\in\mathbb{N}^*$ and $m\in\mathbb{N}$ with $(n,m)\neq (1,0)$, and
\begin{align*}
f_{1,0}(\Lambda,\mathbf{b})=1, && g_{1,0}(\Lambda,\mathbf{b})=\Lambda.
\end{align*}
Rewriting these expansions in terms of the original independent variables $t$ and $\phi$, the formulas
\begin{align*}
f^{0,+}(t,\phi;\Lambda,\mathbf{b})=f\left(t\phi,\phi^{-1};\Lambda,\mathbf{b}\right), && g^{0,+}(t,\phi;\Lambda,\mathbf{b})=g\left(t\phi,\phi^{-1};\Lambda,\mathbf{b}\right),
\end{align*}
define formal series solutions of $q$-$P(A_1)$, precisely as described in equations \eqref{powerexpansionsol0+}. Furthermore the $q$-Briot-Bouquet Theorem \ref{genbriottheoremsev} implies that the power series \eqref{fgzeta} converge in an open environment of $(\zeta_1,\zeta_2)=(0,0)$, if $1$ is not a limit point of $Q$. Note that this condition is trivially satisfied if $0<|q_1|,|q_2|<1$, which is equivalent to $\Lambda\in L_0(\mathbf{b})$. We conclude that the series \eqref{fgzetavar} indeed converge locally at $(\zeta_1,\zeta_2)=(0,0)$, for $\Lambda\in L_0(\mathbf{b})$. Strictly speaking, this only shows that the series defines a solution in the two variables $t$ and $\phi$ for a fixed $\Lambda\in\mathbb{C}^*$ such that condition \eqref{Qcondition} holds. Note however, that the proof of Theorem \ref{genbriottheoremsev} gives an explicit recursion for the coefficients, which proves  that the coefficients $F_{n,i}^{0,+}(\Lambda,\mathbf{b})$ and $G_{n,i}^{0,+}(\Lambda,\mathbf{b})$ are rational functions in their inputs and the formal series solution defines a solution on a formal level.
This finishes the proof of the first part of the theorem.

As to the second part, we would like to prove that the solutions \eqref{fgzeta} depend holomorphically on $\Lambda$, which is equivalent to proving that the expansions 
\eqref{yzeta} are holomorphic in $\Lambda$. To this end we apply Theorem \ref{genbriottheoremsevuniform}.
As, for any $L\subseteq L_0(\mathbf{b})$ with $\overline{L}\subseteq L_0(\mathbf{b})$, the set $\overline{L}$ is compact, a simple compactness argument shows that it suffices to prove that for any $\Lambda_0\in L_0(\mathbf{b})$, there is an open environment $L\subseteq L_0(\mathbf{b})$  of $\Lambda_0$, and an open environment $Z\subseteq \mathbb{C}^2$ of $\mathbf{0}$, such that the series \eqref{yzeta} converge uniformly on $Z\times L$.

 So let us take a $\Lambda_0\in L_0(\mathbf{b})$, we denote $\mathbf{q}_0=\left(q_1(\Lambda_0,\mathbf{b}),q_2(\Lambda_0,\mathbf{b})\right)$ and determine an $r>0$ such that
\begin{equation*}
B_{\text{max}}^2(\mathbf{q}_0,r)\subseteq \overline{B}_{\text{max}}^2(\mathbf{q}_0,r)\subseteq B_{\text{max}}^2(\mathbf{0},1)\setminus\{\mathbf{q}\in\mathbb{C}^2| q_1 q_2=0\}\subseteq\mathbb{C}^2,
\end{equation*}
and set $U=B_{\text{max}}^2(\mathbf{q}_0,r)$.\\
We have to modify the functions $H_1(\bm{\zeta},\mathbf{y};\Lambda,\mathbf{b})$ and $H_2(\bm{\zeta},\mathbf{y};\Lambda,\mathbf{b})$ a bit in order to be able to apply Theorem \ref{genbriottheoremsevuniform}, as $\Lambda$ and $\mathbf{b}$ are not independent of $\mathbf{q}=(q_1,q_2)$. Indeed, we have to reparameterise all variables in terms of $q_1$ and $q_2$.
 To this end, we keep the value of $b_i$ fixed for $2\leq i \leq 8$, but allow $b_1$ and $\Lambda$ to vary with $\mathbf{q}$. More explicitly, we define
\begin{align}
b_1'(\mathbf{q})=\frac{q_1q_2b_5b_6b_7b_8}{b_2b_3b_4}, && \mathbf{b}'(\mathbf{q})=(b_1'(\mathbf{q}),b_2,b_3,b_4,b_5,b_6,b_7,b_8),&&, \Lambda(\mathbf{q})=(b_5b_6b_7b_8)^{\tfrac{1}{2}}q_1^{\tfrac{1}{2}},\label{eq:repara}
\end{align}
for $\mathbf{q}\in U$, where we choose the sign of the square root such that $\Lambda(\mathbf{q}_0)=\Lambda_0$.\\
Note that at $\mathbf{q}=\mathbf{q}_0$, the original values of the parameters are recovered, as
\begin{align*}
\mathbf{b}'(\mathbf{q}_0)=\mathbf{b},&& \Lambda(\mathbf{q}_0)=\Lambda_0,
\end{align*}
and $\Lambda(\mathbf{q})$ is a univalued holomorphic function on $U$.\\
We modify $H_1$ and $H_2$, by setting
\begin{equation*}
\tilde{H}(\bm{\zeta},\mathbf{y};\mathbf{q})=H(\bm{\zeta},\mathbf{y};\Lambda(\mathbf{q}),\mathbf{b}'(\mathbf{q})).
\end{equation*}
The function $\tilde{H}(\bm{\zeta},\mathbf{y};\mathbf{q})$ is holomorphic at $(\bm{\zeta},\mathbf{y},\mathbf{q})=(\bm{0},\mathbf{0},\mathbf{q}')$ with $\tilde{H}(\bm{0},\mathbf{0},\mathbf{q}')=\mathbf{0}$, for every $\mathbf{q}'\in U$. The relevant Jacobian matrix of $\tilde{H}$ is given by
\begin{equation*}
\tilde{D}(\mathbf{q}):=\begin{pmatrix}
\frac{\partial \tilde{H}_1}{\partial y_1}(\mathbf{0},\mathbf{0};\mathbf{q})  &&\frac{\partial \tilde{H}_1}{\partial y_2}(\mathbf{0},\mathbf{0};\mathbf{q}) \\
\frac{\partial \tilde{H}_2}{\partial y_1}(\mathbf{0},\mathbf{0};\mathbf{q})  &&\frac{\partial \tilde{H}_2}{\partial y_2}(\mathbf{0},\mathbf{0};\mathbf{q}) 
\end{pmatrix}=D(\Lambda(\mathbf{q}),\mathbf{b}'(\mathbf{q}))=\begin{pmatrix}
-1 & 2\Lambda(\mathbf{q})^{-1}\\
-2\Lambda(\mathbf{q}) & 3
\end{pmatrix},
\end{equation*}
for $\mathbf{q}\in U$, where $D(\Lambda,\mathbf{b})$ is the Jacobian matrix of $H$, as defined in equation \eqref{DH}.\\
Again $1$ is the only eigenvalue of $\tilde{D}(\mathbf{q})$, which is not an element of $Q_0$ as defined in \eqref{eq:defiQ} with $m=2$, for $\mathbf{q}\in U$.
We can hence apply Theorem \ref{genbriottheoremsevuniform}, which gives open environments $Z\subseteq \mathbb{C}^2$ and $V\subseteq U$ of $\mathbf{0}$ and $\mathbf{q}_0$ respectively, such that the series $y_i(\bm{\zeta};\Lambda(\mathbf{q}),\mathbf{b}'(\mathbf{q}))$, with notation as in equation \eqref{yzeta} for $i=1,2$, converge uniformly on $Z\times V$, defining holomorphic functions in $(\bm{\zeta},\mathbf{q})$ on this set.
To undo the reparameterisation \eqref{eq:repara}, we define
\begin{equation*}
\mathbf{q}(s)=\left(\frac{s^2}{b_5b_6b_7b_8},\frac{b_1b_2b_3b_4}{s^2}\right),
\end{equation*}
 and determine an open connected environment $L\subseteq L_0(\mathbf{b})$ of $\Lambda_0$, such that 
 \begin{equation*}
 \{\mathbf{q}(s):s\in L\}\subseteq V.
 \end{equation*}
 Then we know that the series
 \begin{equation*}
 Y_i(\bm{\zeta},s):=y_i\left(\bm{\zeta};\Lambda\left(\mathbf{q}(s)\right),\mathbf{b}'\left(\mathbf{q}(s)\right)\right),
 \end{equation*}
 converge uniformly on $Z\times L$,  defining holomorphic functions in $(\bm{\zeta},s)$ on this set.\\
 Note however, that we have, for $s\in L$,
 \begin{align*}
 \Lambda\left(\mathbf{q}(s)\right)=s,&& \mathbf{b}'\left(\mathbf{q}(s)\right)=\mathbf{b},
 \end{align*}
 and hence
 \begin{equation*}
 Y_i(\bm{\zeta},s)=y_i\left(\bm{\zeta};s,\mathbf{b}\right).
 \end{equation*}
 The theorem follows.
 \end{proof}
 \begin{remark}\label{remark:parameterregular}
 	In fact, the expansions \eqref{fgzetavar} also depend holomorphically on the parameters $\mathbf{b}$. That is, given $\mathbf{b}_0\in\mathcal{B}$ and $\Lambda_0\in L_0(\mathbf{b}_0)$, there exist open environments $Z\subseteq \mathbb{C}^2$, $L\subseteq \mathbb{C}$ and $B\subseteq\mathcal{B}$ of $\mathbf{0}$, $\Lambda_0$ and $\mathbf{b}_0$ respectively, such that for any $(\Lambda,\mathbf{b})\in L\times B$, we have $\Lambda\in L_0(\mathbf{b})$ and the series \eqref{fgzetavar} converge uniformly on $Z\times L\times B$, defining holomorphic functions on this set in $\left(\bm{\zeta},\Lambda,\mathbf{b}\right)$. This can be proven easily by incorporating parameters in Theorem \ref{genbriottheoremsevuniform}, see Remark \ref{remarkincorporatepar}.
 \end{remark}
As desired, we have 
\begin{align}
F_1^{0,+}(\phi;\Lambda,\mathbf{b})=F(\phi), && G_1^{0,+}(\phi;\Lambda,\mathbf{b})=G(\phi), \label{F0G0formal}
\end{align}
where $F$ and $G$ are defined as in equations \eqref{generalsolaut}.\\
Furthermore the coefficients $F_n^{0,+}(\phi;\Lambda,\mathbf{b})$ and $G_n^{0,+}(\phi;\Lambda,\mathbf{b})$ indeed satisfy equations \eqref{recursiongencoef}, which allows us to calculate them recursively.

In Theorem \ref{thmgeneralsol0+} the plus superscripts reflect the fact that there are only finitely many positive powers of $\phi$ occuring in the Laurent series \eqref{coefpowerexpansionsol0+}, we define the dual `minus' solutions as follows
\begin{subequations}
\label{minussol}
\begin{align}
f^{0,-}(t,\phi;\Lambda,\mathbf{b})&=f^{0,+}\left(t,\mu(\Lambda,\mathbf{b})\phi^{-1};\frac{b_1b_2b_3b_4}{\Lambda},\mathbf{b}\right),\\
g^{0,-}(t,\phi;\Lambda,\mathbf{b})&=g^{0,+}\left(t,\mu(\Lambda,\mathbf{b})\phi^{-1};\frac{b_1b_2b_3b_4}{\Lambda},\mathbf{b}\right).
\end{align}
\end{subequations}
Note that indeed, by Theorem \ref{thmgeneralsol0+}, this defines a formal solution to $q$-$P(A_1)$, as
\begin{equation*}
\overline{\mu(\Lambda,\mathbf{b})\phi^{-1}}=\frac{\left(\frac{b_1b_2b_3b_4}{\Lambda}\right)^2}{b_1b_2b_3b_4}\mu(\Lambda,\mathbf{b})\phi^{-1}.
\end{equation*}
Analogously to the expansions \eqref{powerexpansionsol0+} and \eqref{coefpowerexpansionsol0+}, we have
\begin{align}
f^{0,-}(t,\phi;\Lambda,\mathbf{b})=\sum_{n=1}^\infty{F_n^{0,-}(\phi;\Lambda,\mathbf{b})t^n}, && g^{0,-}(t,\phi;\Lambda,\mathbf{b})=\sum_{n=1}^\infty{G_n^{0,-}(\phi;\Lambda,\mathbf{b})t^n},\label{powerexpansionsol0-}
\end{align}
with, for $n\in\mathbb{N}^*$,
\begin{align}
F_n^{0,-}(\phi;\Lambda,\mathbf{b})=\sum_{i=-n}^\infty{F_{n,i}^{0,-}(\Lambda,\mathbf{b})\phi^i}, && G_n^{0,-}(\phi;\Lambda,\mathbf{b})=\sum_{i=-n}^\infty{G_{n,i}^{0,-}(\Lambda,\mathbf{b})\phi^i},\label{coefpowerexpansionsol0-}
\end{align}
where, for $i\in\mathbb{Z}_{\geq -n}$,
\begin{align*}
F_{n,i}^{0,-}(\Lambda,\mathbf{b})=F_{n,-i}^{0,+}\left(\frac{b_1b_2b_3b_4}{\Lambda},\mathbf{b}\right)\mu(\Lambda,\mathbf{b})^{-i},&&
G_{n,i}^{0,-}(\Lambda,\mathbf{b})=G_{n,-i}^{0,+}\left(\frac{b_1b_2b_3b_4}{\Lambda},\mathbf{b}\right)\mu(\Lambda,\mathbf{b})^{-i}.
\end{align*}
Using the symmetries
\begin{align*}
\mu\left(\frac{b_1b_2b_3b_4}{\Lambda}\right)=\mu(\Lambda,\mathbf{b}),&& F_{eq}\left(\frac{b_1b_2b_3b_4}{\Lambda}\right)=F_{\text{eq}}(\Lambda,\mathbf{b}), &&G_{eq}\left(\frac{b_1b_2b_3b_4}{\Lambda}\right)=G_{\text{eq}}(\Lambda,\mathbf{b}),
\end{align*}
it is easy to see that
\begin{align}
F_1^{0,-}(\phi;\Lambda,\mathbf{b})=F(\phi)=F_1^{0,+}(\phi;\Lambda,\mathbf{b}), && G_1^{0,-}(\phi;\Lambda,\mathbf{b})=G(\phi)=G_1^{0,+}(\phi;\Lambda,\mathbf{b}), \label{firstcoef+and-}
\end{align}
where $F(\phi)$ and $G(\phi)$ are as defined in equations \eqref{generalsolaut}.\\
Note that this implies that the coefficients of the formal `plus' and `minus' series solutions, \eqref{coefpowerexpansionsol0+} and \eqref{coefpowerexpansionsol0-}, satisfy the same recursive system of difference equations \eqref{recursiongencoef}, with the same initial values \eqref{firstcoef+and-}. Motivated by this plausibility argument, we formulate the following conjecture.
\begin{conj}
\label{conjecture}
The formal series solutions \eqref{powerexpansionsol0+} and \eqref{powerexpansionsol0-} are equal, that is,
\begin{align*}
f^{0,+}(t,\phi;\Lambda,\mathbf{b})=f^{0,-}(t,\phi;\Lambda,\mathbf{b}), && g^{0,+}(t,\phi;\Lambda,\mathbf{b})=g^{0,-}(t,\phi;\Lambda,\mathbf{b}),
\end{align*}
or equivalently, for $n\in\mathbb{N}^*$, the Laurent series \eqref{coefpowerexpansionsol0+} terminate at $i=-n$, that is,
\begin{align}
F_n^{0,+}(\phi;\Lambda,\mathbf{b})=\sum_{i=-n}^n{F_{n,i}^{0,+}(\Lambda,\mathbf{b})\phi^i}, && G_n^{0,+}(\phi;\Lambda,\mathbf{b})=\sum_{i=-n}^n{G_{n,i}^{0,+}(\Lambda,\mathbf{b})\phi^i}.\label{eq:terminate}
\end{align}
In particular, by equations \eqref{minussol}, we have, for $n\in\mathbb{N}^*$ and $i\leq n$,
\begin{subequations}
\label{conjidcoef}
\begin{align}
\mu(\Lambda,\mathbf{b})^iF_{n,i}^{0,+}\left(\frac{b_1b_2b_3b_4}{\Lambda},\mathbf{b}\right)&=F_{n,-i}^{0,+}(\phi;\Lambda,\mathbf{b}),\\
\mu(\Lambda,\mathbf{b})^iG_{n,i}^{0,+}\left(\frac{b_1b_2b_3b_4}{\Lambda},\mathbf{b}\right)&=G_{n,-i}^{0,+}(\phi;\Lambda,\mathbf{b}).
\end{align}
\end{subequations}
\end{conj}
Equations \eqref{F0G0formal} show that \eqref{eq:terminate} is true when $n=1$ and we have checked the case $n=2$ using Mathematica. As an additional check, Proposition \ref{prop:holopert} is consistent with equations \eqref{conjidcoef}.
\begin{remark}
\label{remarkrelax}
By equations \eqref{F0G0formal} and \eqref{generalsolaut}, we see that the coefficients $F_{1}^{0,+}$ and $G_{1}^{0,+}$  are only singular when $\Lambda^2=b_1b_2b_3b_4$. Reflecting on the proofs of Theorem \ref{thmgeneralsol0+} and \ref{genbriottheoremsev}, this implies that condition \eqref{Qcondition} in Theorem \ref{thmgeneralsol0+} can be relaxed to $1\notin Q_1$, where $Q_{1}\subseteq Q$ equals
\begin{equation*}
Q_{1}=\{q_1^mq_2^n:\text{$(m,n)\in\mathbb{N}^2\setminus\{0,0\}$ with $m\geq 1$}\}\cup\{q_2\}.
\end{equation*}
In particular, if $|q_2|=1$ with $q_2\neq 1$ and $|q_1|<1$, then the convergence of expansions \eqref{fgzetavar} still holds.
If Conjecture \ref{conjecture} is true, then condition \eqref{Qcondition} can be relaxed further, to $1\notin Q_{\text{rel}}$, where $Q_{\text{rel}}\subseteq Q$ is defined as
\begin{equation*}
Q_{\text{rel}}=\{q_1^mq_2^n:\text{$(m,n)\in\mathbb{N}^2\setminus\{0,0\}$ with $n\leq m+1$}\}.
\end{equation*}
\end{remark}

\subsection{Formal series solution at infinity}
The B\"acklund transformation $\mathcal{T}_4$ defined in \eqref{backlundt}, shows that the critical points $0$ and $\infty$ play an essentially equivalent role in $q$-$P(A_1)$. Using B\"acklund transformation $\mathcal{T}_2$ and Theorem \ref{thmgeneralsol0+}, it is easy to see that
\begin{align}
f&=tg^{0,+}\left(\frac{1}{t},\widetilde{\phi};\widetilde{\Lambda},\mathbf{b}^{(2)}\right),&
g&=tf^{0,+}\left(\frac{1}{t},\widetilde{\phi};\widetilde{\Lambda},\mathbf{b}^{(2)}\right),\label{solinf}
\end{align}
defines a formal solution to $q$-$P(A_1)(\mathbf{b})$, if $\widetilde{\Lambda}$ and $\widetilde{\phi}$ satisfy
\begin{align}
\overline{\widetilde{\Lambda}}=\widetilde{\Lambda},&& \overline{\widetilde{\phi}}=\widetilde{\lambda}^{-1}\widetilde{\phi},&& \widetilde{\lambda}=\frac{\widetilde{\Lambda}^2}{b_1^{(2)}b_2^{(2)}b_3^{(2)}b_4^{(2)}}. \label{relationstilde1}
\end{align}
We introduce formal variables $\Lambda_\infty$ and $\phi_\infty$ satisfying
\begin{align}
\overline{\Lambda}_\infty=\Lambda_\infty,&& \overline{\phi}_\infty=\lambda_\infty \phi_\infty, && \lambda_\infty=\frac{\Lambda_\infty^2}{b_5b_6b_7b_8},\label{formalinf}
\end{align}
and set
\begin{align*}
\widetilde{\Lambda}=\frac{1}{\Lambda_\infty}, && \widetilde{\phi}=\Lambda_\infty \phi_\infty.
\end{align*}
Observe that equations \eqref{relationstilde1} are satisfied, and upon substitution into equation \eqref{solinf}, we find that
\begin{subequations}
	\label{powerexpansionsinf+}
	\begin{align}
	f^{\infty,+}(t,\phi_\infty;\Lambda_\infty,\mathbf{b})&=tg^{0,+}\left(\frac{1}{t},\Lambda_\infty \phi_\infty;\frac{1}{\Lambda_\infty},\mathbf{b}^{(2)}\right),\\
	g^{\infty,+}(t,\phi_\infty;\Lambda_\infty,\mathbf{b})&=tf^{0,+}\left(\frac{1}{t},\Lambda_\infty \phi_\infty;\frac{1}{\Lambda_\infty},\mathbf{b}^{(2)}\right),
	\end{align}
\end{subequations}
defines a formal series solution to $q$-$P(A_1)(\mathbf{b})$ at $t=\infty$.\\
Indeed, expanding this solution in $t$ and $\phi_\infty$, we find
\begin{align*}
f^{\infty,+}(t,\phi_\infty;\Lambda_\infty,\mathbf{b})&=\sum_{n=0}^\infty{F_n^{\infty,+}(\phi_\infty;\Lambda_\infty,\mathbf{b})t^{-n}},\\
g^{\infty,+}(t,\phi_\infty;\Lambda_\infty,\mathbf{b})&=\sum_{n=0}^\infty{G_n^{\infty,+}(\phi_\infty;\Lambda_\infty,\mathbf{b})t^{-n}},
\end{align*}
with, for $n\in\mathbb{N}$,
\begin{align*}
F_n^{\infty,+}(\phi_\infty;\Lambda_\infty,\mathbf{b})=G_{n+1}^{0,+}\left(\Lambda_\infty\phi_\infty;\frac{1}{\Lambda_\infty},\mathbf{b}^{(2)}\right),&&
G_n^{\infty,+}(\phi_\infty;\Lambda_\infty,\mathbf{b})=F_{n+1}^{0,+}\left(\Lambda_\infty\phi_\infty;\frac{1}{\Lambda_\infty},\mathbf{b}^{(2)}\right),
\end{align*}
and hence
\begin{align*}
F_n^{\infty,+}(\phi_\infty;\Lambda_\infty,\mathbf{b})=\sum_{i=-\infty}^{n+1}{F_{n,i}^{\infty,+}(\Lambda_\infty,\mathbf{b})\phi_\infty^i},&&
G_n^{\infty,+}(\phi_\infty;\Lambda_\infty,\mathbf{b})=\sum_{i=-\infty}^{n+1}{G_{n,i}^{\infty,+}(\Lambda_\infty,\mathbf{b})\phi_\infty^i},
\end{align*}
where, for $n\in\mathbb{N}$ and $i\in\mathbb{Z}_{\leq n+1}$,
\begin{align*}
F_{n,i}^{\infty,+}(\Lambda_\infty,\mathbf{b})=\Lambda_\infty^i G_{n+1,i}^{0,+}\left(\frac{1}{\Lambda_\infty},\mathbf{b}^{(2)}\right),&& G_{n,i}^{\infty,+}(\Lambda_\infty,\mathbf{b})=\Lambda_\infty^i F_{n+1,i}^{0,+}\left(\frac{1}{\Lambda_\infty},\mathbf{b}^{(2)}\right).
\end{align*}
Of course we can formulate analogous convergence results to the ones in Theorem \ref{thmgeneralsol0+}. To obtain the dual `minus' solutions at infinity, we again take $\Lambda_\infty$ and $\phi_\infty$ satisfying equations \eqref{formalinf}, and set
\begin{align*}
\widetilde{\Lambda}=\frac{\Lambda_\infty}{b_5b_6b_7b_8}, && \widetilde{\phi}=\mu\left(\frac{\Lambda_\infty}{b_5b_6b_7b_8},\mathbf{b}^{(2)}\right) \frac{1}{\Lambda_\infty\phi_\infty}.
\end{align*}
in equations \eqref{solinf}.

\subsection{Symmetries of the formal series solution}
\label{sectionsymmetries}
In Appendix \ref{appback} we discuss several B\"acklund transformations of $q$-$P(A_1)$. Using these we can find symmetries of the formal series solution \eqref{powerexpansionsol0+}. We discuss $3$ such examples.
First of all, note that for any permutation
 \begin{equation}
 \sigma\in Sym(\{1,2,3,4\})\times Sym(\{5,6,7,8\}),\label{perm}
 \end{equation}
$q$-$P(A_1)$ is invariant under permutation of the parameters $\sigma(b)_i=b_{\sigma(i)}$ correspondingly for $1\leq i \leq 8$,
and using Theorem \ref{thmgeneralsol0+} we deduce
\begin{align}
f^{0,+}(t,\phi;\Lambda,\mathbf{b})=f^{0,+}\left(t,\phi;\Lambda,\sigma(\mathbf{b})\right), && g^{0,+}(t,\phi;\Lambda,\mathbf{b})=g^{0,+}\left(t,\phi;\Lambda,\sigma(\mathbf{b})\right).\label{sigmapermu}
\end{align}
Next, we would like to derive a symmetry of the formal series solution \eqref{powerexpansionsol0+} by application of B\"acklund transformation $\mathcal{T}_1$ as defined in \eqref{backlundt}. Consider formal variables $\phi$ and $\Lambda$ satisfying \eqref{indvariabledefi} and put
\begin{align*}
\widetilde{\phi}=\frac{1}{t\phi},&&\widetilde{\Lambda}=\frac{1}{\Lambda}.
\end{align*}
Then we have 
\begin{align*}
\overline{\widetilde{\Lambda}}=\widetilde{\Lambda},&&\overline{\widetilde{\phi}}=\frac{{\widetilde{\Lambda}}^2}{b_1^{(1)}b_2^{(1)}b_3^{(1)}b_4^{(1)}}\widetilde{\phi},
\end{align*}
and by Theorem \ref{thmgeneralsol0+} this implies that
\begin{align*}
f^{0,+}\left(t,\widetilde{\phi};\widetilde{\Lambda},\mathbf{b}^{(1)}\right)=f^{0,+}\left(t,\frac{1}{t\phi};\frac{1}{\Lambda},\mathbf{b}^{(1)}\right),&& g^{0,+}\left(t,\widetilde{\phi};\widetilde{\Lambda},\mathbf{b}^{(1)}\right)=g^{0,+}\left(t,\frac{1}{t\phi};\frac{1}{\Lambda},\mathbf{b}^{(1)}\right),
\end{align*}
defines a formal solution to $q$-$P(A_1)(\mathbf{b}^{(1)})$.\\
We apply B\"acklund transformation $\mathcal{T}_1$, which shows that
\begin{align}
f(t,\phi)=\frac{t}{f^{0,+}\left(t,\frac{1}{t\phi};\frac{1}{\Lambda},\mathbf{b}^{(1)}\right)}, && g(t,\phi)=\frac{t}{g^{0,+}\left(t,\frac{1}{t\phi};\frac{1}{\Lambda},\mathbf{b}^{(1)}\right)}, \label{fgnewsol}
\end{align}
defines a formal solution to $q$-$P(A_1)(\mathbf{b})$.\\
We expand this solution in powers of $t$ and $\phi$ and prove that it is exactly the formal series solution \eqref{powerexpansionsol0+}. First of all, for the denominators in \eqref{fgnewsol}, expanding in $t$ gives
\begin{align*}
f^{0,+}\left(t,\frac{1}{t\phi};\frac{1}{\Lambda},\mathbf{b}^{(1)}\right)=\sum_{m=0}^\infty{\widetilde{f}_m(\phi;\Lambda,\mathbf{b})t^m},\\
g^{0,+}\left(t,\frac{1}{t\phi};\frac{1}{\Lambda},\mathbf{b}^{(1)}\right)=\sum_{m=0}^\infty{\widetilde{g}_m(\phi;\Lambda,\mathbf{b})t^m},
\end{align*}
where, for $m\in\mathbb{N}$,
\begin{subequations}
	\label{tildefgcoef}
	\begin{align}
	\widetilde{f}_m(\phi;\Lambda,\mathbf{b})=\sum_{i=-\infty}^{m-1}{F_{m-i,-i}^{0,+}\left(\frac{1}{\Lambda},\mathbf{b}^{(1)}\right)\phi^i},&&
	\widetilde{g}_m(\phi;\Lambda,\mathbf{b})=\sum_{i=-\infty}^{m-1}{G_{m-i,-i}^{0,+}\left(\frac{1}{\Lambda},\mathbf{b}^{(1)}\right)\phi^i}.
	\end{align}
\end{subequations}
We can hence expand equations \eqref{fgnewsol} in $t$, using for instance the Lagrange inversion formula, to obtain
\begin{subequations}
	\label{fgexpansiont}
	\begin{align}
	f(t,\phi)&=\frac{1}{\widetilde{f}_0(\phi;\Lambda,\mathbf{b})} t-\frac{\widetilde{f}_1(\phi;\Lambda,\mathbf{b})}{\widetilde{f}_0(\phi;\Lambda,\mathbf{b})^2}t^2+\ldots,\\
	g(t,\phi)&=\frac{1}{\widetilde{g}_0(\phi;\Lambda,\mathbf{b})} t-\frac{\widetilde{g}_1(\phi;\Lambda,\mathbf{b})}{\widetilde{g}_0(\phi;\Lambda,\mathbf{b})^2}t^2+\ldots,
	\end{align}
\end{subequations}
and compare the result with the formal series solution \eqref{powerexpansionsol0+}.\\
Indeed, by expanding the coefficients of the series \eqref{fgexpansiont} with respect to $\phi$, we see that they are of exactly the same form as solutions \eqref{powerexpansionsol0+}, that is, we can find $\widetilde{F}_{n,i}$ and $\widetilde{G}_{n,i}$ for $i\in\mathbb{N}_{\leq n}$ and $n\in\mathbb{N}^*$ such that
\begin{align*}
f(t,\phi)=\sum_{n=1}^\infty{\sum_{i=-\infty}^n{ \widetilde{F}_{n,i} t^n\phi^i}},&&
g(t,\phi)=\sum_{n=1}^\infty{\sum_{i=-\infty}^n{ \widetilde{G}_{n,i} t^n\phi^i}}.
\end{align*}
In particular, calculating $\widetilde{F}_{1,1}$ and $\widetilde{G}_{1,1}$ gives
\begin{align*}
\widetilde{F}_{1,1}=\frac{1}{F_{1,1}^{0,+}\left(\frac{1}{\Lambda},\mathbf{b}^{(1)}\right)}=1,&&
\widetilde{G}_{1,1}=\frac{1}{G_{1,1}^{0,+}\left(\frac{1}{\Lambda},\mathbf{b}^{(1)}\right)}=\Lambda.
\end{align*}
Therefore, by the uniqueness property of the formal series solution \eqref{powerexpansionsol0+} in Theorem \ref{thmgeneralsol0+}, we have
\begin{align*}
f(t,\phi)=f^{0,+}(t,\phi;\Lambda,\mathbf{b}), && g(t,\phi)=g^{0,+}(t,\phi;\Lambda,\mathbf{b}),
\end{align*}
and hence, by the definition of $f$ and $g$ \eqref{fgnewsol}, we obtain the formal identities
\begin{align}
f^{0,+}(t,\phi;\Lambda,\mathbf{b})f^{0,+}\left(t,\frac{1}{t\phi};\frac{1}{\Lambda},\mathbf{b}^{(1)}\right)=t,&&
g^{0,+}(t,\phi;\Lambda,\mathbf{b})g^{0,+}\left(t,\frac{1}{t\phi};\frac{1}{\Lambda},\mathbf{b}^{(1)}\right)=t. \label{formalid}
\end{align}
These equations induce a countable number of identities among the coefficients, each one given by comparing the coefficients of a positive power of $t$. In particular, comparing the coefficients of the lowest order term $t$, we obtain
\begin{align}
F_1^{0,+}(\phi;\Lambda,\mathbf{b})\widetilde{f}_0\left(\phi;\Lambda,\mathbf{b}\right)=1,&& G_1^{0,+}(\phi;\Lambda,\mathbf{b})\widetilde{g}_0\left(\phi;\Lambda,\mathbf{b}\right)=1. \label{formalid1}
\end{align}
Combining this identity with equations \eqref{F0G0formal}, \eqref{autonomoufg} and \eqref{tildefgcoef}, we find generating functions,
\begin{subequations}
	\label{genfunc}
\begin{gather}
\frac{x}{1+F_{\text{eq}}(\Lambda,\mathbf{b})x+\mu(\Lambda,\mathbf{b})x^2}=\sum_{i=1}^\infty{F_{i,i}^{0,+}\left(\frac{1}{\Lambda},\mathbf{b}^{(1)}\right)x^i},\\
\frac{x}{\Lambda+G_{\text{eq}}(\Lambda,\mathbf{b})x+\frac{b_1b_2b_3b_4}{\Lambda}\mu(\Lambda,\mathbf{b})x^2}=\sum_{i=1}^\infty{G_{i,i}^{0,+}\left(\frac{1}{\Lambda},\mathbf{b}^{(1)}\right)x^i}.
\end{gather}
\end{subequations}

Similarly, using B\"acklund transformation $\mathcal{T}_3$, we find formal identities
\begin{align}
f^{0,+}(t,\phi;\Lambda,\mathbf{b})&=g^{0,+}\left(q^{-\tfrac{1}{2}}t,q^{-\tfrac{1}{2}}\frac{b_1b_2b_3b_4}{\Lambda}\phi;\frac{\Lambda}{b_5b_6b_7b_8},\mathbf{b}^{(3)}\right),\label{symmetryfg1}\\
g^{0,+}(t,\phi;\Lambda,\mathbf{b})&=f^{0,+}\left(q^{-\tfrac{1}{2}}t,q^{\tfrac{1}{2}}\Lambda\phi;\frac{\Lambda}{b_5b_6b_7b_8},\mathbf{b}^{(3)}\right).
\end{align}
These equations plays an important role in the symmetrisation of $q$-$P(A_1)$, as described in Section \ref{sectionsym}.

\section{Constructing true solutions to $q$-$P(A_1)$}
\label{sectiontruesol}
In this section we use the formal series solution \eqref{powerexpansionsol0+} to construct true solutions of $q$-$P(A_1)$. The idea is relatively straightforward, we replace the formal variables $\Lambda$ and $\phi$ with actual functions satisfying equations \eqref{indvariabledefi}. Let us take any choice of parameter values $\mathbf{b}\in\mathcal{B}$, define $q$ by equation \eqref{eqconstraint} and assume $|q|<1$, or equivalently,
\begin{equation*}
|b_1b_2b_3b_4|<|b_5b_6b_7b_8|.
\end{equation*}
We first discuss how to find solutions on a discrete time domain $q^\mathbb{Z}t_0$. More precisely, we adopt the following discrete time interpretation of $q$-$P(A_1)$, we fix a $t_0\in\mathbb{C}^*$ and define
\begin{align*}
t_s=q^s t_0, && f_s=f(t_s), && g_s=g(t_s),
\end{align*}
then we have, for $s\in\mathbb{Z}$,
\begin{subequations}
	\label{qpa1discrete}
\begin{align}
\frac{(g_sf_s-t_s^2)(g_sf_{s+1}-qt_s^2)}{(g_sf_s-1)(g_sf_{s+1}-1)}&=\frac{(g_s-b_1  t_s)(g_s-b_2  t_s)(g_s-b_3  t_s)(g_s-b_4  t_s)}{(g_s-b_5)(g_s-b_6)(g_s-b_7)(g_s-b_8)},\\
\frac{(g_sf_{s+1}-qt_s^2)(g_{s+1}f_{s+1}-q^2t_{s})}{(g_sf_{s+1}-1)(g_{s+1}f_{s+1}-1)}&=\frac{(f_{s+1}- b_1^{-1}qt_{s})(f_{s+1}- b_2^{-1}qt_{s})(f_{s+1}- b_3^{-1} qt_{s})(f_{s+1}-  b_4^{-1}qt_{s})}{(f_{s+1}-b_5^{-1})(f_{s+1}-b_6^{-1})(f_{s+1}-b_7^{-1})(f_{s+1}-b_8^{-1})}.
\end{align}
\end{subequations}
We remark that Grammaticos and Ramani \cite{firstqp6} initially introduced $q$-$P(A_1)$ in this form.
In this setting, we give meaning to the autonomous equations \eqref{indvariabledefi}, by interpreting them as follows,
\begin{align}
\Lambda_{s+1}=\Lambda_s,&& \phi_{s+1}=\lambda_s \phi_s,&& \lambda_s=\frac{\Lambda_s^2}{b_1b_2b_3b_4}. \label{discreteinter}
\end{align}
Let us take any $\phi_0\in\mathbb{C}^*$ and $\Lambda_0\in L_0(\mathbf{b})$, as defined in \eqref{eq:L0defi}.
In accordance with Theorem \ref{thmgeneralsol0+} and equations \eqref{discreteinter}, we put
\begin{align*}
\lambda_0&=\frac{\Lambda_0^2}{b_1b_2b_3b_4}, & q_1&=q\lambda_0, & q_2=\lambda_0^{-1},
\end{align*} 
and define, for $s\in\mathbb{Z}$,
\begin{align*}
t_s&=q^st_0, & \phi_s&=\lambda_0^s \phi_0,& (\zeta_1)_s&=q_1^s \phi_0t_0,& (\zeta_2)_s&=q_2^s \phi_0^{-1}.
\end{align*}
As $\Lambda_0\in L_0(\mathbf{b})$, Theorem \ref{thmgeneralsol0+} shows that there is an $r>0$ such that the expansions \eqref{fgzetavar} with $\Lambda=\Lambda_0$ converge for all $(\zeta_1,\zeta_2)\in\mathbb{C}^2$ with $|\zeta_1|,|\zeta_2|<r$.
Note that $0<|q_1|,|q_2|<1$ and determine an $S\in\mathbb{Z}$, such that, for all $s\geq S$,
\begin{equation*}
|(\zeta_1)_s|,|(\zeta_2)_s|<r.
\end{equation*}
Then we know, that for all $s\geq S$,
\begin{align*}
f_s&=f^{0,+}(t_s,\phi_s;\Lambda_0,\mathbf{b})=\sum_{n=1}^\infty{\sum_{m=0}^\infty{F_{n,n-m}^{0,+}(\Lambda_0,\mathbf{b})(\zeta_1)_s^n(\zeta_2)_s^m}},\\
g_s&=g^{0,+}(t_s,\phi_s;\Lambda_0,\mathbf{b})=\sum_{n=1}^\infty{\sum_{m=0}^\infty{G_{n,n-m}^{0,+}(\Lambda_0,\mathbf{b})(\zeta_1)_s^n(\zeta_2)_s^m}},
\end{align*}
are well-defined, and converge uniformly in $s$ on $\mathbb{Z}_{\geq S}$, defining a solution of \eqref{qpa1discrete}.\\
Grammaticos and Ramani \cite{firstqp6} showed that \eqref{qpa1discrete} has the singularity confinement property, which allows us to extend $(f_s,g_s)_{s\geq S}$ to a full solution $(f_s,g_s)_{s\in\mathbb{Z}}$, with $f_s,g_s\in \mathbb{C}_\infty$ for $s\in \mathbb{Z}$, where $\mathbb{C}_\infty$ denotes the Riemann sphere. Note that this solution $(f_s,g_s)_{s\in\mathbb{Z}}$ is completely determined by our initial choices for $\Lambda_0$ and $\phi_0$, that is, writing
\begin{equation*}
(f_s,g_s)_{s\in\mathbb{Z}}=(f_s(\Lambda_0,\phi_0),g_s(\Lambda_0,\phi_0))_{s\in\mathbb{Z}},
\end{equation*} 
we found a family of solutions to $q$-$P(A_1)$ with discrete time $t=q^s t_0$, i.e. equations \eqref{qpa1discrete}, with two arbitrary integration constants $\Lambda_0\in L_0(\mathbf{b})$ and $\phi_0\in\mathbb{C}^*$.

\subsection{Constructing analytic solutions}
To construct solution with continuous time $t$, we replace the formal variables $\phi$ and $\Lambda$ in the formal series solution \eqref{powerexpansionsol0+} by analytic functions satisfying equations \eqref{indvariabledefi}, as done in the following theorem.
\begin{thm}
	\label{thmtruesol}
	Let $\mathbf{b}\in\mathcal{B}$, define $q$ by \eqref{eqconstraint} and assume $|q|<1$.
	Suppose we have a nonempty open set $T\subseteq\mathbb{C}^*$ with $qT=T$, a function $\Lambda(t)$ which is holomorphic on $T$ and $q$-periodic, i.e. $\Lambda(qt)=\Lambda(t)$, satisfying inequalities 
	\begin{equation}
	\label{condLambda}
	|b_1b_2b_3b_4|<|\Lambda(t)|^2<|b_5b_6b_7b_8|,
	\end{equation}
	for $t\in T$.\\
	Assume $\phi(t)$ is a holomorphic nonvanishing function with, for $t\in T$,
	\begin{align}
	\label{condphi}
	\phi(qt)=\lambda(t)\phi(t),&& \lambda(t):=\frac{\Lambda(t)^2}{b_1b_2b_3b_4}.
	\end{align}
	Then there is an unique meromorphic solution $(f(t),g(t))$ of $q$-$P(A_1)$ on $T$, characteristed by the fact that, for every $V\subseteq T$ open with $\overline{V}^*\subseteq T$ and $qV=V$, there is an $r>0$, such that the series
	\begin{subequations}
		\label{truesolconv}
		\begin{align}
		f^{0,+}(t,\phi(t);\Lambda(t),\mathbf{b})&=\sum_{n=1}^\infty{\sum_{i=-\infty}^n{F_{n,i}^{0,+}(\Lambda(t),\mathbf{b})t^n\phi(t)^i}},\\ g^{0,+}(t,\phi(t);\Lambda(t),\mathbf{b})&=\sum_{n=1}^\infty{\sum_{i=-\infty}^n{G_{n,i}^{0,+}(\Lambda(t),\mathbf{b})t^n\phi(t)^i}},
		\end{align}
	\end{subequations}
	converge uniformly on
	\begin{equation*}
	V\cap \{t\in\mathbb{C}^*:|t|<r\}, \label{Vdomain}
	\end{equation*}
	and we have $f(t)=f^{0,+}(t,\phi(t);\Lambda(t),\mathbf{b})$ and $g(t)=g^{0,+}(t,\phi(t);\Lambda(t),\mathbf{b})$ on this set.\\
	In particular the leading order behaviour of this solution on $V$ as $t\rightarrow 0$, is given by
	\begin{align}
	f(t)\sim \phi(t)t,&& g(t)\sim \Lambda(t)\phi(t)t. \label{leadingordertrue}
	\end{align}
\end{thm}
\begin{proof}
	Let us take any nonempty open set $V\subseteq T$ with $qV=V$, such that $\overline{V}^*\subseteq T$.
	We define 
	\begin{equation*}
	V_{\text{ann}}=\overline{V}\cap \{t\in \mathbb{C}:1\leq |t|\leq |q|^{-1}\},
	\end{equation*}
	then $V_{\text{ann}}$ is a compact subset of $T$.\\
	We define
	\begin{align*}
	\lambda^+=\sup_{t\in V_{\text{ann}}}{|\lambda(t)|},&& \lambda^-=\inf_{t\in V_{\text{ann}}}{|\lambda(t)|},&& \phi^+=\sup_{t\in V_{\text{ann}}}{|\phi(t)|}, &&\phi^-=\inf_{t\in V_{\text{ann}}}{|\phi(t)|},
	\end{align*}
	then we have, by inequalities \eqref{condLambda},
	\begin{equation}
	\label{lambdaineq}
	1<\lambda^-\leq \lambda^+<|q|^{-1}.
	\end{equation}
	By equation \eqref{condphi}, we obtain, 
	\begin{equation*}
	(\lambda^-)^{\log_{|q|}(|t|)}\phi^-\leq |\phi(t)|\leq (\lambda^+)^{\log_{|q|}(|t|)+1}\phi^+,
	\end{equation*}
	for all $t\in V$.\\
	Let us introduce the variables 
	\begin{align}
	\zeta_1(t)=t\phi(t),&&\zeta_2(t)=\phi(t)^{-1},\label{zetatdefi}
	\end{align}
	then we have inequalities
	\begin{subequations}
		\label{ineqzeta}
		\begin{align}
		|\zeta_1(t)|&\leq \lambda^+\phi^+(|q|\lambda^+)^{\log_{|q|}(|t|)},\\
		|\zeta_2(t)|&\leq \left(\phi^-\right)^{-1}(\lambda^-)^{-\log_{|q|}(|t|)},
		\end{align}
	\end{subequations}
	for $t\in V$.\\
	Determine $L\subseteq L_0(\mathbf{b})$ open with $\overline{L}\subseteq L_0(\mathbf{b})$, such that
	\begin{equation*}
	\{\Lambda(t):t\in V\}\subseteq L.
	\end{equation*}
	By Theorem \ref{thmgeneralsol0+}, there is an open environment $Z\subseteq \mathbb{C}^2$ of $\mathbf{0}$, such that the series \eqref{fgzetavar} converge uniformly in $(\bm{\zeta},\Lambda)$ on $Z\times L$. By inequalities \eqref{lambdaineq} and \eqref{ineqzeta}, we can determine an $r>0$ such that $\bm{\zeta}(t)\in Z$ for $|t|<r$. It follows that the series \eqref{truesolconv} converge uniformly on \eqref{Vdomain}, defining analytic solutions of $q$-$P(A_1)$ on this set.
	Of course, we can now use the $q$-$P(A_1)$ equation, to analytically extend the domain of the solution $(f_V(t),g_V(t))$ to $V$, giving a uni-valued meromorphic solution on $V$. Indeed, by rewriting $q$-$P(A_1)$ as
	\begin{align*}
	f=H_1\left(t,\overline{f},g\right), &&g=H_2\left(t,\overline{f},\overline{g}\right),
	\end{align*}
	and setting $f=f_V(t)$ and $g=g_V(t)$ in the above equations, analytic continuation to $V$ is easily obtained.\\
	As we can do so for any open set $V\subseteq T$ with $\overline{V}^*\subseteq T$ and $qV=V$, we take the union of these solutions $f_V(t)$ and $g_V(t)$, giving an unique meromorphic solution $(f(t),g(t))$ of $q$-$P(A_1)$ on $T$.
	\end{proof}

\begin{remark}
We note that if Conjecture \ref{conjecture} is true, then condition \eqref{condLambda} can be relaxed substantially. Indeed note that condition \eqref{condLambda} is equivalent to  $1<|\lambda(t)|<|q|^{-1}$ on $T$. If Conjecture \ref{conjecture} is true, then we can relax this condition to $|q|<|\lambda(t)|<|q|^{-1}$ and $\lambda(t)\neq 1$ on $T$. Indeed, let us set
\begin{align*}
T_+=\{t\in T: |\lambda(t)|>1\}, && T_0=\{t\in T: |\lambda(t)|=1\}, && T_-=\{t\in T: |\lambda(t)|<1\}.
\end{align*}
Using Theorem \ref{thmtruesol} we can construct a solution $(f_+(t),g_+(t))$ on $T_+$. Using a dual result of Theorem \ref{thmtruesol} for the formal `minus' solutions \eqref{minussol}, we construct a solution $(f_-(t),g_-(t))$ on $T_-$. If Conjecture \ref{conjecture} is true, then these solutions should match on $T_0$, giving a global solution $(f(t),g(t))$ on $T$. Passing through the set $T_0=\{t\in T: |\lambda(t)|=1\}$, an interesting transition in leading order behaviour near $t=0$ is expected, see in particular equations \eqref{oscasympdamp}. 
\end{remark}
Note that, upon fixing the domain $T$, the solution defined in the above theorem contains essentially two free parameters, the functions $\Lambda(t)$ and $\phi(t)$. After a proper choice for $\Lambda(t)$, we can take any $\phi(t)$ which satisfies the $q$-difference equation \eqref{condphi} and the other required conditions in Theorem \ref{thmtruesol}. As the $q$-difference equation \eqref{condphi} determines $\phi(t)$ upto a $q$-periodic function, this implies that we essentially have the freedom of choosing two arbitrary $q$-periodic functions for the solution defined in Theorem \ref{thmtruesol}, which is precisely what one would expect for a (locally) general solution to a second order $q$-discrete equation. This however, does not imply that Theorem \ref{thmtruesol} contains all possible critical behaviours at $t=0$ for solutions of $q$-$P(A_1)$. Indeed the holomorphic solutions, defined in Proposition \ref{merosolution0}, are not captured directly by Theorem \ref{thmtruesol}.

Allthough holomorphicity seems natural in Theorem \ref{thmtruesol}, one could leave it away or replace it for, for instance, continuity, to obtain a broader class of solutions.
Note that we can also formulate a real version of Theorem \ref{thmtruesol}. That is, assume that the parameters $\mathbf{b}$ satisfy
\begin{align*}
\overline{\{b_1,b_2,b_3,b_4\}}=\{b_1,b_2,b_3,b_4\},&& \overline{\{b_5,b_6,b_7,b_8\}}=\{b_5,b_6,b_7,b_8\},
\end{align*}
with $q(\mathbf{b})\in\left(0,1\right)$.\\
We let $\phi(t)$ and $\Lambda(t)$ be nonzero real-valued continuous functions satisfying \eqref{condLambda} and \eqref{condphi} on $\mathbb{R}_+$. Then there is an unique real-valued piecewise continuous solution $(f(t),g(t))$ of $q$-$P(A_1)$ on $\mathbb{R}_{+}$, such that there is an $r>0$ such that the series expansions \eqref{truesolconv} converge uniformly on $(0,r)$ and $f(t)=f^{0,+}(t,\phi(t);\Lambda(t),\mathbf{b})$ and $g(t)=g^{0,+}(t,\phi(t);\Lambda(t),\mathbf{b})$ on $(0,r)$. Setting $f(0)=g(0)=0$ the obtained solution $(f(t),g(t))$ is continuous on $[0,r)$.

\subsection{Complex power-type series solutions}
Let us discuss a special case of Theorem \ref{thmtruesol}. We take $\Lambda(t)$ constant, so we take a $\Lambda\in\mathbb{C}$ satisfying
\begin{equation*}
|b_1b_2b_3b_4|<|\Lambda|^2<|b_5b_6b_7b_8|,
\end{equation*}
and set $\Lambda(t)=\Lambda$.\\
We define $\lambda=\frac{\Lambda^2}{b_1b_2b_3b_4}\in \mathbb{C}$ and determine a $\rho\in\mathbb{C}$ such that $q^\rho=\lambda$.
We choose a $\phi_0\in\mathbb{C}^*$ and set $\phi(t)=\phi_0 t^\rho$. As $\rho\notin \mathbb{Z}$, we have to impose a branchcut on the domain $T\subseteq\mathbb{C}^*$, and in order to meet the requirement $qT=T$, we set this branchcut equal to a (continuous) $q$-spiral. That is, we choose a $\theta_b\in\mathbb{R}$, and set
\begin{equation}
\label{Tdefi}
T=\mathbb{C}^*\setminus \left\{e^{\theta_bi}q^s:s\in\mathbb{R}\right\}. 
\end{equation}
Then we can define the complex exponential $t^\rho$ uni-valued on $T$, such that
\begin{equation*}
\phi(qt)=\lambda \phi(t),
\end{equation*}
holds for all $t\in T$.\\
More explicitly, we could for instance define $t^\rho$ on $T$ as follows. Let $t\in T$, then there is an unique $s\in\mathbb{R}$ and an unique $\theta\in (\theta_b,\theta_b+2\pi)$ such that $t=q^se^{i\theta}$, and we define
\begin{equation}
\label{tpowerrho}
t^\rho=\lambda^se^{-\Im{(\rho)}\theta} e^{\Re{(\rho)}\theta i}.
\end{equation}
Applying Theorem \ref{thmtruesol} gives us an unique meromorphic solution $(f(t),g(t))$ of $q$-$P(A_1)$ on $T$, such that $f(t)$ and $g(t)$ are described by
\begin{subequations}
\label{specialcaseexp}
\begin{align}
f(t)&=f^{0,+}\left(t,\phi_0 t^\rho;q^\rho,\mathbf{b}\right)=\sum_{n=1}^\infty\sum_{i=-\infty}^n{F_{n,i}^{0,+}\left(q^\rho,\mathbf{b}\right)\phi_0^it^{\rho i+n}},\\ 
g(t)&=g^{0,+}\left(t,\phi_0 t^\rho;q^\rho,\mathbf{b}\right)=\sum_{n=1}^\infty\sum_{i=-\infty}^n{G_{n,i}^{0,+}\left(q^\rho,\mathbf{b}\right)\phi_0^it^{\rho i+n}},
\end{align}
\end{subequations}
for $t$ close to $0$, on every open set $V\subseteq T$ with $\overline{V}^*\subseteq T$ and $qV=V$.\\
However, since $\lambda$ is constant, it might come as no surprise that we actually have global uniform convergence of the expansions \eqref{specialcaseexp} on $T$ intersected with a neighbourhood of the origin. That is, there is an $r>0$, such that the expansions \eqref{specialcaseexp} converge uniformly on 
\begin{equation*}
T\cap \{t\in\mathbb{C}^*:|t|<r\}.
\end{equation*}
An interesting special case occurs when $\lambda\in q^\mathbb{R}$, then we can choose $\rho\in\mathbb{R}$, and as $1<|\lambda|<|q|^{-1}$, we have $-1<\rho<0$.
In Section \ref{solcontinuumlimit} we identify the solutions \eqref{specialcaseexp} with the solutions of Painlev\'e VI found by Jimbo \cite{Jimbo} in the continuum limit $q\rightarrow 1$.

\subsection{Oscillatory-type solutions}
Another case of special interest is given by setting $\lambda=e^{\theta i}$ in Theorem \ref{thmtruesol}, where $\theta\in \mathbb{R}\setminus 2\pi \mathbb{Z}$. Indeed, by Remark \ref{remarkrelax}, the expansions \eqref{powerexpansionsol0+} are well-defined in this case, and converge. This gives rise to solutions of $q$-$P(A_1)$ with leading order behaviour of damped oscillatory-type. Indeed, let $|q|<1$ as before, given a nonempty open set $T\subseteq\mathbb{C}^*$ with $qT=T$ and a nonvanishing function $\phi(t)$ satisfying $\phi(qt)=e^{\theta i}\phi(t)$ on $T$, setting 
\begin{equation*}
\Lambda(t)=\Lambda=\pm (b_1b_2b_3b_4)^{\tfrac{1}{2}}e^{\tfrac{1}{2}\theta i},
\end{equation*}
we can construct an unique meromorphic solution $(f(t),g(t))$ of $q$-$P(A_1)$ on $T$, such that,
for every $V\subseteq T$ open with $\overline{V}^*\subseteq T$ and $qV=V$, there is an $r>0$, such that the series \eqref{truesolconv} converges converge uniformly in $t$ on
\begin{equation*}
V\cap \{t\in\mathbb{C}^*:|t|<r\},
\end{equation*}
and we have $f(t)=f^{0,+}(t,\phi(t);\Lambda,\mathbf{b})$ and $g(t)=g^{0,+}(t,\phi(t);\Lambda,\mathbf{b})$ on this set.\\
The leading order behaviour of this solution is given by
\begin{subequations}
\label{oscasympdamp}
\begin{align}
f(t)&=t\left(\phi(t)+F_{\text{eq}}(\Lambda,\mathbf{b})+\mu(\Lambda,\mathbf{b})\phi(t)^{-1}\right)+\mathcal{O}\left(t^2\right),\\ 
g(t)&=t\left(\Lambda \phi(t)+G_{\text{eq}}(\Lambda,\mathbf{b})+\frac{b_1b_2b_3b_4}{\Lambda}\mu(\Lambda,\mathbf{b})\phi(t)^{-1}\right)+\mathcal{O}\left(t^2\right),
\end{align}
\end{subequations}
as $t\rightarrow 0$ in $V$ as above.\\
We are tempted to call these damped oscillatory-type solutions. Indeed, if $\theta\in \mathbb{Q}$, then $\phi(t)$ is periodic, leading to a vast number of possible damped oscillatory-type asymptotics in equations \eqref{oscasymp}, for different choices of $\phi(t)$. On the other hand, for any $\phi_0\in\mathbb{C}^*$, setting $\phi(t)=\phi_0t^{i\rho}$, where $\rho=\theta\log{(q)}^{-1}$, we have $\phi(q^s t)=\phi(t)$, where $s=\frac{2\pi}{\theta}\in\mathbb{R}$, which gives damped oscillatory-type asymptotics in equations \eqref{oscasympdamp} on $q$-spirals as well. Heuristically speaking, the latter solutions are related to the damped oscillatory-type solutions Guzzetti \cite{Guzzettielliptic} obtained for the continuous sixth Painlev\'e equation \eqref{pVI}, via the continuum limit discussed in Section \ref{sectioncontinuum}.\\
Let us get back to the general case. Applying B\"acklund Transformation $\mathcal{T}_1$ to the solution $(f(t),g(t))$, we see that
\begin{align*}
\widetilde{f}(t)=\frac{t}{f(t)}, && \widetilde{g}(t)=\frac{t}{g(t)},
\end{align*}
defines a meromorphic solution of $q$-$P(A_1)(\mathbf{b}^{(1)})$ on $T$.\\
Note that the leading order terms in equations \eqref{oscasympdamp} can not vanish identically on a non-empty open subset of $T$.
For any set $V\subseteq T$ satisfying $\overline{V}^*\subseteq T$ and $0\in\overline{V}$, such that $\overline{V}^*$ does not contain any zeros of the denominators appearing in equations \eqref{oscasymp}, we have
\begin{subequations}
\label{oscasymp}
\begin{align}
\widetilde{f}(t)&=\frac{1}{\phi(t)+F_{\text{eq}}(\Lambda,\mathbf{b})+\mu(\Lambda,\mathbf{b})\phi(t)^{-1}}+\mathcal{O}\left(t\right),\\ 
\widetilde{g}(t)&=\frac{1}{\Lambda \phi(t)+G_{\text{eq}}(\Lambda,\mathbf{b})+\frac{b_1b_2b_3b_4}{\Lambda}\mu(\Lambda,\mathbf{b})\phi(t)^{-1}}+\mathcal{O}\left(t\right),
\end{align}
\end{subequations}
as $t\rightarrow 0$ on $V$.\\
We are tempted to call these undamped oscillatory-type solutions. Note however, that for a bad choice of $\phi(t)$, the poles of $\widetilde{f}(t)$ and $\widetilde{g}(t)$ in $T$ might accumulate at $t=0$ in a rather unpleasant way. We easily circumvent this problem by demanding $\phi(t)$ to be such that the denominators appearing \eqref{oscasymp} do not vanish on $T$, for $t$ close enough to $0$. 
We remark that Guzzetti \cite{Guzzettielliptic} also obtained undamped oscillatory-type solutions to the continuous sixth Painlev\'e equation \eqref{pVI}.

\subsection{Asymptotics of true solutions}
Note that the asymptotics \eqref{leadingordertrue} are useful for identifying $\Lambda(t)$ and $\phi(t)$ for a given numerical solution of $q$-$P(A_1)$. That is, suppose one has a solution $(f(t),g(t))$ on an open domain $T\subseteq\mathbb{C}^*$ with $qT=T$. By numerically comparing the leading order behaviour of $f(t)$ and $g(t)$ as $t\rightarrow 0$ in $T$ with equations \eqref{leadingordertrue}, one can determine whether this solution can be described by equations \eqref{truesolconv} and determine the associated $\Lambda(t)$ and $\phi(t)$ numerically. The other way around, given $\Lambda(t)$ and $\phi(t)$, we are interested in obtaining numerics of the via Theorem \ref{thmtruesol} associated solution. So suppose we are given $\Lambda(t)$ and $\phi(t)$ on some open domain $T\subseteq\mathbb{C}^*$ with $qT=T$ satisfying the necessary conditions in Theorem \ref{thmtruesol}. Let $V\subseteq T$ be open with $\overline{V}^*\subseteq T$ and $qV=V$, then there is an $r>0$ such that the series \eqref{truesolconv} converge uniformly to $(f(t),g(t))$ on 
\begin{equation*}
V\cap \{t\in\mathbb{C}^*:|t|<r\}.
\end{equation*}
From a theoretical point of view, this means that we could obtain arbitrary accurate numerics of the solution $(f(t),g(t))$ on this set by calculating sufficiently many coefficients in the series \eqref{truesolconv}. However, in practice these coefficients seem to be quite hard to calculate for large $n$. As an example, writing the coefficients $F_{2,0}^{0,+}(\Lambda,\mathbf{b})$ and $G_{2,0}^{0,+}(\Lambda,\mathbf{b})$ of expansions \eqref{coefpowerexpansionsol0+}, down explicitly as a ratio of polynomials in $\Lambda$ and $b_1,\ldots, b_8$ already requires a couple of pages. Despite this drawback, note that we have $1\ll \phi(t)\ll t^{-1}$ on $V$ as $t\rightarrow 0$.  Therefore
\begin{align*}
\phi(t)^{-1}\ll t, && \phi(t)^it^n\ll t,
\end{align*}
for $i\in\mathbb{Z}_{<n}$ and $n\geq 2$, as $t\rightarrow 0$ in $V$.\\
Therefore by Theorem \ref{thmtruesol} and equations \eqref{truesolconv}, we have
\begin{align*}
f(t)&=f^{0,+}(t,\phi(t);\Lambda(t),\mathbf{b})=F_{1,0}^{0,+}(\Lambda(t),\mathbf{b})t+\sum_{n=1}^\infty{{F_{n,n}^{0,+}(\Lambda(t),\mathbf{b})t^n\phi(t)^n}}+o(t),\\
g(t)&=g^{0,+}(t,\phi(t);\Lambda(t),\mathbf{b})=G_{1,0}^{0,+}(\Lambda(t),\mathbf{b})t+\sum_{n=1}^\infty{{G_{n,n}^{0,+}(\Lambda(t),\mathbf{b})t^n\phi(t)^n}}+o(t),
\end{align*}
as $t\rightarrow 0$ in $V$.\\
And hence, using equations \eqref{genfunc}, we obtain
\begin{align*}
f(t)&=F_{1,0}^{0,+}(\Lambda(t),\mathbf{b})t+\frac{1}{F_1^{0,+}\left(t^{-1}\phi(t)^{-1},\Lambda(t)^{-1},\mathbf{b}^{(1)}\right)}+o(t),\\
g(t)&=G_{1,0}^{0,+}(\Lambda(t),\mathbf{b})t+\frac{1}{G_1^{0,+}\left(t^{-1}\phi(t)^{-1},\Lambda(t)^{-1},\mathbf{b}^{(1)}\right)}+o(t),
\end{align*}
as $t\rightarrow 0$ in $V$, and explicit formulas for $F_1^{0,+}$ and $G_1^{0,+}$ are given by equations \eqref{generalsolaut} and \eqref{F0G0formal}.\\
These closed formulas seem to indeed be quite useful for practical purposes such as numerical analysis.  To illustrate this, let us rewrite $q$-$P(A_1)$ as
\begin{align*}
f=H_1\left(t,\overline{f},g\right), &&g=H_2\left(t,\overline{f},\overline{g}\right).
\end{align*}
We put
\begin{align*}
f_0(t)&=F_{1,0}^{0,+}(\Lambda(t),\mathbf{b})t+\frac{1}{F_1^{0,+}\left(t^{-1}\phi(t)^{-1},\Lambda(t)^{-1},\mathbf{b}^{(1)}\right)},\\
g_0(t)&=G_{1,0}^{0,+}(\Lambda(t),\mathbf{b})t+\frac{1}{G_1^{0,+}\left(t^{-1}\phi(t)^{-1},\Lambda(t)^{-1},\mathbf{b}^{(1)}\right)},
\end{align*}
and define, for $n\in\mathbb{N}$,
\begin{align}
f_{n+1}(t)=H_1\left(t,f_{n}(qt),g_{n+1}(t)\right),&& g_{n+1}(t)=H_2\left(t,f_{n}(qt),g_{n}(qt)\right). \label{Hrec}
\end{align}
Of course $(f(t),g(t))$ is a fixed point of the above mapping and we hope that the $f_n$ and $g_n$ converge at least uniform to $f$ and $g$ on compact sets, which do not contain poles of $f$ and $g$, as $n\rightarrow \infty$. So far, numerics indicate that this is mostly the case, which allows us to obtain accurate numerics with relatively small computational effort. Clearly this motivates to do a stability analysis of the mapping \eqref{Hrec}. We will not pursue this issue further here.

\section{Six special $1$-parameter families of solutions}
\label{section:6special}
As a consequence of Conjecture \ref{conjecture}, we expect the inner summations in \eqref{coefpowerexpansionsol0+} to terminate at $i=0$, i.e. all negative powers of $\phi$ to disappear, when $\Lambda$ is equal to any of the roots of $\mu(\Lambda,\mathbf{b})$. Indeed we have the following result.
\begin{prop}
	\label{prop:holopert}
	Let $k\in\{1,2,3\}$, and $\Lambda_k^\pm$ and $\lambda_k$ be defined as in Section \eqref{section:special6}, where we fix the sign $\pm$ throughout the proposition. Take $\mathbf{b}\in \mathcal{B}$ such that 
	\begin{equation}\label{eq:1Qcon}
	1\notin Q_s:=\{(\lambda_k^{\pm 1}q)^{m-1}q^{n}: (m,n)\in\mathbb{N}^2\setminus\{(1,0)\}.
	\end{equation}
	Then, setting $\Lambda=\Lambda_k^\pm$, the formal solution \eqref{powerexpansionsol0+} of $q$-$P(A_1)$, defined in Theorem \ref{thmgeneralsol0+}, takes the form
	\begin{align}
	f^{0,+}(t,\phi;\Lambda_k^\pm,\mathbf{b})=\sum_{n=1}^\infty{\sum_{i=0}^n{F_{n,i}^{0,+}(\Lambda_k^\pm,\mathbf{b})\phi^it^n}}, && g^{0,+}(t,\phi;\Lambda_k^\pm,\mathbf{b})=\sum_{n=1}^\infty{\sum_{i=0}^n{G_{n,i}^{0,+}(\Lambda_k^\pm,\mathbf{b})\phi^it^n}},\label{eq:fgs}
	\end{align}
	where $\phi$ satisfies $\overline{\phi}=\lambda_k^{\pm 1}\phi$.\\
	Assuming $|q|<1$, $|\lambda_k^{\pm 1}|<|q|^{-1}$ and $\lambda_k^{\pm 1}\notin q^{\mathbb{N}^*}$, condition \eqref{eq:1Qcon} is satisfied and this formal solution, written in terms of the variables $t$ and $\zeta_1=t\phi$,
	\begin{subequations}
	\label{eq:fgzeta1var}
	\begin{align}
	f^{0,+}(t,\zeta_1/t;\Lambda_k^\pm,\mathbf{b})&=\sum_{m=1}^\infty{F_{m,0}^{0,+}(\Lambda_k^\pm,\mathbf{b})t^m}+
	\sum_{i=1}^\infty{\sum_{m=0}^\infty{F_{m+i,i}^{0,+}(\Lambda_k^\pm,\mathbf{b})\zeta_1^it^m}},\\ g^{0,+}(t,\zeta_1/t;\Lambda_k^\pm,\mathbf{b})&=\sum_{m=1}^\infty{G_{m,0}^{0,+}(\Lambda_k^\pm,\mathbf{b})t^m}+
	\sum_{i=1}^\infty{\sum_{m=0}^\infty{G_{m+i,i}^{0,+}(\Lambda_k^\pm,\mathbf{b})\zeta_1^it^m}},
	\end{align}
	\end{subequations}
	converges near $(t,\zeta_1)=(0,0)$.\\
	Furthermore, the pair of isolated power series in \eqref{eq:fgzeta1var}, equals the solution $(f^{(1,k)},g^{(1,k)})$, holomorphic at $t=0$, defined in Proposition \ref{merosolution1}, that is,
	\begin{subequations}
	\label{eq:identifyholo}
	\begin{align}
	f^{0,+}(t,0;\Lambda_k^\pm,\mathbf{b})&=\sum_{m=1}^\infty{F_{m,0}^{0,+}(\Lambda_k^\pm,\mathbf{b})t^m}=f^{(1,k)}(t),\\ g^{0,+}(t,0;\Lambda_k^\pm,\mathbf{b})&=\sum_{m=1}^\infty{G_{m,0}^{0,+}(\Lambda_k^\pm,\mathbf{b})t^m}=g^{(1,k)}(t),
	\end{align}
\end{subequations}
	and in particular these do not depend on the choice of sign $\pm$ in $\Lambda_k^\pm$.	
\end{prop}
\begin{proof}
	For notational simplicity, we discuss the particular case $\Lambda=\Lambda_1^+=-b_1b_2$, noting that the other cases can be dealt with analogously. We assume condition \eqref{eq:1Qcon} with $k=1$ and $\pm=+$, plus the additional conditions
	\begin{align}
	b_1+b_2\neq b_3+b_4, && b_1^{-1}+b_2^{-1}\neq b_3^{-1}+b_4^{-1}, && 1\notin \lambda_1^2 Q_s.\label{eq:additional}
	\end{align}
	Once we have proven the proposition with these additional assumptions, we can easily discard them by analytic continuation using Remark \ref{remark:parameterregular}. Indeed, given the proposition, we find, that condition \eqref{Qcondition} in Theorem \ref{thmgeneralsol0+} can be replaced by $1\notin Q_s$ when $\Lambda=\Lambda_1^+$, as in Remark \ref{remarkrelax}.
	 The idea of the proof is to construct a formal solution $\left(f(t,\phi),g(t,\phi)\right)$ of $q$-$P(A_1)$, which has an expansion in $t$ and $\phi$, exactly of the form \eqref{eq:fgs}, and subsequently use the uniqueness property in Theorem \ref{thmgeneralsol0+} to conclude
	\begin{align}
	f^{0,+}(t,\phi;\Lambda_1^+,\mathbf{b})=f(t,\phi), && g^{0,+}(t,\phi;\Lambda_1^+,\mathbf{b})=g(t,\phi).\label{desiredid}
	\end{align}
	Firstly, by \eqref{eq:1Qcon}, we have $\lambda_1\notin q^\mathbb{Z}$, and using the first two conditions in \eqref{eq:additional}, we construct the solution $(f^{(1,1)},g^{(1,1)})$ of $q$-$P(A_1)$, holomorphic at $t=0$, defined in Proposition \ref{merosolution1}. 
	Next we apply the following change of variables
	\begin{align}
	f(t,\phi)=f^{(1,1)}(t)+\zeta_1\left(1+y_1(t,\zeta_1)\right),&& g(t,\phi)=g^{(1,1)}(t)+\zeta_1\left(-b_1b_2+y_2(t,\zeta_1)\right),\label{eq:changefgy1y2}
	\end{align}
	which allows us to rewrite $q$-$P(A_1)$ as
	\begin{subequations}
		\label{eq:changey1y2}
	\begin{align}
	y_1(q t,q\lambda \zeta_1)&=H_1\left(t,\zeta_1,y_1(t,\zeta_1),y_2(t,\zeta_2)\right),\\
	y_2(q t,q\lambda \zeta_1)&=H_2\left(t,\zeta_1,y_1(t,\zeta_1),y_2(t,\zeta_2)\right),
	\end{align}
\end{subequations}
	for some functions $H_1(t,\zeta_1,y_1,y_2)$ and $H_2(t,\zeta_1,y_1,y_2)$ which are rational in the elements of
	\begin{equation}\label{eq:ratvar}
	\left\{t,\zeta_1,\zeta_2,y_1,y_2,f^{(1,1)}(t),g^{(1,1)}(t),f^{(1,1)}(qt),g^{(1,1)}(qt)\right\}.
	\end{equation}
	We wish to apply the $q$-Briot Bouquet theorem \ref{genbriottheoremsev} with $\mathbf{Y}=(0,0)$, therefore the first condition we have to check is that $H_1$ and $H_2$ are holomorphic at $(t,\zeta_1,y_1,y_2)=(0,0,0,0)$. As $H_1$ and $H_2$ are rational in the elements of \eqref{eq:ratvar}, it is enough to expand $H_1$ and $H_2$ as series in $t,\zeta_1,y_1,y_2$ and check that no negative powers appear. Expanding $H_1$ and $H_2$ in $\zeta_1$, we find for $i=1,2$,
	\begin{equation}
	H_i(t,\zeta_1,y_1,y_2)=h_{-1}^{(i)}(t)\zeta_1^{-1}+h_{0}^{(i)}(t,y_1,y_2)+h_{1}^{(i)}(t,y_1,y_2)\zeta_1+\ldots.
	\end{equation}
	The coefficients $h_{-1}^{(i)}(t)$ are rational in $t$, $f^{(1,1)}(t)$, $g^{(1,1)}(t)$, $f^{(1,1)}(qt)$ and $g^{(1,1)}(qt)$. Formally speaking $h_{-1}^{(1)}(t)=0$ and $h_{-1}^{(2)}(t)=0$ is equivalent to the $q$-$P(A_1)$ equation with $f=f^{(1,1)}(t)$ and $g=g^{(1,1)}(t)$. That is, $h_{-1}^{(1)}(t)$ and $h_{-1}^{(2)}(t)$ are identically zero, precisely because we are perturbing around a solution of $q$-$P(A_1)$. We conclude, for $i=1,2$,
	\begin{equation}
	H_i(t,\zeta_1,y_1,y_2)=h_{0}^{(i)}(t,y_1,y_2)+h_{1}^{(i)}(t,y_1,y_2)\zeta_1+h_{2}^{(i)}(t,y_1,y_2)\zeta_1^2+\ldots.
	\end{equation}	
	In a similar fashion one can calculate that $H_i(t,\zeta_1,y_1,y_2)$ enjoys a power series expansion in the other variables $y_1,y_2$ and $t$, for $i=1,2$. The $y_1$ and $y_2$ cases are rather trivial, but in the $t$ case, we use the fact that we perturb around a solution of $q$-$P(A_1)$, holomorphic at $t=0$, in an essential way.
	 We conclude that $H_1$ and $H_2$ are holomorphic at $(t,\zeta_1,y_1,y_2)=(0,0,0,0)$, and calculate
	\begin{align}
	H_1(0,0,y_1,y_2)&=-\lambda_1^{-1}y_1-\frac{1}{b_1b_2}\left(1+\lambda_1^{-1}\right)y_2,\\
	H_2(0,0,y_1,y_2)&=b_3b_4\left(1+\lambda_1^{-1}\right)y_1+\left(1+\lambda_1^{-1}+\lambda_1^{-2}\right)y_2.	
	\end{align}
	So $H_i(0,0,0,0)=0$ for $i=1,2$, and the Jacobi matrix
	\begin{equation*}
	\begin{pmatrix}
	\frac{\partial H_1}{\partial y_1}(0,0,0,0) & \frac{\partial H_1}{\partial y_2}(0,0,0,0)\\
	\frac{\partial H_1}{\partial y_1}(0,0,0,0) & \frac{\partial H_1}{\partial y_2}(0,0,0,0)
	\end{pmatrix}=
	\begin{pmatrix}
	-\lambda_1^{-1} & -\frac{1}{b_1b_2}\left(1+\lambda_1^{-1}\right)\\
	b_3b_4\left(1+\lambda_1^{-1}\right) & 1+\lambda_1^{-1}+\lambda_1^{-2}
	\end{pmatrix},	
	\end{equation*}
	has eigenvalues $1$ and $\lambda_1^{-2}$.\\
	By \eqref{eq:1Qcon} and the third additional assumption in \eqref{eq:additional}, we can apply the $q$-Briot Bouquet Theorem \ref{genbriottheoremsev}, to obtain an unique power series solution to \eqref{eq:changey1y2} of the form
	\begin{equation*}
	y_i(t,\zeta_1)=\sum_{m,n=0}^\infty{y_{m,n}^{(i)}t^m \zeta_1^n},
	\end{equation*}
	with $y_{0,0}^{(i)}=0$ for $i=1,2$.\\
	Associated via equations \eqref{eq:changefgy1y2}, we have the solution $\left(f(t,\phi),g(t,\phi)\right)$ of $q$-$P(A_1)$ with
	\begin{align*}
	f(t,\phi)=\sum_{n=1}^\infty{\sum_{i=0}^n{f_{n,i}\phi^it^n}}, && g(t,\phi)=\sum_{n=1}^\infty{\sum_{i=0}^n{g_{n,i}\phi^it^n}},
	\end{align*}
	where
	\begin{align}
	f_{1,1}&=1, & f_{1,0}&=f_1^{(1,1)},& f_{n,0}&=f_n^{(1,1)}, & f_{n,i}&=y_{n-i,i-1}^{(1)},\\
	g_{1,1}&=-b_1b_2, & g_{1,0}&=g_1^{(1,1)},& g_{n,0}&=g_n^{(1,1)}, & g_{n,i}&=y_{n-i,i-1}^{(2)},
	\end{align}
	for $1\leq i\leq n$ and $n\in\mathbb{N}_{\geq 2}$.\\
	By the uniqueness property in Theorem \ref{thmgeneralsol0+} we conclude that \eqref{desiredid} must hold. The remaining convergence result follows from the $q$-Briot Bouquet Theorem \ref{genbriottheoremsev}.
\end{proof}
The proof of Proposition \ref{prop:holopert} is not particularly elegant. This lies in the fact that we are dealing with a strongly resonant case in light of the general solution of a $q$-Briot Bouquet type equation. We do not want to delve too far into this issue, but just like to point out that the difficulty comes from the fact that in the case of solutions, holomorphic at $t=0$, the two eigenvalues of the relevant Jacobi matrix are each other's reciprocals, as the proof of Proposition \ref{merosolution0} shows. We avoid this issue by a change of dependent and independent variables, with the cost of dealing with some additional assumptions \eqref{eq:additional}.
\begin{remark}
	Equations \eqref{eq:identifyholo} allow us to analytically continue, for instance the solution $\left(f^{(1,1)}(t),g^{(1,1)}(t)\right)$ defined in Proposition \ref{merosolution1}, to the degenerate parameter cases $b_1+b_2=b_3+b_4$ and $b_1^{-1}+b_2^{-1}=b_3^{-1}+b_4^{-1}$.
\end{remark}
Let us discuss the particular case $\Lambda=\Lambda_1^+=-b_1b_2$ in Proposition \ref{prop:holopert}, in more detail. We choose some parameter values $\mathbf{b}\in \mathcal{B}$, such that $|q|<1$, $|\lambda_1|<|q|^{-1}$ and $\lambda_1\notin q^{\mathbb{N}^*}$. In particular condition \eqref{eq:1Qcon} is satisfied. Strictly speaking, Theorem \ref{thmtruesol} is only applicable if inequalities \eqref{condLambda} are satisfied by $\Lambda(t)=-b_1b_2$, or equivalently,
\begin{equation*}
1<|\lambda_1|<|q|^{-1}.
\end{equation*}
However, the convergence result of \eqref{eq:fgzeta1var} in Proposition \ref{prop:holopert}, allows us to easily extend the results of Theorem \ref{thmtruesol} to the cases $|\lambda_1|<1$ and $|\lambda_1|=1$. Indeed, let us consider the case $|\lambda_1|<1$, and take some analytic function $\phi(t)$ which satisfies $\phi(qt)=\lambda_1\phi(t)$ on a nonempty open set $T\subseteq\mathbb{C}^*$ with $qT=T$. Then there exists an unique meromorphic solution $(f(t),g(t))$ of $q$-$P(A_1)$ on $T$, characterised by
\begin{subequations}
\label{eq:fgspecial}
\begin{align}
f(t)=f^{0,+}(t,\phi(t);-b_1b_2,\mathbf{b})&=\sum_{n=1}^\infty\sum_{i=0}^n{F_{n,i}^{0,+}(-b_1b_2,\mathbf{b})\phi(t)^it^n},\\ g(t)=g^{0,+}(t,\phi(t);-b_1b_2,\mathbf{b})&=\sum_{n=1}^\infty\sum_{i=0}^n{G_{n,i}^{0,+}(-b_1b_2,\mathbf{b})\phi(t)^it^n},
\end{align}
\end{subequations}
for $t$ small in $T$, as the right-hand sides converge uniformly in $t$ on any set $V\subseteq\overline{V}^*\subseteq T$ with $qV=V$, intersected with a disk centered at the origin with radius chosen small enough.\\
In particular, by equations \eqref{eq:identifyholo}, the leading order behaviour of $f(t)$ and $g(t)$ is given by
\begin{align*}
f(t)&=f^{(1,1)}(t)+ \phi(t)t+\mathcal{O}\left(\phi(t)t^2\right),\\
g(t)&=g^{(1,1)}(t)-b_1b_2\phi(t)t+\mathcal{O}\left(\phi(t)t^2\right),
\end{align*}
as $t\rightarrow 0$ in $V$ as above.\\
Of course the choice $\phi(t)\equiv 0$ gives $f(t)=f^{(1,1)}(t)$ and $g(t)=g^{(1,1)}(t)$.
Now let us realise the case $f\ll t$ with $g\asymp t$ in \eqref{eq:conditionalcomb}, by assuming the condition
\begin{equation}\label{eq:cond1}
b_1+b_2=b_3+b_4,
\end{equation}
given in \eqref{eq:condfsmallertgt}.\\
 Indeed the leading term of $f^{(1,1)}(t)$ vanishes, as $f_1^{(1,1)}=0$, and hence we generically have $f\ll t$ and $g\asymp t$ as $t\rightarrow 0$ in $V\subseteq T$ as above. If we also set
\begin{equation}\label{eq:cond2}
b_1^{-1}+b_2^{-1}=b_3^{-1}+b_4^{-1},
\end{equation}
then the leading term of $g^{(1,1)}(t)$ also vanishes, as $g_1^{(1,1)}=0$, and this realises the case $f,g\ll t$ as $t\rightarrow 0$ in \eqref{eq:conditionalcomb}. Note that \eqref{eq:cond1} and \eqref{eq:cond2} imply $b_1=-b_2$ and $b_3=-b_4$, so condition \eqref{eq:condfgsmallert} is trivially satisfied, as expected. To give the reader an appreciation how far the rabbit hole of degenerations goes, let us consider the case
\begin{align}
b_1=-b_2,&& b_3=-b_4, && b_5=-b_6, && b_7=-b_8, && b_2=b_6.
\end{align}
The solution $(f(t),g(t))$ takes the form
\begin{align*}
f(t)=t\phi(t), && g(t)=b_2^2 t\phi(t), && \phi(qt)=\lambda_1\phi(t),&&\lambda_1=\left(\frac{b_2}{b_4}\right)^2, && q=\left(\frac{b_4}{b_8}\right)^2,
\end{align*}
where the parameters $b_2,b_4$ and $b_8$ can be chosen to our pleasure.\\
In particular, let us fix some $b_4,b_8\in\mathbb{C}^*$ with $|b_4|<|b_8|$. Then, for any $m\in\mathbb{N}$, we can choose $b_2\in\mathbb{C}^*$ small enough, such that $|\lambda_1|<|q|^m$, which gives 
\begin{equation*}
f(t),g(t)\ll t^{m+1},
\end{equation*}
as $t\rightarrow 0$ in any open set $V\subseteq \overline{V}^*\subseteq T$ with $qV=V$.\\
Let us return to the generic case \eqref{eq:fgspecial}, application of B\"acklund transformation $\mathcal{T}_1$ and a permutation $\mathbf{b}\mapsto\mathbf{b}^{(1)}$ of the parameters, gives an associated solution $(\tilde{f}(t),\tilde{g}(t))$ of $q$-$P(A_1)(\mathbf{b})$, with
\begin{subequations}
	\label{eq:asympsp}
\begin{align}
\tilde{f}(t)&=f^{(0,1)}(t)\left[1+f^{(0,1)}(t)\widetilde{\phi}(t)\right]^{-1}+\mathcal{O}\left(\widetilde{\phi}(t)t\right),\\
\tilde{g}(t)&=g^{(0,1)}(t)\left[1-b_5^{-1}b_6^{-1}g^{(0,1)}(t)\widetilde{\phi}(t)\right]^{-1}+\mathcal{O}\left(\widetilde{\phi}(t)t\right),
\end{align}
\end{subequations}
as $t\rightarrow 0$ in $V$ as above, where $\widetilde{\phi}(qt)=\lambda\widetilde{\phi}(t)$ with $\lambda=\frac{b_7b_8}{b_5b_6}$, subject to conditions $|q|<1$, $|\lambda|<1$, $\lambda\notin q^{\mathbb{N}^*}$ and, to ensure the validity of the asymptotics \eqref{eq:asympsp},
\begin{align*}
b_5+b_6\neq b_7+b_8 & &\text{and} & & b_5^{-1}+b_6^{-1}\neq b_7^{-1}+b_8^{-1}.
\end{align*}
Setting $\widetilde{\phi}(t)\equiv 0$, gives  the solution $\tilde{f}(t)=f^{(0,1)}(t)$ and $\tilde{g}(t)=g^{(0,1)}(t)$ defined in Proposition \ref{merosolution0}.

\section{Reduction to symmetric $q$-$P(A_1)$}
\label{sectionsym}
There are some natural conditions on the parameters $\mathbf{b}$ wich allow a reduction of $q$-$P(A_1)$ to its symmetric form,
\begin{equation}
\label{qpa1symmetric}
\frac{\left(x\hat{x}-\xi t^2\right)\left(x\uhat{x}-\xi^{-1} t^2\right)}{\left(x\hat{x}-1\right)\left(x\uhat{x}-1\right)}=\frac{\left(x-at\right)\left(x-a^{-1}t\right)\left(x-bt\right)\left(x-b^{-1}t\right)}{\left(x-c\right)\left(x-c^{-1}\right)\left(x-d\right)\left(x-d^{-1}\right)},
\end{equation}
where $a,b,c,d\in\mathbb{C}^*$ are complex parameters and $\xi\in\mathbb{C}^*$ defines the time evolution,
\begin{align*}
\hat{t}=\xi t, && x=x(t),&& \hat{x}=x(\xi t),&& \uhat{x}=x\left(\xi^{-1} t\right).
\end{align*}
We write $\mathbf{b}_s=(a,b,c,d)$ and denote the parameter space of symmetric $q$-$P(A_1)$ by
\begin{equation*}
\mathcal{B}_s=\{(a,b,c,d)\in \mathbb{C}^4|a,b,c,d\neq 0\}.
\end{equation*}
Consider $q$-$P(A_1)$, let $\xi^2=q$ and assume that the parameters $\mathbf{b}$ satisfy
\begin{align}
b_1b_2=\xi,&& b_3b_4=\xi,&& b_5b_6=1,&& b_7b_8=1. \label{parametersymcond}
\end{align}
We set
\begin{align*}
a=b_1\xi^{-\tfrac{1}{2}}, && b=b_3\xi^{-\tfrac{1}{2}}, && c=b_5, && d=b_7.
\end{align*}
If $(f(t),g(t))$ is a solution of $q$-$P(A_1)$ which satisfies
\begin{equation}
f\left( t\right)=g\left(\xi^{-1} t\right), \label{symmetrycondition}
\end{equation}
then it is easy to see that,
\begin{equation}
\label{xdefisym}
x(t)=f\left(\xi^{\tfrac{1}{2}}t\right),
\end{equation}
defines a solution of symmetric $q$-$P(A_1)$.\\
Consider the formal series solution \eqref{powerexpansionsol0+} and assume \eqref{parametersymcond}, then we have
\begin{align}
\hat{\hat{\Lambda}}=\Lambda, && \hat{\hat{\phi}}=\lambda\phi, && \lambda=\left(\frac{\Lambda}{\xi}\right)^2. \label{lambdaphisy}
\end{align}
In order to make sense of condition \eqref{symmetrycondition}, we have to define the time evolution $\hat{\cdot}$ on $\Lambda$ and $\phi$. Inspired by equations \eqref{lambdaphisy}, we set 
\begin{align}
\hat{\Lambda}=\Lambda, && \hat{\phi}=\frac{\Lambda}{\xi}\phi. \label{timeevophilam}
\end{align}
Note that this is indeed consistent with \eqref{lambdaphisy} and condition \eqref{symmetrycondition} becomes
\begin{equation}
f^{(0,+)}\left(t,\phi;\Lambda,\mathbf{b}(\mathbf{b}_s,\xi)\right)=g^{(0,+)}\left(\xi^{-1}t,\uhat{\phi};\Lambda,\mathbf{b}(\mathbf{b}_s,\xi)\right), \label{explsymmetr}
\end{equation}
where we denote
\begin{equation}
\label{bsym}
\mathbf{b}(\mathbf{b}_s,\xi)=\left(a\xi^{\tfrac{1}{2}},a^{-1}\xi^{\tfrac{1}{2}},b\xi^{\tfrac{1}{2}},b^{-1}\xi^{\tfrac{1}{2}},c,c^{-1},d,d^{-1}\right).
\end{equation}
We prove that this condition indeed holds, which implies that the formal series solution \eqref{powerexpansionsol0+} reduces `naturally' to a solution of symmetric $q$-$P(A_1)$. By equation \eqref{symmetryfg1}, we have
\begin{align*}
f^{0,+}(t,\phi;\Lambda,\mathbf{b}(\mathbf{b}_s,\xi))&=g^{0,+}\left(\xi^{-1} t,\xi^{-1}\frac{b_1b_2b_3b_4}{\Lambda}\phi;\frac{\Lambda}{b_5b_6b_7b_8},\mathbf{b}^{(3)}(\mathbf{b}_s,\xi)\right)\\
&=g^{0,+}\left(\xi^{-1} t,\uhat{\phi};\Lambda,\mathbf{b}^{(3)}(\mathbf{b}_s,\xi)\right),
\end{align*}
where
\begin{equation}
\label{bsymwide}
\mathbf{b}^{(3)}(\mathbf{b}_s,\xi)=\left(a^{-1}\xi^{\tfrac{1}{2}},a\xi^{\tfrac{1}{2}},b^{-1}\xi^{\tfrac{1}{2}},b\xi^{\tfrac{1}{2}},c^{-1},c,d^{-1},d\right).
\end{equation}
so it remains to prove
\begin{equation*}
g^{(0,+)}\left(\xi^{-1}t,\uhat{\phi};\Lambda,\mathbf{b}(\mathbf{b}_s,\xi)\right)=g^{0,+}\left(\xi^{-1} t,\uhat{\phi};\Lambda,\mathbf{b}^{(3)}(\mathbf{b}_s,\xi)\right).
\end{equation*}
This identity, however, follows directly from equation \eqref{sigmapermu}, by comparing \eqref{bsym} and \eqref{bsymwide}, where the permutation $\sigma\in Sym(\{1,2,3,4\})\times Sym(\{5,6,7,8\})$ equals
\begin{equation*}
\sigma=(1\ 2)(3\ 4)(5\ 6)(7\ 8).
\end{equation*}
We conclude that, assuming equations \eqref{timeevophilam}, the symmetry condition \eqref{explsymmetr} always holds. Formula \eqref{xdefisym}, however, does not allow for any straightforward interpretation. Luckily we are working with formal variables, so let us for a moment, denote the time evolution $t\mapsto \xi^{\tfrac{1}{2}}t$ by $\tilde{t}$, so $\tilde{t}=\xi^{\tfrac{1}{2}}t$ and in general $\tilde{\tilde{\cdot}}=\hat{\cdot}$.
We simply introduce new formal variables $\Lambda_s$ and $\phi_s$ which are forced to satisfy
\begin{align*}
\widetilde{\Lambda}=\xi \Lambda_s,&& \widetilde{\phi}=\xi^{-\tfrac{1}{2}}\phi_s,
\end{align*}
and hence, by equations \eqref{timeevophilam}, satisfy
\begin{align}
\hat{\Lambda}_s=\Lambda_s, && \hat{\phi}_s=\Lambda_s\phi_s.\label{timeevolutionlambdaphi}
\end{align}
We conclude, using equation \eqref{xdefisym}, that
\begin{equation*}
x^{0,+}\left(t,\phi_s;\Lambda_s,\xi,\mathbf{b}_s\right)=f^{0,+}\left(\xi^{\tfrac{1}{2}}t,\xi^{-\tfrac{1}{2}}\phi_s;\xi\Lambda_s,\mathbf{b}(\mathbf{b}_s,\xi)\right),
\end{equation*}
defines a solution of symmetric $q$-$P(A_1)$.\\
Despite the appearance of square roots of $\xi$ in the above expression, the coefficients in the expansion are rational in $\xi$ and we have the following result.
\begin{thm}
\label{thmsymgeneralsol0+}
There exists an unique formal series solution to the symmetric $q$-$P(A_1)$ equation \eqref{qpa1symmetric}, of the form
\begin{equation}
\label{powersolsym}
x^{0,+}\left(t,\phi_s;\Lambda_s,\xi,\mathbf{b}_s\right)=\sum_{n=1}^\infty{\sum_{i=-\infty}^n{x_{n,i}^{0,+}\left(\Lambda_s,\xi,\mathbf{b}_s\right)t^n \phi_s^i}},
\end{equation}
where $x_{1,1}^{0,+}\left(\Lambda_s,\xi,\mathbf{b}_s\right)=1$ and $\xi$ and $\phi_s$ satisfy \eqref{timeevolutionlambdaphi}.\\
For $n\in\mathbb{N}^*$ and $i\in\mathbb{Z}_{\leq n}$, the coefficient $x_{n,i}^{0,+}\left(\Lambda_s,\xi,\mathbf{b}_s\right)$ is a rational function in its inputs which is regular at points $(\Lambda_s,\xi,\mathbf{b}_s)\in\mathbb{C}^*\times\mathbb{C}^*\times\mathcal{B}_s$ satisfying
\begin{equation}
1\notin Q:=\{q_1^mq_2^n:(m,n)\in\mathbb{N}^2\setminus\{(0,0)\}\}, \label{Qconditionsym}
\end{equation}
where $q_1=q_1(\xi,\Lambda_s)=\xi \Lambda_s$ and $q_2=q_2(\Lambda_s)=\Lambda_s^{-1}$.\\
Furthermore, let $|\xi|<1$ and $\Lambda_s\in L_0^s:=\{x\in\mathbb{C}:1<|x|<|\xi|^{-1}\}$, then condition \eqref{Qconditionsym}
is satisfied and this formal solution, written in terms of the variables 
$\zeta_1=t\phi_s$ and $\zeta_2=\phi_s^{-1}$,
\begin{equation}
\label{xzeta}
x^{0,+}(\zeta_1\zeta_2,\zeta_2^{-1};\Lambda_s,\xi,\mathbf{b}_s)=\sum_{n=1}^\infty{\sum_{m=0}^\infty{x_{n,n-m}^{0,+}(\Lambda_s,\xi,\mathbf{b}_s)\zeta_1^n\zeta_2^m}},
\end{equation}
converges near $(\zeta_1,\zeta_2)=(0,0)$.\\
In fact, these expansions also depend holomorphic on $\Lambda_s$. That is, for any $L\subseteq L_0^s$ open with $\overline{L}\subseteq L_0^s$, there is an open environment $Z\subseteq \mathbb{C}^2$ of $\mathbf{0}$, such that the series \eqref{xzeta} converges uniformly on $Z\times L$, defining holomorphic functions on this set in $\left(\bm{\zeta},\Lambda_s\right)$.
\end{thm}
\begin{proof}
We prove this analogously to Theorem \ref{thmgeneralsol0+}.
\end{proof}
For the formal series solution \eqref{powersolsym}, conjecture \ref{conjecture} implies that the coefficients $x_{n,i}^{0,+}\left(\Lambda_s,\mathbf{b}_s,\xi\right)$ vanish for $i<-n$ and $n\in\mathbb{N}^*$. Indeed, by direct computation we checked this assertion for $n=1,2,3$. As to the case $n=1$, it is easy to see that
\begin{equation*}
x_{1}^{0,+}=x_{1}^{0,+}\left(\phi_s;\Lambda_s,\xi,\mathbf{b}_s\right)=\sum_{i=-\infty}^1{x_{1,i}^{0,+}\left(\Lambda_s,\xi,\mathbf{b}_s\right) \phi_s^i},
\end{equation*}
equals
\begin{equation}
\label{x1exp}
x_{1}^{0,+}=\phi_s-\Lambda_s\frac{a+a^{-1}+b+b^{-1}}{(\Lambda_s-1)^2}+\frac{\Lambda_s(\Lambda_s+ab)(\Lambda_s+\tfrac{a}{b})(\Lambda_s+\tfrac{b}{a})(\Lambda_s+\tfrac{1}{ab})}{(\Lambda_s-1)^4(\Lambda_s+1)^2}\phi_s^{-1},
\end{equation}
which defines a solution to the autonomous QRT-mapping
\begin{equation*}
\left(x_1\overline{x}_1-1\right)\left(x_1\underline{x}_1-1\right)=\left(x_1-a\right)\left(x_1-a^{-1}\right)\left(x_1-b\right)\left(x_1-b^{-1}\right).
\end{equation*}
Similar to Theorem \ref{thmtruesol}, we can use the formal series solution \eqref{powersolsym} to construct true solutions to symmetric $q$-$P(A_1)$. As an example we construct complex power-type series solutions, let $\rho\in\mathbb{C}^*$, set $\Lambda_s=\xi^\rho$ and assume
\begin{equation*}
1<|\Lambda_s|<|\xi|^{-1}.
\end{equation*}
Define $t^\rho$ analogously to \eqref{tpowerrho} on a domain $T$ defined by equation \eqref{Tdefi}, with $q$ replaced by $\xi$, for a $\theta_b\in\mathbb{R}$. Then there is an unique meromorphic solution $x(t)$ of symmetric $q$-$P(A_1)$ on $T$ such that
\begin{equation}
\label{xspecialcaseexp}
x(t)=x^{0,+}\left(t,\phi_0 t^\rho;\xi^\rho,\xi,\mathbf{b}_s\right)=\sum_{n=1}^\infty\sum_{i=-\infty}^n{x_{n,i}^{0,+}\left(\xi^\rho,\xi,\mathbf{b}_s\right)\phi_0^it^{\rho i+n}},
\end{equation}
for $t$ close to $0$.\\
More precisely, there is an $r>0$, such that the expansion on the right-hand side of equation \eqref{xspecialcaseexp} converges uniformly in $t$ on 
\begin{equation*}
T\cap \{t\in\mathbb{C}^*:|t|<r\},
\end{equation*}
and the equation holds on this set.

\subsection{Continuum limit of formal series solution}
\label{sectioncontinuum}
Grammaticos and Ramani \cite{firstqp6} calculated the following continuum limit of the symmetric $q$-$P(A_1)$ equation \eqref{qpa1symmetric} to the sixth Painlev\'e equation \eqref{pVI}.
We set
\begin{align}
&a=-\xi^{\alpha}, & & b=\xi^{\beta},& & c=-\xi^{\gamma}, & &d=\xi^{\delta}, \label{parametervalues}
\end{align}
for some fixed $\alpha,\beta,\gamma,\delta\in\mathbb{C}$. \\
Then, by letting $\xi\rightarrow 1$, symmetric $q$-$P(A_1)$ \eqref{qpa1symmetric} becomes
\begin{align}
x''=&\tfrac{1}{2}\left(\frac{1}{x+1}+\frac{1}{x-1}+\frac{1}{x+z}+\frac{1}{x-t}\right)x'^2 \nonumber \\
&-\left(\frac{1}{t}+\frac{1}{t-1}+\frac{1}{t+1}+\frac{1}{x-t}-\frac{1}{x+t}\right)x' \label{equationcontinuouslimit} \\ 
&+\frac{(x^2-t^2)(x^2-1)}{t^2(t^2-1)}\left(\frac{(\alpha^2-\tfrac{1}{4}) t}{(x+t)^2}-\frac{(\beta^2-\tfrac{1}{4})t}{(x-t)^2}-\frac{\gamma^2}{(x+1)^2}+\frac{\delta^2}{(x-1)^2}\right), \nonumber
\end{align}
which is a non-canonical form of the sixth Painlev\'e equation.\\
Indeed, the change of variables
\begin{align}
t=\frac{1-\sqrt{\zeta}}{1+\sqrt{\zeta}},& & x=\frac{\sqrt{\zeta}-w}{\sqrt{\zeta}+w}, \label{changevar}
\end{align}
leads to the sixth Painlev\'e equation $P_{\text{VI}}$ in canonical form,
\begin{equation}
\label{pVI}
P_{\text{VI}}\quad\begin{cases}
&w''=\tfrac{1}{2}\left(\frac{1}{w}+\frac{1}{w-1}+\frac{1}{w-\zeta}\right)w'^2-\left(\frac{1}{\zeta}+\frac{1}{\zeta-1}+\frac{1}{w-\zeta}\right)w'\\
&+\frac{w(w-1)(w-\zeta)}{2\zeta^2(\zeta-1)^2}\left(\gamma^2-\frac{\delta^2\zeta}{w^2}+\frac{\alpha^2(\zeta-1)}{(w-1)^2}+\frac{(1-\beta^2)\zeta(\zeta-1)}{(w-\zeta)^2}\right).
\end{cases}
\end{equation}
We remark that the parametrisation of $P_{\text{VI}}$ in equation \eqref{pVI} differs slightly from the literature. Indeed, by applying the substitution,
\begin{align*}
\alpha^2\mapsto 2\gamma,&&(1-\beta^2)\mapsto 2\delta,&&\gamma^2\mapsto 2\alpha,&&\delta^2\mapsto -2\beta,
\end{align*}
we obtain a more common form.\\
To calculate the continuum limit of the formal series solution \eqref{powersolsym}, we consider solutions of the form \eqref{xspecialcaseexp}, so $\Lambda_s$ is fixed and we set $\phi(t)=t^\rho$, where the branchcut is taken independent of $\xi$. More precisely, we take any $\rho\in\mathbb{C}^*$ with $-1<\Re{(\rho)}<0$, let $\log(t)$ denote the natural logarithm with respect to a fixed branchcut and define $t^\rho=\exp{(\rho \log{t})}$ as usual. We let $\xi\in (0,1)$ and define $\mathbf{b}_s=\mathbf{b}_s(\xi)$ by equations \eqref{parametervalues} for some fixed $\alpha,\beta,\gamma,\delta\in\mathbb{C}$. We define, as in equation \eqref{xspecialcaseexp},
\begin{equation}
\label{solcontinuumlimit}
x(t;\phi_0,\rho,\xi)=x^{0,+}\left(t,\phi_0 t^\rho;\xi^\rho,\xi,\mathbf{b}_s(\xi)\right)=\sum_{n=1}^\infty\sum_{i=-\infty}^n{x_{n,i}^{0,+}\left(\xi^\rho,\xi,\mathbf{b}_s(\xi)\right)\phi_0^it^{\rho i+n}},
\end{equation}
which converges uniformly in $t$ on an open disc punctured at the origin.\\
Via analytic continuation we extend the domain of $x(t;\phi_0,\rho,\xi)$ to $\mathbb{C}^*$ excluding the branchcut of $\log(t)$.
Using Theorem \ref{genbriottheoremsevuniform}, it is not hard to see that this solution also depends holomorphic on $\xi$ for $\xi\in (0,1)$ and $t$ in an open disc punctured at the origin. However, we are interested in the limit $\xi\uparrow 1$ of solution \eqref{solcontinuumlimit}, but condition \eqref{Qconditionsym} for the existence of the formal series solution \eqref{powersolsym}, is not satisfied at $\xi=1$. Despite the fact that Theorem \ref{thmsymgeneralsol0+} becomes inapplicable in this limit, we do believe solution \eqref{solcontinuumlimit} converges to a true solution of the differential equation \eqref{equationcontinuouslimit} as $\xi\uparrow 1$. Proving this rigorously, probably requires an extension of Theorem \ref{genbriottheoremsevuniform} which incorporates limits of the variable $\mathbf{q}$ as it approaches the boundary of $B_{\text{max}}^m(\mathbf{0},1)$ under some specific assumptions, an interesting direction for future research. Instead, we proceed by heuristically calculating the continuum limit on a formal level. By equation \eqref{x1exp}, we have,
\begin{equation*}
\lim_{\xi\rightarrow 1}{x_{1,0}^{0,+}\left(\xi^\rho,\xi,\mathbf{b}_s(\xi)\right)}=\lim_{\xi\rightarrow 1}{-\xi^\rho\frac{\xi^\beta+\xi^{-\beta}-\xi^{\alpha}-\xi^{-\alpha}}{\left(\xi^\rho-1\right)^2}}=\frac{\alpha^2-\beta^2}{\rho^2},
\end{equation*}
where the second equality is obtained by applying L'H\^opital's rule twice.\\
Similarly, we find
\begin{equation*}
\lim_{\xi\rightarrow 1}{x_{1,-1}^{0,+}\left(\xi^\rho,\xi,\mathbf{b}_s(\xi)\right)}=\frac{\rho(\rho-\alpha-\beta)(\rho-\alpha-\beta)(\rho+\alpha-\beta)(\rho-\alpha+\beta))(\rho+\alpha+\beta)}{4\rho^4},
\end{equation*}
which motivates us to make the following bold move.\\
We assume that the limit,
\begin{equation*}
x_{n,i}^{c}\left(\rho,\mathbf{b}_c\right)=\lim_{\xi\rightarrow 1}{x_{n,i}^{0,+}\left(\xi^\rho,\xi,\mathbf{b}_s(\xi)\right)},
\end{equation*}
exists on a formal level, for all $i\in\mathbb{Z}_{\leq n}$ and $n\in\mathbb{N}^*$, where $\mathbf{b}_c=(\alpha,\beta,\gamma,\delta)$.\\
We hence obtain the following formal solution to the differential equation \eqref{equationcontinuouslimit},
\begin{equation*}
x^c\left(t;\phi_0,\rho,\mathbf{b}_c\right)=\sum_{n=1}^\infty\sum_{i=-\infty}^n{x_{n,i}^{c}\left(\rho,\mathbf{b}_c\right)\phi_0^it^{\rho i+n}},
\end{equation*}
and assuming convergence, its leading order behaviour is given by
\begin{multline*}
x^c\left(t;\phi_0,\rho,\mathbf{b}_c\right)=\phi_0t^{1+\rho}+\frac{\alpha^2-\beta^2}{\rho^2}t+\\
\frac{(\rho-\alpha-\beta)(\rho+\alpha-\beta)(\rho-\alpha+\beta)(\rho+\alpha+\beta)}{4\rho^4}\phi_0^{-1}t^{1-\rho}+\mathcal{O}\left(t^{2-2|\Re(\rho)|}\right),
\end{multline*}
as $t\rightarrow 0$.\\
Applying the change of variables \eqref{changevar} to this solution, we find an associated solution $w^0\left(\zeta;\rho,s,\mathbf{b}_c\right)$ to $P_{\text{VI}}$, where $s=2^{-1-2\rho}\phi_0$, whose leading order behaviour is given by
\begin{multline}
w^0\left(\zeta;s,\rho,\mathbf{b}_c\right)=1-s(1-\zeta)^{1+\rho}-\frac{1}{2}\left(1+\frac{\alpha^2-\beta^2}{\rho^2}\right)(1-\zeta)\\-\frac{(\rho-\alpha-\beta)(\rho+\alpha-\beta)(\rho-\alpha+\beta)(\rho+\alpha+\beta)}{16\rho^4}s^{-1}(1-\zeta)^{1-\rho}+\mathcal{O}\left(\zeta^{2-2|\Re(\rho)|}\right), \label{expanscon}
\end{multline}
as $\zeta\rightarrow 1$.\\
This is exactly the asymptotic expansion around $\zeta=1$ which characterises the solutions obtained by Jimbo \cite{Jimbo} for $P_{\text{VI}}$ \eqref{pVI}. Indeed Guzzetti \cite{Guzetti} showed how to calculate the coefficients in the asymptotic expansions of these solutions, in particular, for any $r\in\mathbb{C}^*$, there is an unique solution $\omega$ of $P_{\text{VI}}$ which satisfies,
\begin{multline}
\omega(\zeta)=-\frac{r}{\rho}\zeta^{1+\rho}+\frac{1}{2}\left(1-\frac{\beta^2-\delta^2}{\rho^2}\right)\zeta\\-\frac{(\rho^2-2\beta\rho+\beta^2-\delta^2)(\rho^2+2\beta\rho+\beta^2-\delta^2)}{16\rho^3}r^{-1}\zeta^{1-\rho}+\mathcal{O}\left(\zeta^{2-2|\Re(\rho)|}\right),\label{expanscon1}
\end{multline}
as $\zeta\rightarrow 0$.\\
Setting $r=-s \rho$ we identify this solution with $w^0$ \eqref{expanscon}, by applying the following B\"acklund transformation of $P_{\text{VI}}$ \eqref{pVI},
\begin{align*}
\widetilde{w}(\zeta)=1-w(1-\zeta),&& \widetilde{\alpha}=\delta, &&\widetilde{\beta}=\beta, && \widetilde{\gamma}=\gamma, && \widetilde{\delta}=\alpha.
\end{align*}
Guzzetti \cite{Guzetti} states that the full expansion of the solution $\omega(t)$ of $P_{\text{VI}}$, characterised by equation \eqref{expanscon1}, is given by
\begin{equation}
\omega(\zeta)=\sum_{n=1}^\infty\sum_{i=-n}^n{\omega_{n,i}\left(\rho,\mathbf{b}_c\right)r^i\zeta^{\rho i+n}}, \label{continuousexp}
\end{equation}
where $\omega_{1,1}=-\tfrac{1}{\rho}$ and the remaining coefficients can be determined uniquely via substitution into $P_{\text{VI}}$ \eqref{pVI} and comparing coefficients.\\
So indeed, the continuous counterpart of Conjecture \ref{conjecture} seems true. That is, there are no terms $t^{\rho i+n}$ in expansion \eqref{continuousexp}, with $i<-n$ and $n\in\mathbb{N}^*$. However, also in the continuous case, this is not a trivial result. Indeed Guzzetti \cite{Guzetti} shows how to determine the coefficients $\omega_{n,i}$ recursively and observes that for at least $n\leq 3$, no terms $t^{\rho i+n}$ with $i<-n$ have to be introduced, but he does not give a proof of this fact for general $n$. Recent work by Lisovyy and collaborators \cite{Bershtein,Gamayun}
expresses the coefficients in the asymptotic expansion near $\zeta=0$ of the $\tau$-function associated with $\omega(\zeta)$, in terms of conformal blocks. One should be able to find corresponding explicit expressions for the coefficiens $\omega_{n,i}$ in the expansion \eqref{continuousexp}, and in particular be able to understand the termination of the inner summations from this perspective. Similarly, it might also be possible to find explicit expressions for the coefficients $F_{n,i}^{0,+}$ and $G_{n,i}^{0,+}$ in the formal series solution \eqref{powerexpansionsol0+}, and in particular prove Conjecture \ref{conjecture} via such an approach.

\section{Concluding Remarks}
In Theorem \ref{thmgeneralsol0+}, we have constructed a formal series solution of $q$-$P(A_1)$, starting from the generic solution \eqref{generalsolaut}, parameterised in terms of formal variables $\Lambda$ and $\phi$, of the leading order autonomous system \eqref{losystemeq}. Furthermore, in Proposition \ref{prop:holopert}, we have proven that the formal series solution takes a special form in the $6$ subcases \eqref{eq:6specialfamilies} of the generic solution  \eqref{generalsolaut}, associated with $6$ special values of the parameter $\Lambda$. It remains to find the formal series solution of $q$-$P(A_1)$ with formal leading order behaviour given by the exceptional logarithmic-type solutions \eqref{eq:logsol} of the autonomous system \eqref{losystemeq}.
We expect these to take the form 
\begin{align}
f(t,\chi)=\sum_{n=1}^\infty{ F_n(\chi)t^n}, && g(t,\chi)=\sum_{n=1}^\infty {G_n(\chi)t^n}, \label{eq:logaexp}
\end{align}
with $F_1(\chi)=F_l^\pm(\chi)$ and $G_1(\chi)=G_l^\pm(\chi)$ as defined in \eqref{eq:logsol}, where each of the coefficients $F_n(\chi)$ and $G_n(\chi)$ is a polynomial of degree $2n$ in $\chi$ for $n\in\mathbb{N}^*$.\\
This, however, does not seem straightforward to prove, and we leave this case open for future research. We would like to note that logarithmic-type solutions have been obtained for $P_{\text{VI}}$, see in particular Guzzetti \cite{guzzettiloga} for an extensive study of these.

By replacing the formal variables $\Lambda$ and $\phi$, in the formal series solution of $q$-$P(A_1)$, with analytic functions, we obtain the generic solutions to $q$-$P(A_1)$ near the origin and infinity.
Guzzetti \cite{guzzettitabulation} gives a tabulation of critical behaviours of solutions of $P_{\text{VI}}$. We have found reflections of all these different critical behaviours in the $q$-$P(A_1)$ case. Indeed, using the terminology in Guzzetti \cite{guzzettitabulation}, we have encountered complex power behaviour in \eqref{specialcaseexp}, (inverse) oscillatory behaviours in \eqref{oscasympdamp} and \eqref{oscasymp}, Taylor expansions in \eqref{eq:holo0} and \eqref{eq:holo1} and logarithmic behaviours in \eqref{eq:logaexp}. As to the remaining case of inverse logarithmic behaviours, these can be constructed by application of B\"acklund transformation $\mathcal{T}_1$ \eqref{backlundt} to \eqref{eq:logaexp}, once rigorously established.

We note that the method described in this study to obtain the formal and true solutions is local in nature and does not explicitly use the integrability of the $q$-$P(A_1)$ equation. We therefore expect it to be applicable to a wide range of equations. As a downside, it does not allow us to obtain global asymptotics, connecting the critical behaviours we found near the origin and infinity. As to $P_{\text{VI}}$, Guzzetti \cite{guzzettitabulation} gives an overview of connection formulae which relate the previously mentioned critical behaviours of solutions, near the different critical points $0$, $1$ and $\infty$, part of which goes back to Jimbo's work \cite{Jimbo}. These formulae can be  established rigorously using the isomonodromic deformation method. The authors are currently analysing the Lax pair of $q$-$P(A_1)$ found by Yamada \cite{yamadalax} to derive similar formulae for $q$-$P(A_1)$. To illustrate this, recall that we have parameterised the critical behaviour of solutions in terms of $(\Lambda,\phi)$ near $t=0$ in \eqref{powerexpansionsol0+} and in terms of $(\Lambda_\infty,\phi_\infty)$ near $t=\infty$ in \eqref{powerexpansionsinf+}. Ideally we would like to find formulae
\begin{align}
\Lambda&=C_{1}^0(\Lambda_\infty,\phi_\infty;\mathbf{b}), & \Lambda_\infty&=C_{1}^\infty(\Lambda,\phi;\mathbf{b}),\\
\phi&=C_{2}^0(\Lambda_\infty,\phi_\infty;\mathbf{b}), & \phi_\infty&=C_{2}^\infty(\Lambda,\phi;\mathbf{b}),
\end{align}
with the $C_{i}^0$ and $C_{i}^\infty$ ($i=1,2$) given explicitly, which relate the behaviour near $t=0$ and $t=\infty$.\\
The authors have parameterised the connection matrix of the spectral part of Yamada's Lax pair in terms of both $\{\Lambda,\phi\}$ and $\{\Lambda_\infty,\phi_\infty\}$, hence these results seem within reach. As, at the time of writing, this analysis is not yet completed, we discuss it elsewhere.

\section{Acknowledgment}
We would like to thank Professor Ohyama for providing us with the reference to Poincare's work on $q$-difference equations \cite{Poincare}.
This research was supported by an Australian Laureate Fellowship \# FL 120100094 and grant \# DP130100967 from the Australian Research Council. Pieter Roffelsen was supported in part by an IPRS scholarship at the University of Sydney.

\appendix

\section{$q$-Briot-Bouquet theorem}
In 1856, Briot and Bouquet \cite{Briot} analysed the existence and uniqueness of ordinary differential equations of a specific type, which are appropriately called Briot-Bouquet equations nowadays. We are interested in $q$-analog equations of this type, wich are systems of $q$-difference equations of the form \eqref{qbrioteq}.
In 1890 Poincar\'e \cite{Poincare} analysed these systems for $m=1$ and proved the so called $q$-Briot-Bouquet theorem, which is the special case of Theorem \ref{genbriottheoremsev} where $m=1$ and $\mathbf{Y}=0$.
In this section we discuss an extension of the classical $q$-Briot-Bouquet Theorem to several independent variables and, more importantly, we prove that the constructed solutions depend regularly on various parameters involved. This is a crucial ingredient in the proof of Theorem \ref{thmtruesol}, where we use the formal series solution defined in Theorem \ref{thmgeneralsol0+}, to construct true solutions of $q$-$P(A_1)$. We use standard multi-index notation, for $n\in\mathbb{N}^*$, for $\alpha=(\alpha_1,\ldots,\alpha_n)\in\mathbb{N}^{n}$, we set
\begin{equation*}
|\alpha|=\alpha_1+\ldots+\alpha_n.
\end{equation*}
For $\alpha,\beta\in \mathbb{N}^{n}$, we write $\alpha\leq\beta$ if and only if for all $1\leq i \leq n$ we have $\alpha_i\leq \beta_i$.
This defines a partial order on $\mathbb{N}^{n}$, and we say $\alpha<\beta$ if and only of $\alpha\leq\beta$ and $\alpha\neq \beta$.\\
If $\mathbf{y}\in\mathbb{C}^n$, we define
\begin{equation*}
\mathbf{y}^{\alpha}=y_1^{\alpha_1}\cdot \ldots\cdot y_n^{\alpha_n}.
\end{equation*}
The following Theorem is an extension of the $q$-Briot-Bouquet Theorem to several independent variables $t_1,\ldots t_m$, each with their own time evolution, $\overline{t}_i=q_i t_i$, where $q_i\in\mathbb{C}^*$ for $1\leq i \leq m$.
\begin{thm}[$q$-Briot-Bouquet theorem (several independent variables)]
\label{genbriottheoremsev}
Let $m,n\in\mathbb{N}^*$ and let us denote
\begin{align}
\mathbf{t}=(t_1,\ldots, t_m), & &\mathbf{q}=(q_1,\ldots, q_m)&& \overline{\mathbf{t}}=(q_1t_1,\ldots,q_m t_m), && \mathbf{y}=(y_1,\ldots,y_n). \label{timeevolution}
\end{align}
Let $H(\mathbf{t},\mathbf{y};\mathbf{q})=\left(H_1(\mathbf{t},\mathbf{y};\mathbf{q}), \ldots, H_n(\mathbf{t},\mathbf{y};\mathbf{q})\right)$ be a vector valued function. Assume there is a $\mathbf{Y}\in\mathbb{C}^n$, such that  $H(\mathbf{t},\mathbf{y})$ is holomorphic at $(\mathbf{t},\mathbf{y})=(\mathbf{0},\mathbf{Y})$ with $H(\mathbf{0},\mathbf{Y})=\mathbf{Y}$.
Suppose the eigenvalues of the Jacobi matrix 
\begin{equation*}
D=\left(\frac{\partial H_j}{\partial y_k}(\mathbf{0},\mathbf{Y})\right)_{1\leq j,k\leq n},
\end{equation*}
are not elements of the set
\begin{equation}
\label{defi Q}
Q:=\{\mathbf{q}^{\alpha} \mid \alpha\in\mathbb{N}^m\setminus\{0\}\}.
\end{equation}
Then the system of $q$-difference equations
\begin{equation}
\label{qbrioteq}
y_j\left(\overline{\mathbf{t}};\mathbf{q}\right)=H_j(\mathbf{t},\mathbf{y}(\mathbf{t});\mathbf{q}) \hspace{1cm} (1\leq j \leq n),
\end{equation}
has an unique power series solution of the form,
\begin{equation*}
y_j(\mathbf{t};\mathbf{q})=Y_j+\sum_{\alpha\in\mathbb{N}^m\setminus\{0\}}{b^{(j)}_{\alpha}(\mathbf{q}) \mathbf{t}^{\alpha}} \hspace{1cm} (1\leq j \leq n).
\end{equation*}
Furthermore, if the eigenvalues of the matrix $D$ are not limit points of the set $Q$, then these power series convergence in an open environment of $\mathbf{t}=\mathbf{0}$.
\end{thm}
\begin{proof}
Several variables can easily be incorporated in the proof by Poincar\'e \cite{Poincare}.
\end{proof}
We would like to extend this result, by proving that the solution $\mathbf{y}$ depends holomorphically on $\mathbf{q}$, as formulated in Theorem \ref{genbriottheoremsevuniform}.  For notational simplicity, we restrict ourselves to the case $\mathbf{Y}=0$.
Iwasaki et al. \cite{Gauss} give an elegant proof of the classical Briot-Bouquet Theorem with several dependent variables, see Proposition 1.1.1 in their book. The proof of Theorem \ref{genbriottheoremsevuniform} is basically an adaptation of their proof, where every estimate is done uniformly in $\mathbf{q}$.
We define the max norm $||\cdot||_{\text{max}}$ on $\mathbb{C}^n$ by
\begin{equation*}
||\mathbf{v}||_{\text{max}}=\max_{1\leq i\leq n} |v_i|,
\end{equation*}
for $\mathbf{v}\in\mathbb{C}^n$, and for matrices $A\in \mathbb{C}^{n\times n}$, we set
\begin{equation*}
||A||_{\text{max}}=\max_{1\leq i,j\leq n} |A_{ij}|.
\end{equation*}
We have the following inequality
\begin{equation}
\label{maxnormineq}
||Av||_{\text{max}}\leq n ||A||_{\text{max}}||v||_{\text{max}},
\end{equation}
for $A\in \mathbb{C}^{n\times n}$ and $\mathbf{v}\in\mathbb{C}^n$.\\
For $\mathbf{v}\in\mathbb{C}^n$ and $R>0$, we define $B_{\text{max}}^n(\mathbf{v},R)$ and $\overline{B}_{\text{max}}^n(\mathbf{v},R)$ to be respectively the open and closed ball of radius $R$ centered at $\mathbf{v}$ in $\mathbb{C}^n$ with respect to the $||\cdot||_{\text{max}}$ norm.

\begin{thm}[$q$-Briot-Bouquet theorem (several independent variables, uniform in $\mathbf{q}$)]
\label{genbriottheoremsevuniform}
Let $m,n\in\mathbb{N}^*$ and denote $\mathbf{t},\mathbf{q},\overline{\mathbf{t}}$ and $\mathbf{y}$ as in \eqref{timeevolution}.
Let $H(\mathbf{t},\mathbf{y};\mathbf{q})=\left(H_1(\mathbf{t},\mathbf{y};\mathbf{q}), \ldots, H_n(\mathbf{t},\mathbf{y};\mathbf{q})\right)$ be a vector valued function.
Assume there is an open set $U\subseteq B_{\text{max}}^m(\mathbf{0},1)\subseteq \mathbb{C}^m$ such that, for every $\mathbf{q}_0\in U$, the function $H(\mathbf{t},\mathbf{y};\mathbf{q})$ is holomorphic at $(\mathbf{t},\mathbf{y};\mathbf{q})=(\mathbf{0},\mathbf{0};\mathbf{q}_0)$ with $H(\mathbf{0},\mathbf{0};\mathbf{q}_0)=0$. For $\mathbf{q}\in U$, let us denote the Jacobian matrix of $H$ with respect to $\mathbf{y}$ at $(\mathbf{t},\mathbf{y})=(\mathbf{0},\mathbf{0})$ by
\begin{equation*}
D(\mathbf{q})=\left(\frac{\partial H_j}{\partial y_k}(\mathbf{0},\mathbf{0};\mathbf{q})\right)_{1\leq j,k\leq n}.
\end{equation*}
We assume that for any $\mathbf{q}\in U$, the eigenvalues of the Jacobi matrix $D(\mathbf{q})$, are not elements of 
\begin{equation}
\label{eq:defiQ}
Q_0:=\{0\}\cup\{\mathbf{q}^{\alpha} \mid \alpha\in\mathbb{N}^m\setminus\{0\}\}.
\end{equation}
Then the $q$-Briot-Bouquet Theorem \ref{genbriottheoremsev} shows, that for every $\mathbf{q}\in U$, the system of $q$-difference equations
\begin{equation}
\label{differential equation}
y_j\left(\overline{\mathbf{t}};\mathbf{q}\right)=H_j(\mathbf{t},\mathbf{y}(\mathbf{t});\mathbf{q}) \hspace{1cm} (1\leq j \leq n),
\end{equation}
has an unique converging power series solution vanishing at $\mathbf{t}=\mathbf{0}$,
\begin{equation}
\label{formalexpansiony}
y_j(\mathbf{t};\mathbf{q})=\sum_{\alpha\in\mathbb{N}^m\setminus\{0\}}{b^{(j)}_{\alpha}(\mathbf{q}) \mathbf{t}^{\alpha}} \hspace{1cm} (1\leq j \leq n).
\end{equation}
For every $\mathbf{q}_0\in U$, for $1\leq j\leq n$, the series \ref{formalexpansiony} converges locally uniformly in $(\mathbf{t},\mathbf{q})$ at $(\mathbf{0},\mathbf{q}_0)$ on $\mathbf{C}^m\times U$.
That is, for every $\mathbf{q}_0\in U$, there are open environments $Z\subseteq \mathbb{C}^m$ and $V\subseteq U$ of $\mathbf{0}$ and $\mathbf{q}_0$ respectively, such that the series \eqref{formalexpansiony} converge uniformly on $Z\times V$ in $(\mathbf{t},\mathbf{q})$, defining analytic functions on this set.
\end{thm}
\begin{proof}
For every $\mathbf{q}\in U$ and $1\leq j\leq n$, since $H_j(\mathbf{t},\mathbf{y};\mathbf{q})$ is holomorphic at $(\mathbf{t},\mathbf{y},\mathbf{q})=(\mathbf{0},\mathbf{0},\mathbf{q})$ with $H_j(\mathbf{0},\mathbf{0};\textbf{q})=0$, we can find a converging power series expansion
\begin{equation}
\label{powerseriesf}
H_j(\mathbf{t},\mathbf{y};\mathbf{q})=\sum_{\alpha\in\mathbb{N}^{m},\beta\in\mathbb{N}^{n}}{C^{(j)}_{(\alpha,\beta)}(\mathbf{q}) \mathbf{t}^{\alpha}\mathbf{y}^{\beta}},
\end{equation}
about $(\mathbf{t},\mathbf{y})=(\mathbf{0},\mathbf{0})$, with $C^{(j)}_{(0,0)}(\mathbf{q})=0$, for $1\leq j\leq n$.\\
The coefficients $C^{(j)}_{(\alpha,\beta)}(\mathbf{q})$ are holomorphic in $\mathbf{q}$ on $U$ for all $\alpha\in\mathbb{N}^{m}$ and $\beta\in\mathbb{N}^{n}$. 
Substituting formal power series expansions \eqref{formalexpansiony} into equation \eqref{differential equation} gives the following recursion for the coefficients $b^{(j)}_{\alpha}(\mathbf{q})$:
\begin{equation}
\label{recursion}
\left(\mathbf{q}^{\alpha} I_n-D(\mathbf{q})\right)\left( \begin{array}{c} b^{(1)}_{\alpha}(\mathbf{q}) \\ b^{(2)}_{\alpha}(\mathbf{q})\\ \vdots \\b^{(n)}_{\alpha}(\mathbf{q})
\end{array} \right)
=\left( \begin{array}{c} M_{\alpha}\left[\left(C^{(1)}_{(\alpha',\beta)}(\mathbf{q})\right)_{(\alpha',\beta)\in L(\alpha)};\left(b^{(1)}_{\alpha'}(\mathbf{q})\right)_{\alpha'<\alpha},\ldots, \left(b^{(n)}_{\alpha'}(\mathbf{q})\right)_{\alpha'<\alpha} \right] \\ M_{\alpha}\left[\left(C^{(2)}_{(\alpha',\beta)}(\mathbf{q})\right)_{(\alpha',\beta)\in L(\alpha)};\left(b^{(1)}_{\alpha'}(\mathbf{q})\right)_{\alpha'<\alpha},\ldots, \left(b^{(n)}_{\alpha'}(\mathbf{q})\right)_{\alpha'<\alpha} \right]\\ \vdots \\ M_{\alpha}\left[\left(C^{(n)}_{(\alpha',\beta)}(\mathbf{q})\right)_{(\alpha',\beta)\in L(\alpha)};\left(b^{(1)}_{\alpha'}(\mathbf{q})\right)_{\alpha'<\alpha},\ldots, \left(b^{(n)}_{\alpha'}(\mathbf{q})\right)_{\alpha'<\alpha} \right]
\end{array} \right),
\end{equation}
for $\alpha\in \mathbb{N}^m\setminus\{0\}$, where the $M_{\alpha}$ are polynomials in their inputs with positive coefficients and the sets $L(\alpha)$ are defined by
\begin{equation*}
L(\alpha)=\{(\alpha',\beta)\in \mathbb{N}^m\times\mathbb{N}^n:\alpha'\leq \alpha, |\beta|\leq |\alpha-\alpha'|\text{ and if $\alpha'=0$, then $|\beta|\geq 2$}\}.
\end{equation*}
As the eigenvalues of $D(\mathbf{q})$ are not elements of $Q_0$ for $\mathbf{q}\in U$, we know that, for every $\alpha\in\mathbb{N}^m\setminus\{0\}$, the matrix $(\mathbf{q}^{\alpha} I_n-D(\mathbf{q}))$ is invertible for $\mathbf{q}\in U$ and, even stronger,
\begin{equation*}
\mathbf{q}\mapsto (\mathbf{q}^{\alpha} I_n-D(\mathbf{q}))^{-1},
\end{equation*}
is a holomorphic matrix-valued function on $U$.\\
Hence this recursion defines unique holomorphic functions $b^{(j)}_{\alpha}(\mathbf{q})$ on $U$ for $1\leq j\leq n$ and $\alpha\in\mathbb{N}^m\setminus\{0\}$. 

Let us take any $\mathbf{q}_0\in U$ and determine $R_U>0$ such that 
\begin{equation*}
B_{\text{max}}^m(\mathbf{q}_0,R_U)\subseteq \overline{B}_{\text{max}}^m(\mathbf{q}_0,R_U)\subseteq U.
\end{equation*}
As $\overline{B}_{\text{max}}^m(\mathbf{q}_0,R_U)\subseteq U$ is compact and the eigenvalues of $D(\mathbf{q})$ are not elements of $Q_0$ for $\mathbf{q}\in U$, we can obtain the following uniform bound on $\overline{B}_{\text{max}}^m(\mathbf{q}_0,R_U)$,
\begin{equation}
\label{uniformbound}
L=\inf_{\mathbf{q}\in \overline{B}_{\text{max}}^m(\mathbf{q}_0,R_U),\alpha\in\mathbb{N}^m\setminus\{0\}}{\left|\det{\left(\mathbf{q}^{\alpha}I_n-D(\mathbf{q})\right)}\right|}>0.
\end{equation}
Hence, for every $\mathbf{q}\in \overline{B}_{\text{max}}^m(\mathbf{q}_0,R_U)$, we have
\begin{align}
\left|\left|\left(\mathbf{q}^{\alpha} I_n-D(\mathbf{q})\right)^{-1}\right|\right|_{\text{max}}&=\left|\left|\frac{\text{adj}\left(\mathbf{q}^{\alpha} I_n-D(\mathbf{q})\right)}{\det\left(\mathbf{q}^{\alpha} I_n-D(\mathbf{q})\right)}\right|\right|_{\text{max}}\\
&=\frac{1}{\left|\det\left(\mathbf{q}^{\alpha} I_n-D(\mathbf{q})\right)\right|}\left|\left|\text{adj}\left(\mathbf{q}^{\alpha} I_n-D(\mathbf{q})\right)\right|\right|_{\text{max}}\\
&\leq \frac{(n-1)!}{L}\left|\left|\mathbf{q}^{\alpha} I_n-D(\mathbf{q})\right|\right|_{\text{max}}^{n-1}\\
&\leq \frac{(n-1)!}{L}\left(\left|\left|\mathbf{q}^{\alpha} I_n\right|\right|_{\text{max}}+\left|\left|(D(\mathbf{q})\right|\right|_{\text{max}}\right)^{n-1}\\
&\leq \frac{(n-1)!}{L}\left(1+\left|\left|(D(\mathbf{q})\right|\right|_{\text{max}}\right)^{n-1},
\end{align}
and, as $\left|\left|(D(\mathbf{q})\right|\right|_{\text{max}}$ is clearly uniformly bounded on the compact set $\overline{B}_{\text{max}}^m(\mathbf{q}_0,R_U)$, we have 
\begin{equation}
\label{Bdefi}
B=\sup_{\mathbf{q}\in \overline{B}_{\text{max}}^m(\mathbf{q}_0,R_U),\alpha\in\mathbb{N}^m\setminus\{0\}}{\left|\left|\left(\mathbf{q}^{\alpha} I_n-D(\mathbf{q})\right)^{-1}\right|\right|_{\text{max}}}<\infty.
\end{equation}
For all $\alpha\in\mathbb{N}^{m}$ and $\beta\in\mathbb{N}^{n}$ and $1\leq j\leq n$, we have a convergent power series expansion
\begin{equation}
\label{cycoef}
C^{(j)}_{(\alpha,\beta)}(\mathbf{q})=\sum_{\gamma\in\mathbb{N}^m}{C^{(j)}_{(\alpha,\beta,\gamma)} (\mathbf{q}-\mathbf{q}_0)^{\gamma}},
\end{equation}
about $\mathbf{q}=\mathbf{q}_0$.\\
Even stronger, for $1\leq j\leq n$, we have a convergent power series expansion,
\begin{equation}
\label{powerseriesH}
H_j(\mathbf{t},\mathbf{y};\mathbf{q})=\sum_{\alpha,\gamma\in\mathbb{N}^{m},\beta\in\mathbb{N}^{n}}{C^{(j)}_{(\alpha,\beta,\gamma)} (\mathbf{q}-\mathbf{q}_0)^{\gamma}\mathbf{t}^{\alpha}\mathbf{y}^{\beta}},
\end{equation}
about $(\mathbf{t},\mathbf{y},\mathbf{q})=(\mathbf{0},\mathbf{0},\mathbf{q}_0)$.\\
For every $1\leq j\leq n$, we determine an $R_j>0$, such that, for all $\mathbf{t},\mathbf{q}\in\mathbb{C}^m$ and $\mathbf{y}\in\mathbb{C}^n$, the series \eqref{powerseriesH} converges if
\begin{align}
||\mathbf{t}||_{\text{max}}<R_j, &&||\mathbf{q}-\mathbf{q}_0||_{\text{max}}<R_j,&& ||\mathbf{y}||_{\text{max}}<R_j.
\end{align}
We set $R_0=\min{(R_U,R_1,\ldots,R_n)}$, take any $0<R<R_0$ and define
\begin{equation*}
M_j=\sup_{\mathbf{q}\in \overline{B}_{\text{max}}^m(\mathbf{q}_0,R),\alpha\in\mathbb{N}^m, \beta\in\mathbb{N}^n}{\left|C_{(\alpha,\beta)}^{(j)}(\mathbf{q})\right|R^{|\alpha+\beta|}}.
\end{equation*}
Clearly the $M_j$ are finite and we set $M_0=\max{(M_1,\ldots, M_n)}$.
We define the function
\begin{equation*}
G(\mathbf{t},Y)=M\left(\left(1-\frac{t_1}{R}\right)^{-1}\cdot\ldots\cdot\left(1-\frac{t_m}{R}\right)^{-1}\left(1-\frac{Y}{R}\right)^{-n}-1-n\frac{Y}{R}\right).
\end{equation*}
Observe that $G$ is holomorphic at $(\mathbf{t},Y)=(\mathbf{0},0)$ with $G(\mathbf{0},0)=0$ and $\frac{\partial G}{\partial Y}(\mathbf{0},0)=0$. Hence $G$ has a convergent power series expansion
\begin{equation*}
G(\mathbf{t},Y)=\sum_{\substack{\alpha\in\mathbb{N}^m,i\in\mathbb{N}\\ (\alpha,i)\neq (0,0),(0,1)}}{C_{(\alpha,i)} \mathbf{t}^{\alpha} Y^i},
\end{equation*}
around $(\mathbf{t},Y)=(\mathbf{0},0)$.\\
Let $\alpha\in\mathbb{N}^m,i\in\mathbb{N}$ with $(\alpha,i)\neq (0,0),(0,1)$, then we have
\begin{equation*}
C_{(\alpha,i)}={n+i-1\choose i}\frac{M_0}{R^{|\alpha|+i}}.
\end{equation*}
Hence, for any $1\leq j\leq n$, for $\alpha\in\mathbb{N}^m,\beta\in\mathbb{N}^n$ such that, if $\alpha=0$, then $|\beta|\geq 2$, we have, by the definition of $M_0$,
\begin{equation*}
\label{Cinequalities}
\left|C^{(j)}_{(\alpha,\beta)}(\mathbf{q})\right|\leq \frac{M_0}{R^{|\alpha+\beta|}}\leq C_{(\alpha,|\beta|)},
\end{equation*}
for $\mathbf{q}\in \overline{B}_{\text{max}}^m(\mathbf{q}_0,R)$.\\
We consider the functional equation
\begin{equation}
\label{functionalequation}
Y(\mathbf{t})=Bn G(\mathbf{t},Y(\mathbf{t})).
\end{equation}
We prove that this equation has an unique solution $Y(\mathbf{t})$ which is holomorphic at $\mathbf{t}=\mathbf{0}$ with $Y(\mathbf{0})=0$.
For this we apply the implicit function theorem to the function $F(\mathbf{t},Y)=B nG(\mathbf{t},Y)-Y$. Observe that $F(\mathbf{0},0)=0$ and
\begin{equation*}
\frac{\partial F}{\partial Y}(\mathbf{0},0)=-1\neq 0.
\end{equation*}
Hence we can apply the implicit function theorem and obtain an unique solution $Y(\mathbf{t})$ of the functional equation \eqref{functionalequation} which is holomorphic at $\mathbf{t}=0$ with $Y(\mathbf{0})=0$. Let the Taylor series expansion of $Y(\mathbf{t})$ at $\mathbf{t}=0$ be given by
\begin{equation}
\label{taylorexpansion}
Y(\mathbf{t})=\sum_{\alpha\in\mathbb{N}^m\setminus \{0\}}{B_{\alpha} \mathbf{t}^{\alpha}}.
\end{equation}
Since $Y$ is a solution of the functional equation \eqref{functionalequation}, the coefficients $B_{\alpha}$ are determined uniquely by the recursion
\begin{equation}
\label{recursion1}
B_{\alpha}=Bn M_{\alpha}\left(\left(C_{(\alpha',|\beta|)}\right)_{(\alpha',\beta)\in L(\alpha)};\left(B_{\alpha'}\right)_{\alpha'<\alpha},\ldots, \left(B_{\alpha'}\right)_{\alpha'<\alpha} \right)
\end{equation}
for $\alpha\in \mathbb{N}^m\setminus \{0\}$, where the polynomials $M_{\alpha}$ are the same as in recursion \eqref{recursion}.\\
We prove the following inequality by complete induction with respect to the partial order $\leq$ on $\mathbb{N}^m$,
\begin{equation}
\label{coefficientssmaller}
\left|b^{(j)}_{\alpha}(\mathbf{q})\right|\leq B_{\alpha}, 
\end{equation}
for every $1\leq j\leq n$ and $\mathbf{q}\in \overline{B}_{\text{max}}^m(\mathbf{q}_0,R)$, for all $\alpha\in \mathbb{N}^m\setminus \{0\}$.\\
Let us fix a $\mathbf{q}\in \overline{B}_{\text{max}}^m(\mathbf{q}_0,R)$, take any $\alpha\in \mathbb{N}^m\setminus \{0\}$, and assume that for all $\alpha'<\alpha$ inequality \eqref{coefficientssmaller} holds for every $1\leq j\leq n$. Then we have, by applying inequality \eqref{maxnormineq} to equation \eqref{recursion},
\begin{align*}
\max_{1\leq i\leq n}{\left|b_{\alpha}^{(i)}(\mathbf{q})\right|}\leq& n \left|\left|\left(\mathbf{q}^{\alpha}I_n-D(\mathbf{q})\right)^{-1}\right|\right|_{\text{max}}\cdot\\
&\max_{1\leq i\leq n}{\left|M_{\alpha}\left[\left(C^{(i)}_{(\alpha',\beta)}(\mathbf{q})\right)_{(\alpha',\beta)\in L(\alpha)};\left(b^{(1)}_{\alpha'}(\mathbf{q})\right)_{\alpha'<\alpha},\ldots, \left(b^{(n)}_{\alpha'}(\mathbf{q})\right)_{\alpha'<\alpha} \right]\right|}\\
\leq &n B \max_{1\leq i\leq n}{M_{\alpha}\left[\left(\left|C^{(i)}_{(\alpha',\beta)}(\mathbf{q})\right|\right)_{(\alpha',\beta)\in L(\alpha)};\left(\left|b^{(1)}_{\alpha'}(\mathbf{q})\right|\right)_{\alpha'<\alpha},\ldots, \left(\left|b^{(n)}_{\alpha'}(\mathbf{q})\right|\right)_{\alpha'<\alpha} \right]}\\
\leq &n B M_{\alpha}\left(\left(C_{(\alpha',|\beta|)}\right)_{(\alpha',\beta)\in L(\alpha)};\left(B_{\alpha'}\right)_{\alpha'<\alpha},\ldots, \left(B_{\alpha'}\right)_{\alpha'<\alpha} \right)\\
=&B_{\alpha},
\end{align*}
where, in the second inequality we used the definition of $B$ \eqref{Bdefi} and the fact that the polynomials $M_{\alpha}$ have positive coefficients, in the third inequality we use the induction hypothesis \eqref{coefficientssmaller}, and in the last equality we used equation \eqref{recursion1}. \\
By complete induction, we conclude inequality \eqref{coefficientssmaller} holds for all $1\leq j\leq n$ and $\mathbf{q}\in \overline{B}_{\text{max}}^m(\mathbf{q}_0,R)$, for every $\alpha\in \mathbb{N}^m\setminus \{0\}$.
Determine $\rho_0>0$, such that the Taylor expansion \eqref{taylorexpansion} converges if $||\mathbf{t}||_{\text{max}}<\rho_0$. Take any $0<\rho<\rho_0$, then we have
\begin{equation}
\label{boundrho}
\sum_{\alpha\in\mathbb{N}^m\setminus \{0\}}{B_{\alpha} \rho^{|\alpha|}}<\infty.
\end{equation}
Take a $1\leq j\leq n$ and define, for $\alpha\in\mathbb{N}^m\setminus \{0\}$, the function
\begin{equation*}
Y_{\alpha}^{(j)}(\mathbf{t},\mathbf{q})=b_{\alpha}^{(j)}(\mathbf{q})\mathbf{t}^{\alpha},
\end{equation*}
which is holomorphic on $U\times \mathbb{C}^m$ and hence, also holomorphic on the compact set
\begin{equation}\label{eq:Sdefi}
S=\overline{B}_{\text{max}}^m(\mathbf{0},\rho)\times \overline{B}_{\text{max}}^m(\mathbf{q}_0,R).
\end{equation}
Let $||\cdot||_\infty^S$ denote the supremum norm on $S$, then we have, by inequality \eqref{coefficientssmaller},
\begin{equation*}
||Y_{\alpha}^{(j)}||_\infty^S\leq B_{\alpha} \rho^{|\alpha|},
\end{equation*}
and therefore, by equation \eqref{boundrho},
\begin{equation*}
\sum_{\alpha\in\mathbb{N}^m\setminus \{0\}}{||Y_{\alpha}^{(j)}||_\infty^S}<\infty.
\end{equation*}
We conclude that
\begin{equation}\label{eq:yjpower}
y_j(\mathbf{t};\mathbf{q})=\sum_{\alpha\in\mathbb{N}^m\setminus \{0\}}{Y_{\alpha}^{(j)}(\mathbf{t},\mathbf{q})}=\sum_{\alpha\in\mathbb{N}^m\setminus \{0\}}{b_{\alpha}^{(j)}(\mathbf{q})\mathbf{t}^{\alpha}},
\end{equation}
converges uniformly on $S$, defining a complex function holomorphic on the interior of $S$.
\end{proof}

\begin{remark}
\label{remarkincorporatepar}
It is straightforward to extend Theorem \ref{genbriottheoremsevuniform} to include parameters. That is, we write $\mathbf{b}=(b_1,\dots,b_s)$ for some $s\in\mathbb{N}^*$, set $H=H(\mathbf{t},\mathbf{y};\mathbf{q},\mathbf{b})$ and assume that the conditions in Theorem \ref{genbriottheoremsevuniform} hold for all $\mathbf{b}$ in some fixed open set $V\subseteq \mathbb{C}^s$. Then the obtained series $y(\mathbf{t};\mathbf{q},\mathbf{b})$ converge locally uniformly in $(\mathbf{t},\mathbf{q},\mathbf{b})$ at $(\mathbf{t},\mathbf{q},\mathbf{b})=(\mathbf{0},\mathbf{q}_0,\mathbf{b}_0)$ for $(\mathbf{q}_0,\mathbf{b}_0)\in U\times V$.
\end{remark}

\section{The QRT mapping}
In this section we discuss the QRT mapping, first introduced in Quispel, Roberts and Thompson \cite{Quispel}. We denote
\begin{equation*}
X=\begin{pmatrix}
x^2\\ x \\ 1
\end{pmatrix},
\end{equation*}
and take square matrices
\begin{equation*}
A_i=\begin{pmatrix}
\alpha_i & \beta_i & \gamma_i\\
\delta_i & \epsilon_i & \zeta_i\\
\kappa_i & \lambda_i & \mu_i
\end{pmatrix},
\end{equation*}
where $i=1,2$.\\
The QRT mapping is the $18$-parameter family of mappings given by
\begin{align}
\overline{x}=\frac{f_1(y)-xf_2(y)}{f_2(y)-xf_3(y)}, & &\overline{y}=\frac{g_1(\overline{x})-yg_2(\overline{x})}{g_2(\overline{x})-yg_3(\overline{x})}, \label{QRT}
\end{align}
with
\begin{align*}
f(x)=(A_0 X)\times (A_1 X), & & g(x)=\left(A_0^T X\right)\times \left(A_1^T X\right).
\end{align*}
Such a mapping has an invariant given by
\begin{equation*}
I(x,y)=\frac{\alpha_0 x^2y^2+  \beta_0x^2 y + \gamma_0 x^2+
\delta_0 x y^2 + \epsilon_0 x y + \zeta_0 x+
\kappa_0y^2 + \lambda_0 y + \mu_0}{\alpha_1 x^2y^2+  \beta_1x^2 y + \gamma_1 x^2+
\delta_1 x y^2 + \epsilon_1 x y + \zeta_1 x+
\kappa_1y^2 + \lambda_1 y + \mu_1},
\end{equation*}
that is, if $(x,y)$ satisfies \eqref{QRT}, then
\begin{equation}
\label{invarianteq}
I\left(\overline{x},\overline{y}\right)=I(\overline{x},y)=I(x,y).
\end{equation}
Let us set $P=I(x,y)$ and
\begin{align*}
A=\alpha_0-\alpha_1 P,& & B=\beta_0-\beta_1 P, & & D=\delta_0-\delta_1 P,&& G=\gamma_0-\gamma_1 P,&& E=\epsilon_0-\epsilon_1 P,\\
K=\kappa_0-\kappa_1 P,&& L=\lambda_0-\lambda_1 P, & & Z=\zeta_0-\zeta_1 P, && U=\mu_0-\mu_1 P,
\end{align*}
then, by equation \eqref{invarianteq}, we have
\begin{align}
Ax^2y^2+Bx^2y+Dxy^2+Gx^2+Exy+Ky^2+Zx+Ly+U&=0,\label{inv1}\\
A\overline{x}^2y^2+B\overline{x}^2y+D\overline{x}y^2+G\overline{x}^2+E\overline{x}y+Ky^2+Z\overline{x}+Ly+U&=0,\label{inv2}\\
A\overline{x}^2\overline{y}^2+B\overline{x}^2\overline{y}+D\overline{x}\overline{y}^2+G\overline{x}^2+E\overline{x}\overline{y}+K\overline{y}^2+Z\overline{x}+L\overline{y}+U&=0.\label{inv3}
\end{align}
Subtracting equation \eqref{inv1} from \eqref{inv2} and equation \eqref{inv2} from \eqref{inv3} we obtain respectively,
\begin{align*}
&\left(\overline{x}-x\right)\left(A(\overline{x}+x)y^2+B\left(\overline{x}+x\right)y+Dy^2+G\left(\overline{x}+x\right)+Ey+Z\right)=0,\\
&\left(\overline{y}-y\right)\left(A\overline{x}^2\left(\overline{y}+y\right)+B\overline{x}^2+D\overline{x}\left(\overline{y}+y\right)+E\overline{x}+K\left(\overline{y}+y\right)+L\right)=0.
\end{align*}
Excluding the cases $\overline{x}=x$ and $\overline{y}=y$, we obtain
\begin{align}
\overline{x}=-x-\frac{Dy^2+Ey+Z}{Ay^2+By+G}, & & \overline{y}=-y-\frac{B\overline{x}^2+E\overline{x}+L}{A\overline{x}^2+D\overline{x}+K}.\label{mcmillan}
\end{align}
If the various parameters $A,B,\ldots,L,Z$ in this system are plain complex numbers, this has been called the asymmetric McMillan map in Iatrou and Roberts \cite{Iatrou}, as it is indeed an asymmetric extension of the classical McMillan map \cite{McMillan}.
\begin{lem}
\label{lemqrtmclillan}
If $x$ and $y$ satisfy equations \eqref{inv1} and \eqref{mcmillan}, then they form a solution of the QRT mapping \eqref{QRT}.
\end{lem}
\begin{proof}
Iatrou and Roberts \cite{Iatrou} prove that the mapping \eqref{mcmillan} leaves \eqref{inv1} invariant, in fact, they show that if $x$ and $y$ satisfy equations \eqref{inv1} and \eqref{mcmillan}, then equations \eqref{inv2} and \eqref{inv3} hold as well. Eliminating $P$ from equation \eqref{inv1} and the first equation in \eqref{mcmillan} gives the time-evolution of $x$ in \eqref{QRT}. Eliminating $P$ from equation \eqref{inv2} and the second equation in \eqref{mcmillan} gives the time-evolution of $y$ in \eqref{QRT}.
\end{proof}
\subsection{QRT mappings of linear type}
\label{qrtsolve}
If $A=B=D=0$, then the autonomous system \eqref{mcmillan} is linear. We therefore consider the following special case of the QRT mapping \eqref{QRT}, where
\begin{equation}
\alpha_0=\alpha_1=\beta_0=\beta_1=\delta_0=\delta_1=0. \label{assumptions}
\end{equation}
Equations \eqref{mcmillan} become the following system of linear equations
\begin{align}
\overline{x}+x+\frac{E}{G}y=-\frac{Z}{G}, && \overline{y}+y+\frac{E}{K}\overline{x}=-\frac{L}{K}.\label{mcmillanlinear}
\end{align}
Solving this is straightforward, we first look for an equilibrium solution $(x_{eq},y_{eq})$, that is, a solution invariant under the time evolution, so
\begin{align*}
2x_{eq}+\frac{E}{G}y_{eq}=-\frac{Z}{G}, & & 2y_{eq}+\frac{E}{K}x_{eq}=-\frac{L}{K},
\end{align*}
which gives
\begin{align}
x_{eq}=\frac{2KZ-EL}{E^2-4GK}, &&y_{eq}=-\frac{2GL-EZ}{E^2-4GK}. \label{xyeq}
\end{align}
It remains to find the general solution to the homogeneous part of system \eqref{mcmillanlinear}, which is
\begin{equation}
\label{homoeq}
\begin{pmatrix}
\overline{x}\\
\overline{y}
\end{pmatrix}=
\begin{pmatrix}
-1 & -\frac{E}{G}\\
\frac{E}{K}& \frac{E^2}{KG}-1
\end{pmatrix}
\begin{pmatrix}
x\\
y
\end{pmatrix}.
\end{equation}
Let us denote
\begin{equation}
\label{Mdefi}
M=\begin{pmatrix}
-1 & -\frac{E}{G}\\
\frac{E}{K}& \frac{E^2}{KG}-1
\end{pmatrix},
\end{equation}
In order to diagonalise this matrix, we consider the associated characteristic equation
\begin{equation}
\label{charactisticpol}
|M-\lambda I|=\lambda^2+\left(2-\frac{E^2}{KG}\right) \lambda+1=0,
\end{equation}
which generically does not have a solution in $\mathbb{C}(P)$.\\
To overcome this limitation we could set $P=c_2 \Lambda+c_1+c_0/\Lambda$ where $\overline{\Lambda}=\Lambda$ for some well chosen $c_0,c_1,c_2$, to guarantee that equation \eqref{charactisticpol} has a root in $\mathbb{C}(\Lambda)$. However the calculations quickly get out of hand, so we illustrate this process by example in \eqref{subsectionlo}, and hope the procedure to solve the system \eqref{mcmillanlinear} in general becomes clear. Once the general solution to this system is found, we substitute it into equation \eqref{inv1}, which forces us to fix the value of one free parameter as is done in equation \eqref{mudefi}. Then, by Lemma \ref{lemqrtmclillan}, we obtain the generic solution of the QRT mapping subject to conditions \eqref{assumptions}. Note that we assumed $E^2-4GK\neq 0$ to obtain the equilibrium solution \eqref{xyeq}. The case $E^2-4GK=0$ is delicate and requires a separate analysis. We discuss such a case in Section \ref{section:loga}.

\section{B\"acklund transformations of $q$-$P(A_1)$}
\label{appback}
The $q$-$P(A_1)$ equation has various B\"acklund transformations, which relate solutions of the equation for possibly different parameter values. The full affine Weyl group symmetry of $q$-$P(A_1)$ is of $E_7^{(1)}$ type, however we only use $4$ such transformations in this study, given by
\begin{equation}
\label{backlundt}
\mathcal{T}_k\left(f,g,\mathbf{b}\right)=\left(f^{(k)},g^{(k)},\mathbf{b}^{(k)}\right),
\end{equation}
for $k\in \{1,2,3,4\}$, where
\begin{align*}
f^{(1)}(t)&=\frac{t}{f(t)},& f^{(2)}(t)&=tg\left(\frac{1}{t}\right),& f^{(3)}(t)&=g\left(q^{-\tfrac{1}{2}}t\right),& f^{(4)}(t)&= \frac{1}{g(\frac{1}{t})},\\
g^{(1)}(t)&=\frac{t}{g(t)},& g^{(2)}(t)&=tf\left(\frac{1}{t}\right), & g^{(3)}(t)&=f\left(q^{\tfrac{1}{2}}t\right),& g^{(4)}(t)&=\frac{1}{f(\frac{1}{t})},
\end{align*}
with for $1\leq i\leq 4$ and $5\leq j \leq 8$,
\begin{align*}
b_i^{(1)}&=b_{i+4}^{-1},& b_i^{(2)}&=b_{i+4}^{-1},& b_i^{(3)}&=q^{\tfrac{1}{2}}b_i^{-1},& b_i^{(4)}&=b_i,\\
b_j^{(1)}&=b_{j-4}^{-1},& b_j^{(2)}&=b_{j-4}^{-1},& b_j^{(3)}&=b_j^{-1},& b_j^{(4)}&=b_j.
\end{align*}
Note that each of these transformations $\mathcal{T}_k$ leaves $q(\mathbf{b})$ invariant and indeed maps solutions of $q$-$P(A_1)(\mathbf{b})$ to solutions of $q$-$P(A_1)(\mathbf{b}^{(k)})$ for $k\in \{1,2,3,4\}$. We remark that these transformations are not independent, for instance $\mathcal{T}_1 \mathcal{T}_2=\mathcal{T}_2\mathcal{T}_1=\mathcal{T}_4$.

\end{document}